\documentclass[a4paper,11pt,reqno]{amsart} 
\usepackage{amsmath}
\usepackage{amsfonts}
\usepackage[T1]{fontenc}  
\usepackage{lmodern}%
\usepackage{microtype}%
\usepackage[utf8]{inputenc}
\usepackage[mathscr]{euscript}
\usepackage{amssymb,latexsym}
\usepackage{amsthm}
\usepackage{color}
\usepackage{graphicx}
\usepackage[hidelinks]{hyperref}
\usepackage[font=scriptsize]{caption}
\usepackage{subcaption}
\usepackage{cite}
\usepackage{bbm}
\usepackage{amsmath}
\usepackage{tikz}
\usetikzlibrary{patterns}

\newcommand{\Ab}{\mathbf{A}}

\usepackage[lmargin=3 cm,rmargin=3 cm,tmargin=3.5cm,bmargin=2.5cm,paper=a4paper]{geometry}
\DeclareMathOperator{\curl}{curl}
 \DeclareMathOperator{\dist}{dist} 
\DeclareMathOperator{\supp}{supp} \DeclareMathOperator{\dom}{\mathrm{Dom}}
\DeclareMathOperator{\csch}{\mathrm{csch}}
\newtheorem{thm}{Theorem}[section]
\newtheorem{pro}[thm]{Proposition}

\newtheorem{lem}[thm]{Lemma}

\newtheorem{theorem}[thm]{Theorem}
\newtheorem{assumption}[thm]{Assumption}

\newtheorem{lemma}[thm]{Lemma}
\newtheorem{proposition}[thm]{Proposition}
\newtheorem{corollary}[thm]{Corollary}

\theoremstyle{remark}
\newtheorem{rem}[thm]{Remark}

\newcommand{\R}{\mathbb{R}}

\newcommand{\Om}{\Omega}

\numberwithin{equation}{section}
\title{A 3D-Schr\"odinger operator under magnetic steps with semiclassical applications}
\author[W. Assaad]{Wafaa Assaad}
\address{Lebanese International University, Faculty of Arts and Sciences, Beirut, Lebanon}
\email{wafaa.assaad@liu.edu.lb}
\author[E. L. Giacomelli]{Emanuela L. Giacomelli}
\address{LMU Munich, Department of Mathematics, Theresienstr. 39, 80333 Munich, Germany}
\email{emanuela.giacomelli@math.lmu.de}
\usepackage[deletedmarkup=xout]{changes}														
	\definechangesauthor[name={Emanuela}, color={blue}]{ema}

	\definechangesauthor[name={Wafaa}, color={green}]{wafaa}

\begin{document}
\maketitle
\begin{abstract}
	We define a Schr\"odinger operator on the half-space  with  a discontinuous magnetic field having a piecewise-constant strength and a uniform direction. Motivated by applications in the theory of superconductivity, we study the infimum of the spectrum of the operator. We  give sufficient conditions on the strength and the direction of the magnetic field such that the aforementioned infimum is an eigenvalue of a reduced model operator on the half-plane. We use the Schr\"odinger operator on the half-space to study a new semiclassical problem in bounded domains of the space, considering a magnetic Neumann Laplacian with a piecewise-constant magnetic field. We then  make precise the localization of the semiclassical ground state near specific points at the discontinuity jump of the magnetic field. 
\end{abstract}

\section{Introduction}\label{sec:int}
We consider a Schr\"odinger operator defined on the half-space and having a magnetic field with a piecewise-constant strength and a uniform direction. Such operator is interesting to be considered in new situations in the theory of superconductivity as we will describe later. 
We set the half-space to be $\mathbb{R}^3_+ := \{\mathbf{x}\in \mathbb{R}^3\, \vert\, \mathbf{x} = (x_1, x_2, x_3), \, \, x_2>0\}$ and we split it in two regions in which the strength of the magnetic field is different as follows.  Let $\alpha \in (0,\pi)$, using spherical coordinates, we define the domains $\mathcal{D}_\alpha^1$ and  $\mathcal{D}_\alpha^2$ of $\mathbb{R}_+^3$:
 \begin{equation}\label{eq: def D1alpha}
 	\mathcal D^1_\alpha = \left\{\mathbf{x}\in\mathbb{R}^3\, \vert\, \mathbf{x} =  \rho(\cos\theta\sin\phi,\sin\theta\sin\phi,\cos\phi),\, \rho\in(0,\infty),\,0<\theta<\alpha,\,\phi\in(0,\pi)\right\},
 \end{equation}
 \begin{equation}\label{eq: def D12alpha}
 	\mathcal D^2_\alpha = \left\{\mathbf{x}\in\mathbb{R}^3\, \vert\, \mathbf{x} =  \rho(\cos\theta\sin\phi,\sin\theta\sin\phi,\cos\phi),\, \rho\in(0,\infty),\,\alpha<\theta<\pi,\,\phi\in(0,\pi)\right\}.
 \end{equation}
Let $a\in[-1,1)\setminus\{0\}$ and\footnote{By symmetry considerations, we restrict the study to the case where $ \gamma \in[0,\pi/2]$.} $\gamma\in [0,\pi/2]$,  we introduce the following magnetic field\footnote {The choice of the discontinuous magnetic field $\mathbf{B}_{\alpha,\gamma,a}$  as in \eqref{eq:Bs} is motivated by getting the operator $\mathcal{L}_{\alpha,\gamma,a}$ as the right tangent operator in localizations problems on bounded domains, considered later in the paper  (see the applications in Section \ref{sec: motivation}).}  in $\mathbb{R}_+^3$
\begin{eqnarray}\label{eq:Bs}\mathbf{B}_{ \alpha,\gamma,a} &=& (\cos\alpha\sin\gamma,\sin\alpha\sin\gamma,\cos\gamma)\left(\mathbbm{1}_{\mathcal D^1_\alpha}+a\mathbbm{1}_{\mathcal D^2_\alpha}\right)
\\
&=& (\cos\alpha\sin\gamma,\sin\alpha\sin\gamma,\cos\gamma)\mathsf s_{ \alpha,a}.\nonumber
\end{eqnarray}
Here (and in the sequel) $\mathbbm{1}_{\sharp}$ denotes the characteristic function corresponding to the set $\sharp$ (in this case $\sharp = \mathcal{D}_\alpha^1, \mathcal{D}_\alpha^2$). The function $\mathsf{s}_{\alpha, a}$ represents the strength of the magnetic field (see Figure \ref{fig:1}). The choice of the values $a$ in $ [-1,1)\setminus \{0\}$ will be discussed later (see~Remark~\ref{rem: a=0}). 

We consider the magnetic Neumann realization of the following self-adjoint operator on $\R^3_+$
	\begin{equation}\label{eq:La+}
	\mathcal L_{ \alpha,\gamma,a} =-(\nabla-i\Ab_{ \alpha,\gamma,a} )^2,
	\end{equation}
 where $\Ab_{ \alpha,\gamma,a} \in H^1_{\mathrm{loc}}(\R_+^3,\R^3)$ is a magnetic potential  such that $\mathrm{curl}\mathbf{A}_{\alpha,\gamma, a} = \mathbf{B}_{\alpha, \gamma, a}$. A choice of the magnetic potential $\Ab_{ \alpha,\gamma,a}$ is fixed in~\eqref{eq:Ag}. 
	The domain of the operator $\mathcal L_{ \alpha,\gamma,a} $ is
	\begin{figure}
		\centering
		\begin{tikzpicture}[scale= 0.5]
		\draw[->] (0,0) to (0,5.5);
		\draw[->] (0,0) to (5.5,0);
		\draw[->] (0,0) to (-2.5,-3);
		\draw[->] (0,0) to (2,4);
		\draw[->] (3.6,3.5) to (2.6,1.5); 
		\draw[->] (-2.6,1.5) to (-1.4,3.5); 
  		\node at (3.5,4.1) {\small{$\mathbf{B}_{\alpha,\gamma,a}$}};
  		\node at (-1.5,4.1) {\small{$\mathbf{B}_{\alpha,\gamma,a}$}};
		\draw (0,0) to (4,-3);
		\draw[dotted] (2,4) to (2,-1.5);
		\draw[dotted] (4,-3) to (4,5);
		\node at (0, -0.5) {\tiny{$\alpha$}};
		\node at (0.2,1) {\tiny{$\gamma$}};
		\node at (-3.7,5.7) {\tiny{$\mathcal{D}^1_\alpha$}};
		\node at (-3.8,5) {\tiny{$\mathsf{s}_{\alpha,a} = 1$}};
		\node at (5.5,5.7) {\tiny{$\mathcal{D}^2_\alpha$}};
		\node at (5.5,5) {\tiny{$\mathsf{s}_{\alpha,a} = a$}};
		\node at (0,5.7) {\tiny{$x_3$}};
		\node at (5.8, 0) {\tiny{$x_2$}};
		\node at (-2.65, -3.2) {\tiny{$x_1$}};
		\node at (3.4, -1.3) {\small{$P_\alpha$}};
		\draw[<-]([shift=(30:1cm)]-0.61,0) arc (60:90:0.5cm);
		\draw[->]([shift=(30:1cm)]-1,-0.65) arc (190:361:0.15cm);
		\node at (-0.1,7) {\small{$\boxed{\textcolor{white}{\sum}\hspace{-0.4cm}{ a \in[-1,0)}\hspace{0.15cm}}$}};
		\end{tikzpicture}
		\hspace{0.6cm}
		\begin{tikzpicture}[scale= 0.5]
		\draw[->] (0,0) to (0,5.5);
		\draw[->] (0,0) to (5.5,0);
		\draw[->] (0,0) to (-2.5,-3);
		\draw[->] (0,0) to (2,4);
		\draw[<-] (3.6,3.5) to (2.6,1.5); 
		\draw[->] (-2.6,1.5) to (-1.4,3.5); 
  		\node at (3.5,4.1) {\small{$\mathbf{B}_{\alpha,\gamma,a}$}};
  		\node at (-1.5,4.1) {\small{$\mathbf{B}_{\alpha,\gamma,a}$}};
		\draw (0,0) to (4,-3);
		\draw[dotted] (2,4) to (2,-1.5);
		\draw[dotted] (4,-3) to (4,5);
		\node at (0, -0.5) {\tiny{$\alpha$}};
		\node at (0.2,1) {\tiny{$\gamma$}};
		\node at (-3.7,5.7) {\tiny{$\mathcal{D}^1_\alpha$}};
		\node at (-3.8,5) {\tiny{$\mathsf{s}_{\alpha,a} = 1$}};
		\node at (5.5,5.7) {\tiny{$\mathcal{D}^2_\alpha$}};
		\node at (5.5,5) {\tiny{$\mathsf{s}_{\alpha,a} = a$}};
		\node at (0,5.7) {\tiny{$x_3$}};
		\node at (5.8, 0) {\tiny{$x_2$}};
		\node at (-2.65, -3.2) {\tiny{$x_1$}};
		\node at (3.4, -1.3) {\small{$P_\alpha$}};
		\draw[<-]([shift=(30:1cm)]-0.61,0) arc (60:90:0.5cm);
		\draw[->]([shift=(30:1cm)]-1,-0.65) arc (190:361:0.15cm);
		\node at (-0.1,7) {\small{$\boxed{\textcolor{white}{\sum}\hspace{-0.4cm}{a\in(0,1)}\hspace{0.15cm}}$}};
		\end{tikzpicture}
		\caption{ Let $\alpha \in (0,\pi)$, $\gamma\in[0,\pi/2]$ and $a\in[-1,1)\setminus\{0\}$. The magnetic field $\mathbf{B}_{ \alpha,\gamma,a}$ in $\R_+^3$ can have different directions in the two regions $\mathcal{D}_\alpha^1$ and $\mathcal{D}_\alpha^2$, according to the sign of $a$. The strength of the magnetic field is $\mathsf s_{ \alpha,a}=1$ in $\mathcal D^1_\alpha$ and $\mathsf s_{ \alpha,a}=a$ in $\mathcal D^2_\alpha$.  The transition of the strength occurs at the plane $P_\alpha$ of equation $x_1\sin\alpha-x_2\cos\alpha=0$, referred to as the discontinuity plane. The angle $\gamma$ (modulo $-\pi$) represents the angle that
		 $\mathbf{B}_{ \alpha,\gamma,a}$ makes with the $x_3$-axis.} 
		\label{fig:1}
	\end{figure}
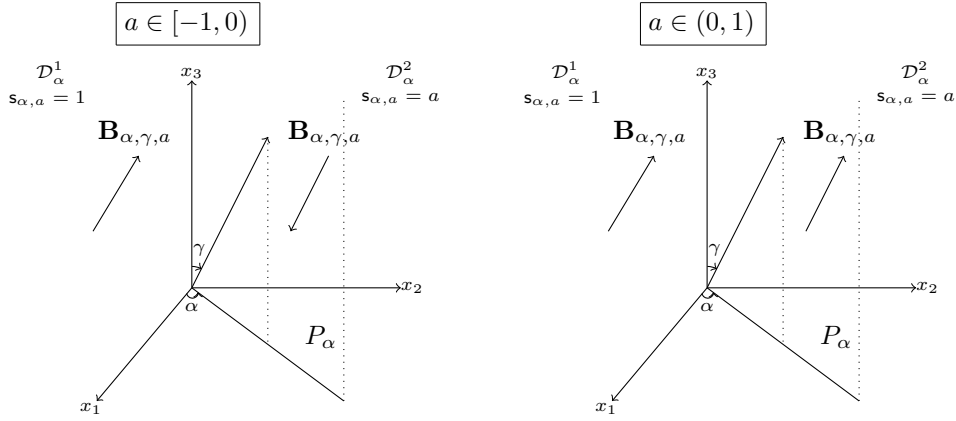
	\begin{multline}\label{eq:domLa++}
	\mathcal{D}(\mathcal L_{ \alpha,\gamma,a}) =\big\{u\in L^2(\R^3_+)~:~ (\nabla-i\Ab_{ \alpha,\gamma,a} )^n u \in L^2(\R^3_+),\\ \mathrm{for}\, n\in \{1,2\},(\nabla-i\Ab_{ \alpha,\gamma,a})u\cdot (0,1,0)|_{\partial \R^3_+}=0\big\}.
	\end{multline}
	In the present paper, we study the bottom of the spectrum of $\mathcal L_{ \alpha,\gamma,a}$  (see Section~\ref{sec:bot}), and give applications to this study in 3D bounded domains (see Section~\ref{sec: application}).

\subsection{Motivation}	\label{sec: motivation}
In the theory of superconductivity and in generic situations, a superconductor submitted to a sufficiently strong magnetic field  loses permanently its superconducting properties when the intensity of the magnetic field exceeds  a certain (unique) critical value---the so-called \textit{third critical field} denoted by $H_{C_3}$. We say that  the material passes to the normal state (see~\cite{saint1963onset,fournais2010spectral}).
	 	The Ginzburg--Landau (GL) model is used   to study  this phase transition from superconducting to normal states.
	 This is naturally a three-dimensional (3D) model, but it is usually reduced  to a two-dimensional (2D) one supposing that the superconductor is a long-cylindrical wire and that the direction of the magnetic  field  is perpendicular to the cross section of the wire (see~e.g.~\cite{SS}).  The 2D GL model was extensively used for both constant or smooth variable external magnetic fields in the case of domains with smooth boundary (see~e.g.~\cite{lu99eigen,helffer2001magnetic,fournais2010spectral,raymond2009sharp, raymond2013var} or domains with corners (see~\cite{BNF, bonnaillie2015ground}). Recently,~\cite{Assaad3,assaad2020magnetic} (see also~\cite{AssaadHearing21}) examined this phase transition for 2D GL models with piecewise-constant magnetic fields . Also, we refer to~\cite{CG1, CG2, CG3, CR1, CR2, HK} for the study of superconductivity right before the normal state.

	 Within this context, 3D models were  studied  in the mathematical literature for more general (bounded or unbounded) domains, not necessarily cylinders, subjected to constant  or smooth variable magnetic fields (see e.g.~\cite{lu2000surface,helffer2004magnetic,popoff2013schrodinger,popoff2015model, raymond2017}). Such studies involved a linear Schr\"odinger operator, $-(h\nabla - i\mathbf{A})^2$, defined  on an open and bounded set   $\Omega\subset \mathbb{R}^3$, with smooth boundary or having edges, where $\mathbf{A}\in H_{\mathrm{loc}}^1(\R^3)$  is a magnetic vector potential and $\mathrm{curl}\mathbf{A} = \mathbf{B}$ is the external magnetic field having a constant or a smooth variable strength. As the semiclassical parameter $h$ goes to $0$, the third critical field $H_{C_3}$ is estimated using the asymptotics of the first eigenvalue, $\lambda(\mathbf{B};\Omega, h)$, of this operator\footnote{Due to gauge invariance \cite[Section~1.1]{fournais2010spectral}, it is standard that the magnetic potential $\mathbf{A}$ contributes to the spectrum of $-(h\nabla - i\mathbf{A})^2$ only through its associated magnetic field $\mathbf{B}$, which justifies the  notation $\lambda(\mathbf{B};\Omega, h)$.}  (see e.g.~ \cite[Proposition 1.9]{fournais2006third},\cite{lu2000surface,giorgi2002breakdown,fournais2010spectral,Assaad3}).  
 Such  asymptotics of $\lambda(\mathbf{B};\Omega, h)$ are usually obtained by using a variational argument where local energies are studied in different zones of the superconductor (like the interior, the boundary, or near the edges). The local study involves effective Schr\"odinger operators  of the form $-(\nabla - i\mathbf{A})^2$, with  magnetic fields having a constant strength, defined on unbounded domains like $\R^3$, $\R^3_+$ or infinite wedges (see~\cite{lu2000surface,bonnaillie2015ground}). While the  operator on $\R^3$ is related to the study in the interior of $\Om$, that on $\R^3_+$ is related to the study at the smooth boundary  of $\Om$ and  depends on the angle between the magnetic field and the boundary $\R^3_+$. Moreover, the  operators on infinite wedges are considered (\cite{Pan2,popoff2013schrodinger,popoff2015model}) for the study near the edges of $\Om$ (when exist), and depend on both the direction of the magnetic field and the opening angle of the wedge. Studying the effective models permit to determine the eventual localization of superconductivity in $\Om$, before its breakdown.  We refer the reader to the introduction in~\cite{popoff2015model} for a brief explanation about the link between the original model on $\Om$ and the various effective models (see also~\cite{bonnaillie2015ground} for a more detailed explanation). 

 Back to the operator $\mathcal{L}_{\alpha, a,\gamma}$ defined in the present contribution, such a 3D operator with a \emph{discontinuous magnetic field} was not considered yet in the literature. We show that  $\mathcal{L}_{\alpha,\gamma,a}$ is a leading operator that plays an essential role in the study of semiclassical problems similar to the aforementioned ones in $\Om$ (see \cite[Section~8]{fournais2022tunneling} and \cite{Gia}), but in new situations where the magnetic field is piecewise-constant. As seen later in the paper, we consider the semiclassical problem associated to a piecewise constant magnetic field in $\Omega\subset\mathbb{R}^3$, provide asymptotics of its ground state energy,  and establish the concentration of its ground state in specific regions near the discontinuity surface, i.e., the surface at which the jump of the magnetic field occurs (see Section \ref{sec: application}).

\subsection{Main results}  We present the  main results of our paper in Section~\ref{sec:bot}, and applications to these results in Section~\ref{sec: application}.
\subsubsection{Bottom of the spectrum of $\mathcal{L}_{\alpha, \gamma, a}$ } \label{sec:bot}
We recall the operator $\mathcal{L}_{\alpha, \gamma, a}$ introduced in \eqref{eq:La+}
	\begin{equation}\label{eq:Lalfa}
		\mathcal{L}_{\alpha, \gamma, a} = - (\nabla - i \mathbf{A}_{\alpha, \gamma, a})^2, \qquad \mathrm{in}\, \mathbb{R}^3_+,
	\end{equation}
	with the domain $\mathcal{D}(\mathcal{L}_{\alpha, \gamma, a})$ defined in \eqref{eq:domLa++}. We consider the bottom of the spectrum of this operator
	\begin{equation}\label{eq: def lambda}\lambda_{\alpha, \gamma, a} := \inf\mathrm{sp}(\mathcal{L}_{\alpha, \gamma , a}).\end{equation} 
	
	Using a Fourier transform, the operator $\mathcal L_{ \alpha,\gamma,a}$ can be decomposed into a family of 2D operators on $\R^2_+$, $\mathcal L_{\underline\Ab_{\alpha,\gamma,a}}+V_{\underline{\mathbf{ B}}_{\alpha,\gamma,a}},\tau$, parametrized by $\tau\in\R$, and defined on $\mathbb{R}^2_+$ as follows
	\begin{equation}\label{eq:LV}
		\mathcal L_{\underline\Ab_{\alpha,\gamma,a}}+V_{\underline{\mathbf{ B}}_{\alpha,\gamma,a},\tau}= -(\nabla-i\underline \Ab_{\alpha,\gamma,a})^2+V_{\underline{\mathbf{ B}}_{\alpha,\gamma,a},\tau},\end{equation} 
	where 
	$\underline{\Ab}_{\alpha, \gamma, a}$  is a vector potential, defined in~\eqref{eq:Abar}, representing a projection of the vector potential $\mathbf{A}_{\alpha,\gamma, a}$ in~\eqref{eq:La+} on $\mathbb{R}^2_+$, $\underline{\mathbf{B}}_{\alpha, \gamma, a}= (\underline b_1,\underline b_2)$ is a magnetic field, defined in~\eqref{eq:bunder}, projecting the field $\underline{\mathbf{B}}_{\alpha, \gamma, a}$ in~\eqref{eq:Bs} on $\mathbb{R}^2_+$, and $V_{\underline{\mathbf{ B}}_{\alpha,\gamma,a},\tau}=\big(x_1\underline b_2-x_2\underline b_1-\tau\big)^2$ is an electric potential defined in~\eqref{eq:V0}.
	The bottom of the spectrum of $\mathcal L_{\underline\Ab_{\alpha,\gamma,a}}+V_{\underline{\mathbf{ B}}_{\alpha,\gamma,a},\tau}$ is denoted by $\underline\sigma(\alpha,\gamma,a,\tau)$, which highlights its dependence on the parameters $\alpha,\gamma,a$, and $\tau$. Having (see~\eqref{eq:l3})
	
	\[\lambda_{\alpha, \gamma,a} =\inf_{\tau\in\R} \underline\sigma(\alpha,\gamma,a,\tau),\]
	the examination of $\lambda_{\alpha, \gamma,a}$ reduces to that of the function $\tau\mapsto\underline\sigma(\alpha,\gamma,a,\tau)$. This examination leads to an important comparison between $\lambda_{\alpha, \gamma,a}$ and other well-known spectral values, $\beta_a$ and $\zeta_{\nu_0}$, where $a\in [-1, 1)\setminus \{0\}$ is the parameter appearing in the definition of $\mathcal{L}_{\alpha, \gamma, a}$ and $\nu_0 :=\arcsin(\sin\alpha\sin\gamma)$. The value $\beta_a$ is the bottom of the spectrum of a Schr\"odinger operator defined on $\R^3$ in \eqref{eq:L3}, with a piecewise-constant magnetic field (splitting $\R^3$ in two half-spaces, the strength of the field takes the values $1$ and $a$  in these half-spaces, respectively). The value $\zeta_{\nu_0}$ is the bottom of the spectrum of a magnetic Neumann Schr\"odinger operator defined on $\mathbb{R}_+^3$ in ~\eqref{eq:O2}, with a constant magnetic field making an angle $\nu_0$ with the $(x_1x_3)$ plane.

\begin{theorem}\label{thm:main}
	Let $a\in[-1,1)\setminus\{0\}$, $\alpha\in(0,\pi)$, $\gamma\in[0,\pi/2]$, and $\nu_0 =\arcsin(\sin\alpha\sin\gamma)$. Let $\lambda_{\alpha, \gamma, a}$ be the bottom of the spectrum of the operator $\mathcal{L}_{\alpha, \gamma, a}$ defined in \eqref{eq:La+}. It holds 
	 \begin{equation}\label{eq:lb1}\lambda_{\alpha, \gamma,a}\leq\min\big(\beta_a,|a|\zeta_{\nu_0}\big),
	 \end{equation}
	where $\beta_a$ and $\zeta_{\nu_0}$ are respectively the bottom of the spectrum of the operators defined in~\eqref{eq:L3} and~\eqref{eq:O2}.
	
	Furthermore, if\begin{equation}\label{eq:lb}
			\lambda_{\alpha, \gamma,a}<\min\big(\beta_a,|a|\zeta_{\nu_0}\big),
		\end{equation}
	then  there exists $\tau_*\in\R$ such that
	\[\lambda_{\alpha, \gamma,a}=\underline\sigma(\alpha,\gamma,a,\tau_*)\] 
	and $\underline\sigma(\alpha,\gamma,a,\tau_*)$ is an eigenvalue of the operator $\mathcal L_{\underline\Ab_{\alpha,\gamma,a}}+V_{\underline{\mathbf{ B}}_{\alpha,\gamma,a}},\tau_*$ defined in \eqref{eq:LV}.
\end{theorem}
\begin{rem}[The choice of $a\in [-1,1)\setminus \{0\}$]\label{rem: a=0}
	One can choose any two distinct real values $b_1$ and $b_2$ for the strength of the magnetic field  $\mathbf{B}_{ \alpha,\gamma,a}$ respectively in $\mathcal D^1_\alpha$ and $\mathcal D^2_\alpha$. However, by a simple scaling argument, one can reduce the study to the case $b_1=1$ and $b_2=a$, where $a$ is a value in $[-1,1)$. 

In the case $a=0$, the energy $\beta_a$, appearing in Theorem~\ref{thm:main}, is equal to zero (see~\cite{hislop2016band}). Hence, the comparison between the three energies $\lambda_{\alpha, \gamma,a}$, $\beta_a$, and $|a|\zeta_{\nu_0}$ is trivial:
		\[\lambda_{\alpha, \gamma,a}\geq\min\big(\beta_a,|a|\zeta_{\nu_0}\big)=0.\]
		Moreover, our proof technically relies on the assumption $a\neq0$ in many places, for instance when using translations to link our problem to the toy models in Section~\ref{sec: known op}, which have  well-explored spectra. We exclude the case $a=0$ from our study. 
\end{rem}
\begin{rem}[On the semiclassical problem]\label{rem: lambda omega}
		As mentioned earlier, the operator $\mathcal L_{ \alpha,\gamma,a}$ will be used in studying a semiclassical problem on a smooth and bounded domain $\Om$ of $\R^3$, subjected to a piecewise-constant magnetic field. This semiclassical problem is introduced in Section~\ref{sec: application}. When the bottom of the spectrum, $\lambda_{\alpha, \gamma,a}$, of  $\mathcal L_{ \alpha,\gamma,a}$ is an eigenvalue of a certain $\mathcal L_{\underline\Ab_{\alpha,\gamma,a}}+V_{\underline{\mathbf{ B}}_{\alpha,\gamma,a}},\tau_*$ , one can use its  corresponding eigenfunction to construct  a trial function in $\Om$, supported near some point(s) of the discontinuity region of $\partial\Om$ corresponding to $(\alpha, \gamma,a)$,  which yields a desired upper bound in the asymptotic estimates of the semiclassical ground state in $\Om$ (see the proof of Proposition \ref{pro: upper bound lambda b}).  This motivates our interest in setting the condition~\eqref{eq:lb} in Theorem~\ref{thm:main}.
\end{rem}
In Theorem~\ref{thm:main}, we gave sufficient conditions for $\lambda_{\alpha, \gamma,a}$ to be an eigenvalue of the operator $\mathcal L_{\underline\Ab_{\alpha,\gamma,a}}+V_{\underline{\mathbf{ B}}_{\alpha,\gamma,a},\tau_*}$ in~\eqref{eq:LV-2}, for a certain $\tau_*\in\R$. Our next result  provides a condition on  $(\alpha,\gamma,a)$ such that~\eqref{eq:lb} is realized.
	\begin{proposition}\label{prop:exn}
		Let $a\in[-1,1)\setminus\{0\}$,  $\alpha\in(0,\pi)$, $\gamma\in[0,\pi/2]$, and $\nu_0 =\arcsin(\sin\alpha\sin\gamma)$. Consider the function $P[\alpha,\gamma,a]:(0,+\infty) \rightarrow \R$ defined by 
		\begin{equation}\label{eq: P}
			P[\alpha,\gamma,a](x) =  A[\alpha, \gamma, a]x^2-\frac \pi 2\Lambda[\alpha,\gamma,a] x+\frac\pi 2,\end{equation}
		with
		\begin{eqnarray}\label{eq: def A}
				A[\alpha, \gamma, a] &:=& \frac 1{128} (-1+\coth\pi)\Big\{\pi\cos^2\gamma\Big[4(a-1)((a-e^\pi)e^{\pi -\alpha} + (ae^{\pi} - 1)e^\alpha)
				\\
				&&-(a-1)^2(e^{2\pi-2\alpha} + e^{2\alpha})-2e^\pi\big(-4a+(3-2a+3a^2)\cosh\pi\big)\Big]\nonumber
				\\
				&&+4(e^{2\pi}-1)\Big[-\big(a^2(\pi-\alpha)+\alpha\big)\big(-3+\cos(2\gamma)\big)+2(a^2-1)\sin^2\gamma\sin(2\alpha)\Big]\Big\}\nonumber
		\end{eqnarray}
		and
		\begin{equation}\label{eq:lmda}
			\Lambda[\alpha,\gamma,a]:=\min(\beta_a,|a|\zeta_{\nu_0}),
		\end{equation}
		where $\beta_a$ and $\zeta_{\nu_0}$ are respectively the bottom of the spectrum of the operators defined in~\eqref{eq:L3} and~\eqref{eq:O2}.
		If there exists $x=x(\alpha,\gamma,a)>0$ such that $P[\alpha,\gamma,a](x)<0$, then $\inf_{\tau} \underline\sigma(\alpha,\gamma,a,\tau)$ is attained in $\R$, i.e. there exists $\tau_*\in\R$ satisfying
		\[\inf_{\tau} \underline\sigma(\alpha,\gamma,a,\tau)=\underline\sigma(\alpha,\gamma,a,\tau_*).\] 
		Moreover, $\underline\sigma(\alpha,\gamma,a,\tau_*)$ is an eigenvalue of the operator $\mathcal L_{\underline\Ab_{\alpha,\gamma,a}}+V_{\underline{\mathbf{ B}}_{\alpha,\gamma,a},\tau_*}$ defined in \eqref{eq:LV-2}.
	\end{proposition}

\begin{rem}[Admissible triplets $(\alpha, \gamma, a)$]\label{rem: admissible triple} In Section ~\ref{sec: known op}, we provide the following lower bound for the value $\Lambda[\alpha,\gamma,a]$  in~\eqref{eq:lmda}
\[\Lambda[\alpha,\gamma,a] \geq|a|\Theta_0,\]  where $\Theta_0$ is the de Gennes constant defined in~\eqref{eq:teta0}. Moreover,\cite{bonnaillie2012harmonic} gives an explicit lower bound , $\Theta_0^{\mathrm{low}}$, of $\Theta_0$ equal to $0.590106125-10^{-9}$. Hence, if one defines
\[
P_{\Theta_0^{\mathrm{low}}}[\alpha,\gamma,a](x) := A[\alpha, \gamma, a]x^2-\frac \pi 2|a|\Theta_0^\mathrm{low} x+\frac\pi 2,\]
for $x>0$ and $A[\alpha, \gamma, a]$  as in \eqref{eq: def A}, one observes that $P[\alpha,\gamma,a](x)\leq P_{\Theta_0^{\mathrm{low}}}[\alpha,\gamma,a](x) $. By computation, we get that for all $a\in[-1,1)\setminus\{0\}$,  $\alpha\in(0,\pi)$ and $\gamma\in[0,\pi/2]$, $P_{\Theta_0^{\mathrm{low}}}[\alpha,\gamma,a](x)$ admits a minimum $\underline{x}>0$. Using Mathematica, we plot the region of triplets $(\alpha,\gamma,a)$ satisfying 
\[
	\min_{x>0}P_{\Theta_0^{\mathrm{low}}}[\alpha,\gamma,a](x)=P_{\Theta_0^{\mathrm{low}}}[\alpha,\gamma,a](\underline{x})<0.
\]
 These triplets are represented by the colored region in Figure~\ref{fig:Exn}. Consequently, the corresponding $\lambda_{\alpha, \gamma,a}=\inf_{\tau} \underline\sigma(\alpha,\gamma,a,\tau)$ is equal to  $\underline\sigma(\alpha,\gamma,a,\tau_*)$, for a certain  $\tau_*=\tau_*(\alpha,\gamma,a)\in\R$. Furthermore, $\underline\sigma(\alpha,\gamma,a,\tau_*)$ is   
an eigenvalue of the corresponding operator $\mathcal L_{\underline\Ab_{\alpha,\gamma,a}}+V_{\underline{\mathbf{B}}_{\alpha,\gamma,a},\tau_*}$ defined in \eqref{eq:LV-2}.
\end{rem}
\begin{figure}
	\centering
	\includegraphics[scale=0.55]{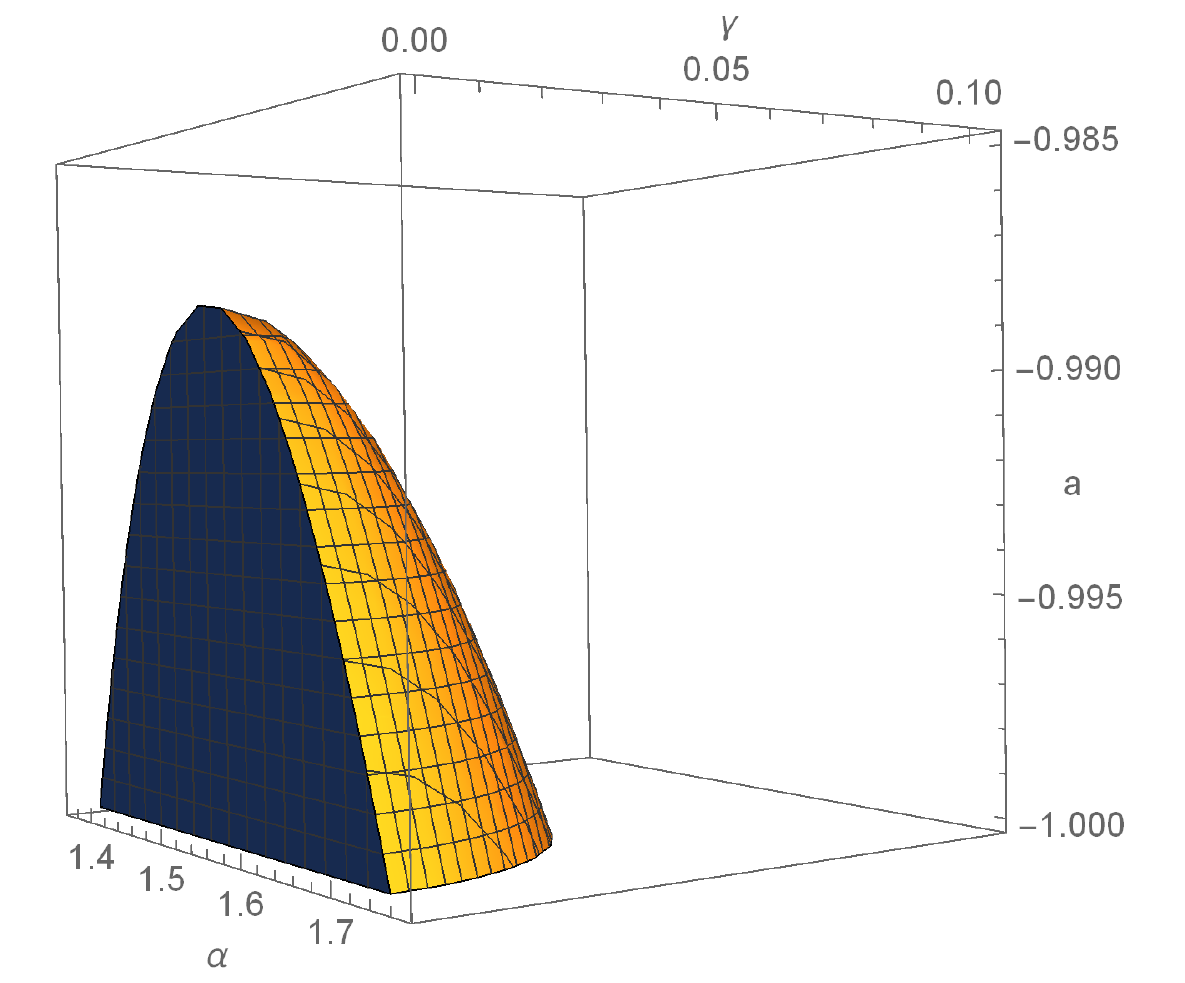}
	\caption{For the triplets $(\alpha,\gamma,a)$ in the colored region, $\lambda_{\alpha, \gamma,a}$ is an eigenvalue of an operator $\mathcal L_{\underline\Ab_{\alpha,\gamma,a}}+V_{\underline{\mathbf{B}}_{\alpha,\gamma,a},\tau_*}$.}
	\label{fig:Exn}
\end{figure}

\subsubsection{Applications: a semiclassical problem in a 3D bounded domain}\label{sec: application}
Let $\Omega\subset\mathbb{R}^3$ be an open bounded and simply connected set, with a smooth boundary. Let $\mathcal C$ be a simple smooth curve in the $(x_1x_2)$ plane, with an infinite length. We define $S$ as the intersection between $\mathcal C\times \R$ and $\Omega$:
\[S:=(\mathcal C\times \R)\cap \Omega.\]
We assume that $S$ cuts $\Omega$ into two disjoint non-empty open sets $\Omega_1$ and $\Omega_2$:
\[\Omega=\Omega_1\cup\Omega_2\cup S.\]

Let $a\in [-1,1) \setminus \{0\}$, we define a piecewise constant magnetic field $\mathbf{B}$ in $\Omega$ as follows:
\begin{equation}\label{eq: magn field applic}
	\mathbf{B}(x) = \mathbf{s}(x)(0,0,1), \qquad \mathbf{s} = \mathbbm{1}_{\Omega_1} + a \mathbbm{1}_{\Omega_2}.
\end{equation}
Note that the magnetic field $\mathbf{B}$ is tangent to $S$.  Moreover, its strength $|\mathbf{B}|$ exhibits a  discontinuity jump at $S$. We will refer to $S$ as the \textit{discontinuity surface}. Moreover, we denote by $\Gamma$  the boundary of $S$ 
\[\Gamma:=\partial S=(\mathcal C\times \R)\cap \partial \Omega.\]
We refer to $\Gamma$ as the \textit{discontinuity curve} (see~Figure~\ref{fig:Omega}).

\begin{figure}
		\centering\begin{tikzpicture}[scale= 0.7]
		 \draw[preaction={draw}]
    plot[smooth cycle]
    coordinates{
      (0,2) (0,5,3) (0.8,6.5) (1.8, 7.5) (2.5,7.22) (3.5,6.5) (5, 6.5) (6,5) (6,2) (6,0) (5, -2) (2, -2.3) (0,0)
    };
     \draw[orange, preaction={draw, pattern=north east lines, pattern color = teal}]
    plot[smooth cycle]
    coordinates{
      (2.5,7.2) (3.8,4.5) (4,3) (4.7,1) (5, -2) (3,-1) (2.5,0) (2,1) (2,3) (1.9,4)
    };
 	\draw[->] (2.8,3) to (2.6,5);
 	\node at (3.2, 4) {\tiny{$\mathbf{B}$}};
    \draw[dotted] (0,2) to [bend left = 10] (6,2);
    \draw[dotted] (0,2) to [bend right = 10] (6,2);
    \node at (4.3, -1) {\tiny{$S$}};
    \node at (3.7, 6) {\tiny{$\Gamma$}};
    \node at (0.5, 0.5) {\tiny{$\Omega$}};
\end{tikzpicture}
\caption{ Illustration of the domain $\Omega$, the discontinuity surface $S$ (shaded), and the discontinuity edge $\Gamma:=\partial S$. The magnetic field $\mathbf{B}$ is tangent to $S$ and its strength exhibits a jump of discontinuity at $S$.} 
	\label{fig:Omega}
\end{figure}
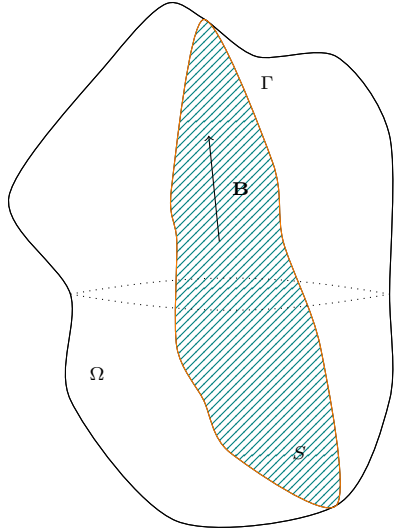

Let  $\mathbf{F}\in H^{1}(\Omega, \mathbb{R}^3)$ be a vector potential  satisfying $\mathrm{curl}\mathbf{F} = \mathbf{B}$ (see~\cite[Theorem D.3.1]{fournais2010spectral}).
Let $\mathfrak{b}>0$, we consider the Neumann realization of the following self-adjoint operator in $\Omega$:
\begin{equation}\label{eq: def semicl op}
	\mathcal{P}_{\mathfrak{b}, \mathbf{F}}= - (\nabla - i \mathfrak{b}\mathbf{F})^2.
\end{equation}
The domain of $\mathcal{P}_{\mathfrak{b}, \mathbf{F}}$ is 
\begin{equation}\label{eq: dom P}
	 \mathcal{D}( \mathcal{P}_{\mathfrak{b}, \mathbf{F}}):=\{  u \in L^2(\Omega) \,:\,  (\nabla-i\mathfrak{b} F)^ju\in L^2(\Omega) \,\,\, \mbox{for}\,\,\, j =1,2,\quad (\nabla - i \mathfrak{b}\mathbf{F}) u \cdot n\vert_{\partial \Omega} = 0 \},
\end{equation}
where $n$ is the inward unit normal vector of $\partial \Omega$. We denote by $\lambda(\mathfrak{b})$ the bottom of the spectrum (the lowest eigenvalue or the ground state energy) of $\mathcal{P}_{\mathfrak{b}, \mathbf{F}}$.

We carry out our analysis in the asymptotic regime\footnote{Taking $\mathfrak b\rightarrow+\infty$ in our problem is equivalent to taking $h\rightarrow  0$ in the semiclassical problems aforementioned in Section~\ref{sec: motivation}. } $\mathfrak b\rightarrow+\infty$. The result of this section (Theorem~\ref{thm: localization} below) is established under the following assumption:
\begin{assumption}\label{asu} For any point $x\in\Gamma$, we denote by $\alpha_x$ the angle between the tangent plane of $\partial\Omega$  and the discontinuity surface $S$ at $x$ (taken towards $\Omega_1$). We also denote by $\gamma_x$ the angle between the magnetic field $\mathbf{B}$ and the discontinuity edge $\Gamma$ at $x$. We assume that there exists $\overline{x}\in\Gamma$ such that 
\begin{equation*} 
	\lambda_{\alpha_{\overline{x}},\gamma_{\overline{x}}, a}<|a|\Theta_0,
\end{equation*}
	where $\lambda_{\alpha_{\overline{x}},\gamma_{\overline{x}}, a}$ is defined in \eqref{eq: def lambda}.

In this case, it obviously holds that
\[\inf_{x\in\Gamma}\lambda_{\alpha_{x},\gamma_{x}, a}<|a|\Theta_0.\]
\end{assumption}
\begin{rem}[Comments on Assumption \ref{asu}]\label{rem:asu}
	 Let $a\in [-1,1]\setminus\{0\}$, $\alpha\in (0,\pi)$, $\gamma\in [0,\pi/2]$.
	\begin{enumerate}
		\item[(i)] We validate the above assumption as follows. Proposition~\ref{prop:exn} (and the computation below it) provide examples of triplets $(\alpha,\gamma, a)$  where
		\begin{equation}\label{eq:1}
			\lambda_{\alpha,\gamma, a}<\min(\beta_a,|a|\zeta_{\nu_0}),
		\end{equation}
	for $\nu_0=\arcsin(\sin\alpha\sin\gamma)$. Having at least a point $\overline{x}\in\Gamma$ such that $\gamma_{\overline x}=0$, we can assume that the corresponding couple $(\alpha_{\overline{x}},a)$ at $\overline{x}$ satisfies~\eqref{eq:1}, that is
	\begin{equation}\label{eq:2}
		\lambda_{\alpha_{\overline{x}},0, a}<\min(\beta_a,|a|\zeta_{0}),
		\end{equation} 	
($\zeta_{0}=\zeta_{\nu_{0}=0}$).	Indeed, one can take $(\alpha_{\overline{x}},a)$ living sufficiently near $(\pi/2,-1)$  to get~\eqref{eq:2} satisfied (see Proposition~\ref{prop:exn} and Figure~\ref{fig:Exn})\footnote{A simple instance of $\Omega$ and $\mathbf{B}$ satisfying Assumption~\ref{asu} is when having $\Omega$ a ball that cuts the plane $\{x_2=0\}$ by an angle $\alpha=\pi/2$ (or sufficiently near $\pi/2$) and having $a=-1$ (or sufficiently near $-1$).}.
	
	Now, by Sections~\ref{sec:eff2} and~\ref{sec:nu+}
	\[\zeta_{0}=\Theta_0,\ \mbox{and}\ \beta_a\geq|a|\Theta_0.\]
	Hence,~\eqref{eq:2} reads
	\begin{equation}\label{eq:3}
			\lambda_{\alpha_{\overline{x}},0, a}<|a|\Theta_0.
	\end{equation}	
	\item[(ii)] The conditions in Assumption~\ref{asu} are motivated in what follows. 
\begin{enumerate}
	\item 
For the points $\overline{x}$ satisfying the foregoing assumption, we have
	\begin{equation}\label{eq:c1}
		\lambda_{\alpha_{\overline{x}},\gamma_{\overline{x}}, a} <|a|\Theta_0\leq\min(\beta_a,|a|\zeta_{\overline\nu}),\end{equation}
	where $\overline\nu:=\arcsin(\sin\overline\alpha\sin\overline\gamma)$.
	The last inequality follows from the fact that $\beta_a\geq|a|\Theta_0$ (see Section~\ref{sec:eff2}), $\zeta_{0}=\Theta_0$, and the function $\nu\mapsto\zeta_\nu$ is strictly increasing on $[0,\pi/2]$ (see Section~\ref{sec:nu+}).
	Hence, Assumption~\ref{asu} implies that the energy $\lambda_{\alpha_{\overline x},\gamma_{\overline x}, a}$ corresponding to each of these points satisfies the condition~\eqref{eq:lb} in Theorem~\ref{thm:main}. This theorem ensures that $\lambda_{\alpha_{\overline{x}},\gamma_{\overline{x}}, a}$ is an eigenvalue of the operator $\mathcal L_{\underline\Ab_{\alpha_{\overline x},\gamma_{\overline x},a}}+V_{\underline{\mathbf{ B}}_{\alpha_{\overline x},\gamma_{\overline x},a},\tau_*}$, for some $\tau_*\in\R$, defined in \eqref{eq:LV-2}.  The corresponding eigenfunction will be used in establishing an upper bound of $\lambda(\mathfrak b)$ (see Proposition~\ref{pro: upper bound lambda b}). 
	\item Furthermore, while the condition
		\[\lambda_{\alpha_{\overline{x}},\gamma_{\overline{x}}, a}<\min(\beta_a,|a|\zeta_{\overline\nu})\]
		in~\eqref{eq:c1} is sufficient to get the foregoing upper bound of Proposition~\ref{pro: upper bound lambda b}, the more strict condition
		\[\lambda_{\alpha_{\overline{x}},\gamma_{\overline{x}}, a}<|a|\Theta_0\]
		in Assumption~\ref{asu}  is crucial in localizing the ground state in $\Omega$ near the points $\overline x$ of $\Gamma$ realizing this assumption (see Theorem~\ref{thm: localization} below). Indeed, under this strict condition, the local energies near the points $\overline{x}$ will be strictly smaller than those away from these points. The minimum of the  latter energies is comparable with $\mathfrak b|a|\Theta_0$, as $\mathfrak b\rightarrow+\infty$ (see Proposition~\ref{pro: lower bound quadratic form}).
\end{enumerate}
	\end{enumerate} 
\end{rem}
Our next result provides a localization of the eigenfunction (the ground state) corresponding to the eigenvalue $\lambda(\mathfrak{b}) $ near the points of the discontinuity edge $\Gamma$ satisfying the conditions in Assumption~\ref{asu}. 
Let $D$ be the set of these points
\begin{equation}\label{eq:D}
D= \big\{ \overline{x}\in \Gamma\,\,\, \vert\,\,\, \lambda_{\alpha_{\overline{x}}, \gamma_{\overline{x}}, a}  < |a| \Theta_0\big\}.
\end{equation}

\begin{thm}\label{thm: localization}
 Under Assumption \ref{asu}, there exist positive constants $C, \eta,\mathfrak{b}_0$ such that for all $\mathfrak{b}\geq \mathfrak{b}_0$ and for any  ground state $\psi$ of $\mathcal{P}_{\mathfrak{b}, \mathbf{F}}$, it holds
 \begin{equation}
 	\int_{\Omega} e^{2\eta \sqrt{\mathfrak{b}} \mathrm{dist}(x, D)}\left(|\psi|^2  + \mathfrak{b}^{-1}|(\nabla - i \mathfrak{b}\mathbf{F})\psi|^2 \right)\,dx \leq C\|\psi\|^2_{L^2(\Omega)}.
 \end{equation}
 Consequently, for any $N\in \mathbb{N}$,  the following localization estimate holds
 \begin{equation}
 	\int_{\Omega} \mathrm{dist}(x, D)^N |\psi|^2\,dx = \mathcal{O}(\mathfrak{b}^{-\frac{N}{2}}).
 \end{equation}
\end{thm}
\bigskip

\subsection{Paper organization}
The rest of the paper is organized as follows. In Section \ref{sec: known op}, we recall some known operators in the plane and the half-space which are useful for our analysis.  In Section \ref{sec: R3+ steps}, we decompose our operator into 2D reduced operators. For these reduced operators, we derive some properties of the bottom of essential spectrum and the bottom of the spectrum.
The proof of Theorem \ref{thm:main} is then established in Section~\ref{sec:proof}. We use Theorem \ref{thm:main} in Section 5 to prove  Theorem \ref{thm: localization}, establishing localization results in 3D bounded domains under discontinuous magnetic fields.
\section{Known effective operators}\label{sec: known op}
In this section, we introduce useful linear Schr\"odinger operators on the plane and  the half-space that were explored earlier in the literature. 
\subsection{An operator with a discontinuous  magnetic field on the plane}\label{sec:eff2}
Let  $a \in [-1,1)\setminus \{0\}$. We consider a magnetic potential ${\mathcal A}_a\in H^1_{\mathrm{loc}}(\R^2,\R^2)$ with the following associated piecewise-constant magnetic field 
\begin{equation*}\label{eq:A1}
\curl{\mathcal A}_a(x)=\mathbbm{1}_{\{x_2 >0\}}(x)+a\mathbbm{1}_{\{x_2 <0\}}(x),\quad (x_1,x_2)\in \R^2.
\end{equation*}
We introduce the self-adjoint operator on $\R^2$
\begin{equation*}\label{eq:ham_operator}
\mathcal L_a=-(\nabla-i{\mathcal A}_a)^2,
\end{equation*}
with domain
\[
	\mathcal{D}(\mathcal L_a) :=\big\{u\in L^2(\R^2)\,:\,(\nabla-i{\mathcal A}_a)^n u \in L^2(\R^2),\, \mathrm{for}\ n\in \{1,2\}\big\}.
\]
  We denote the bottom of the spectrum by $\beta_a$.
The operator $\mathcal L_a$ has been studied in~\cite{hislop2015edge,hislop2016band,Assaad2019,Assaadlowest20}: using a  Fourier transform, $\mathcal L_a$ was reduced to  a  family of Schr\"odinger operators on $L^2(\R)$, $\mathfrak h_a[\xi]$, parametrized by $\xi\in\R$. For each fixed $\xi\in \mathbb{R}$, the operator $\mathfrak h_a[\xi]$ is defined by
\begin{equation}\label{eq:ha}
\mathfrak h_a[\xi] =
\begin{cases}
-\frac{d^2}{dt^2}+(at-\xi)^2,& t<0\,,\\
-\frac{d^2}{dt^2}+(t-\xi)^2,& t>0\,.
\end{cases}
\end{equation}
We have {(see~\cite{Assaad2019})}
\begin{equation}\label{eq:beta}
\beta_a=\inf_{\xi \in \R} \mu_a(\xi)\,,
\end{equation} 
where $\mu_a(\xi)$ is the  bottom of the spectrum of the operator $\mathfrak h_a[\xi]$, i.e., 
\begin{equation}\label{eq:mu_a}
\mu_a(\xi)=\inf \mathrm{sp}\big(\mathfrak h_a[\xi]\big).
\end{equation}

  We collect the following useful properties of $\beta_a$:
\begin{itemize}
	\item For $0<a<1$,  $\beta_a=a$ and $\beta_a$ is not attained by $\mu_a(\xi)$, for all $\xi\in \R$.
	\item For $-1\leq a<0$, $|a|\Theta_0\leq\beta_a<|a|$ and  $\beta_a=\mu_a(\xi_a)$, for a certain (unique) $\xi_a \in \R$. Here, $\Theta_0$ is the de Gennes constant defined as the bottom of the spectrum of the magnetic Neumann realization of the Schr\"odinger operator $-(\nabla - i\Ab)^2$, with a unit magnetic field ($\curl\Ab=1$),  on the half-plane (see~e.g.~\cite{fournais2010spectral})

\begin{equation}\label{eq:teta0}
\Theta_0=\inf\mathrm{sp} [-(\nabla - i\Ab)^2]\approxeq 0.59.
\end{equation}
\end{itemize}

\begin{rem}[The value $\beta_a$ as the bottom of spectrum of a Schr\"odinger operator on $\R^3$]\label{rem: beta a R3}
Consider the following Schr\"odinger operator on $\R^3$
	\begin{equation}\label{eq:L3}
\mathsf{L}_{a} := -(\nabla - i\mathsf{A}_{a})^2,	
\end{equation}
with domain 
\begin{equation} 
	\mathcal{D}(\mathsf{L}_a) := \left\{ u\in L^2(\mathbb{R}^3)\, \vert\, (\nabla - i \mathsf{A}_{a})^ju\in L^2(\mathbb{R}^3), \quad \mbox{for}\,\,\, j\in\{1,2\} \right\},
\end{equation}
where $\mathsf{A}_{a}\in H^1_{\mathrm{loc}}(\mathbb{R}^3, \mathbb{R}^3)$ is a magnetic potential such that the corresponding magnetic field has a piecewise-constant strength equal to $\mathbbm{1}_{\{x_2 >0\}} + a\mathbbm{1}_{\{x_2 <0\}}=:\delta_a$. More precisely, we can fix the gauge and set $\mathsf{A}_{a} = (-\delta_a x_2, 0,0)$. With this choice of the magnetic potential, performing a partial Fourier transform it is possible to compare the bottom of the spectrum of $\mathsf{L}_a$ with the bottom of the spectrum of the operator $\mathcal{L}_a$ introduced above, proving that 
\begin{equation}
	\inf\mathrm{sp}(\mathsf{L}_a) = \beta_a,
\end{equation}
where $\beta_a$ is as in \eqref{eq:beta}.
\end{rem}
\subsection{An operator with a constant field on the half-space}\label{sec:nu+}
Let  $\nu\in[0,\pi/2]$. We introduce the  following magnetic field  with a unit strength on $\R_+^3$
\begin{equation*}\label{eq:Bnu}
	\mathbf{B}_\nu=(0,\sin\nu,\cos\nu),
\end{equation*}
and an associated magnetic potential $\Ab_\nu\in H^1_{\mathrm{loc}}(\R_+^3,\R^3)$, ($\curl \Ab_\nu=\mathbf{B}_\nu$). Note that $\mathbf{B}_\nu$ makes an angle $\nu$ with the $(x_1x_3)$ plane.

Now, we consider the magnetic Neumann realization of the following self-adjoint operator on the half space 
\begin{equation}\label{eq:O2}
	\mathrm H_\nu=-(\nabla-i\Ab_\nu)^2\quad  \mathrm{in}\ L^2(\R^3_+),
\end{equation}
We denote by $\zeta_\nu$ the bottom of the spectrum of $\mathrm{H}_\nu$, 
\begin{equation}\label{eq:zeta}
	\zeta_\nu = \inf \mathrm {sp}\big(\mathrm H_\nu\big).
\end{equation} 
We present the following useful properties of $\zeta_\nu$ (see~e.g.~\cite{lu2000surface,lu99eigen,morame2005remarks}):
\begin{equation}\label{eq: prop zeta}
\zeta_0=\Theta_0,\qquad  \zeta_{\pi/2}= 1,\qquad   \zeta_\nu\in(\Theta_0,1)\,\,\, \mathrm{for}\  \nu\in(0,\pi/2),
\end{equation} 
where $\Theta_0$ is the de Gennes constant defined in~\eqref{eq:teta0}. Moreover, (see~e.g.~\cite{lu2000surface,lu99eigen,morame2005remarks}), the map $\nu\mapsto\zeta_\nu$ is strictly increasing for $\nu\in [0, \pi/2]$.
%
\section{The operator with magnetic steps on the half space}\label{sec: R3+ steps}
Let  $a\in[-1,1)\setminus\{0\}$, $\alpha\in(0,\pi)$ and $ \gamma \in[0,\pi/2]$. We recall the operator $\mathcal L_{ \alpha,\gamma,a}$ with a discontinuous field  on $\R^3_+$ introduced in \eqref{eq:La+}
\begin{equation} \mathcal{L}_{ \alpha,\gamma,a} = -(\nabla - i \mathbf{A}_{\alpha,\gamma,a})^2, \end{equation}
with the domain defined in \eqref{eq:domLa++} as  
\begin{multline}
	\mathcal{D}(\mathcal L_{ \alpha,\gamma,a}) =\big\{u\in L^2(\R^3_+)~:~ (\nabla-i\Ab_{ \alpha,\gamma,a} )^{n} u \in L^2(\R^3_+),\\ \mathrm{for}\ n\in \{1,2\},(\nabla-i\Ab_{ \alpha,\gamma,a})u\cdot (0,1,0)|_{\partial \R^3_+}=0\big\}.
	\end{multline}
Using the min-max principle, we  write the bottom of the spectrum of $\mathcal{L}_{\alpha, \gamma, a}$ as 
\begin{equation}\label{eq:Lmdas2}
\lambda_{\alpha, \gamma,a} =\inf_{\substack{u\in \mathcal{D}(Q_{\alpha, \gamma,a}) \\ u\neq 0 }}\frac{Q_{\alpha, \gamma,a}(u)}{\|u\|_{ L^2(\R^3_+)}^2},
\end{equation}
where $Q_{\alpha, \gamma,a}$ is the quadratic form associated to the operator $\mathcal{L}_{\alpha,\gamma,a}$, defined by  
\[Q_{\alpha, \gamma,a} (u)=\|(\nabla-i\Ab_{ \alpha,\gamma,a} )u\|_{L^2(\R_+^3)}^2\]
on the domain
\[\mathcal{D}(Q_{\alpha, \gamma,a}) :=\big\{u\in L^2(\R_+^3)~:~(\nabla-i\Ab_{ \alpha,\gamma,a} ) u \in L^2(\R_+^3)\big\}.\]
We also recall the magnetic field introduced in~\eqref{eq:Bs}, and we denote by $b_j$, $j=1,2,3$ its components:
\begin{equation}\mathbf{B}_{ \alpha,\gamma,a} = (\cos\alpha\sin\gamma,\sin\alpha\sin\gamma,\cos\gamma)\mathsf s_{ \alpha,a} =: (b_1,b_2, b_3),
\end{equation}
where $s_{ \alpha,a} = \mathbbm{1}_{\mathcal D^1_\alpha}+a\mathbbm{1}_{\mathcal D^2_\alpha}$ (see Figure \ref{fig:1}).
Now, we fix the choice of the magnetic potential $\Ab_{ \alpha,\gamma,a} $. Let 
\begin{equation}\label{eq:Ag}
\Ab_{ \alpha,\gamma,a} =(A_1,A_2,A_3)
\end{equation}
such that
\begin{align*}
A_1&=0\\
A_2&=\begin{cases}
\cos\gamma x_1-(1-a)\cos\gamma  \cot\alpha x_2 &\mathrm{for}\ x\in\mathcal D^1_\alpha\\
 a \cos\gamma x_1&\mathrm{for}\ x\in\mathcal D^2_\alpha
\end{cases}\\
A_3&=\begin{cases}
x_2\cos\alpha\sin\gamma-x_1\sin\alpha\sin\gamma &\mathrm{for}\ x\in\mathcal D^1_\alpha\\
 a(x_2\cos\alpha\sin\gamma-x_1\sin\alpha\sin\gamma)&\mathrm{for}\ x\in\mathcal D^2_\alpha.
\end{cases}
\end{align*}
This choice of the vector potential guarantees the continuity of $\Ab_{ \alpha,\gamma,a} $ at the discontinuity plane $P_\alpha$ (see Figure \ref{fig:1}) and that $\Ab_{ \alpha,\gamma,a} \in H^1_{\mathrm{loc}}(\R^3_+,\R^3)$. Moreover, with this vector potential, the operator $\mathcal{L}_{\alpha,\gamma, a}$ is translation invariant in the $x_3$ variable. Hence, its spectrum is absolutely continuous, and a reduction of the study to a family of 2D operators is allowed as we see below.
\subsection{A family of reduced 2D operators}\label{sec: red 2D op}
Let $a\in [-1,1)\setminus \{0\}$, $\alpha\in (0,\pi)$, $\gamma\in [0,\pi/2]$. A partial Fourier transform in the $x_3$ variable yields the following decomposition of the operator $\mathcal L_{ \alpha,\gamma,a} $ (see~\cite{simon1972methods})
\begin{equation}\label{eq:F1}
	\mathcal L_{ \alpha,\gamma,a} =\int^\oplus_{\tau\in\R}\big(\mathcal L_{\underline\Ab_{\alpha,\gamma,a}}+V_{\underline{\mathbf{ B}}_{\alpha,\gamma,a},\tau}\big)\,d\tau,\end{equation}
where 
\begin{equation}\label{eq:LV-2}
	\mathcal L_{\underline\Ab_{\alpha,\gamma,a}}+V_{\underline{\mathbf{ B}}_{\alpha,\gamma,a},\tau}= -(\nabla-i\underline \Ab_{\alpha,\gamma,a})^2+V_{\underline{\mathbf{ B}}_{\alpha,\gamma,a},\tau}\end{equation}
is a  Schr\"odinger operator on $\R_+^2:=\{(x_1,x_2)\in \R^2~:~x_2>0\}$, parametrized by $\tau\in\R$ and such that we have the following.
\begin{itemize}	
	\item The magnetic potential $\underline{\Ab}_{\alpha, \gamma, a}:=(\underline A_1,\underline A_2)$ represents the projection of the vector potential $\mathbf{A}_{\alpha,\gamma, a}$ defined in \eqref{eq:Ag} on $\mathbb{R}^2_+$, i.e., 
	\begin{align}\label{eq:Abar}
	\underline A_1&:=0\nonumber\\
	\underline A_2&:=\begin{cases}
	\cos\gamma x_1-(1-a)\cos\gamma\cot\alpha x_2&\mathrm{for}\ (x_1,x_2)\in D^1_\alpha,\\
	a \cos\gamma x_1&\mathrm{for}\ (x_1,x_2)\in D^2_\alpha,
	\end{cases}
	\end{align}
where $D_\alpha^1$ and $D^2_\alpha$ represent respectively the orthogonal projection of the regions $\mathcal{D}_\alpha^1$ and $\mathcal{D}_\alpha^2$ over the plane $(x_1x_2)$:
\begin{equation}\label{eq: def D1alpha R2}
 	D^1_\alpha = \left\{(x_1,x_2)\in\mathbb{R}^2\, \vert\, (x_1,x_2)=  \rho(\cos\theta,\sin\theta),\, \rho\in(0,\infty),\,0<\theta<\alpha\right\},
 \end{equation}
 \begin{equation}\label{eq: def D12alpha R2}
 D^2_\alpha = \left\{(x_1,x_2)\in\mathbb{R}^2\, \vert\, (x_1,x_2) =  \rho(\cos\theta,\sin\theta),\, \rho\in(0,\infty),\,\alpha<\theta<\pi\right\}.
 \end{equation}
Note that $\underline{\mathbf{A}}_{\alpha,\gamma, a}$ satisfies 
\begin{equation}\label{eq:Ab}
\underline b_3:=\curl\underline \Ab_{ \alpha,\gamma,a}=\underline {\mathsf s}_{{\alpha,a}}\cos\gamma, 
\end{equation} 
where $\underline {\mathsf s}_{\alpha, a}$ is the step function defined in $\R^2_+$ by
	\begin{equation}\label{eq:step}
		\underline {\mathsf s}_{\alpha,a} = \mathbbm{1}_{D^1_\alpha}+a\mathbbm{1}_{D^2_\alpha}.\end{equation}
	\item The field $\underline{\mathbf{B}}_{\alpha, \gamma, a}$ is a magnetic field that projects $\mathbf{B}_{\alpha,\gamma,a}$ on $\mathbb{R}^2_+$ and it is defined as follows
	\begin{eqnarray}\label{eq:bunder}
	 \underline{\mathbf{B}}_{\alpha, \gamma, a}= (\underline b_1,\underline b_2) = (\cos\alpha\sin\gamma, \sin\alpha\sin\gamma) \underline{\mathbf{s}}_{\alpha,a}.
	\end{eqnarray}
	Note that $\underline{\mathbf{B}}_{\alpha, \gamma, a}$ is discontinuous along the line $l_\alpha := P_\alpha\cap \mathbb{R}_+^2$ (see Figure \ref{fig:1}), in the following we refer to $l_\alpha$ as the discontinuity line.
	\item The electric potential $V_{\underline{\mathbf{ B}}_{\alpha,\gamma,a},\tau}$ is defined as
	\begin{eqnarray}\label{eq:V0} V_{\underline{\mathbf{ B}}_{\alpha,\gamma,a},\tau}&=&\big(x_1\underline b_2-x_2\underline b_1-\tau\big)^2,
		\\
		&=& \nonumber} 	{\begin{cases}  (x_1\sin\alpha\sin\gamma - x_2 \cos\alpha\sin\gamma - \tau)^2 &\mbox{for }\, \, (x_1, x_2)\in D^1_\alpha, \\ [a(x_1\sin\alpha\sin\gamma - x_2 \cos\alpha\sin\gamma)- \tau]^2 &\mbox{for }\,\, (x_1, x_2)\in D^2_\alpha.\end{cases}
	\end{eqnarray}
 We highlight the dependence of the electric potential on the magnetic field $\underline{\mathbf{ B}}_{\alpha, \gamma, a}$.
	\end{itemize}
We introduce the quadratic form associated to $\mathcal L_{\underline\Ab_{\alpha, \gamma, a}}+V_{\underline{\mathbf{B}}_{\alpha, \gamma, a},\tau}$:
\begin{equation}\label{eq:Qa}
\underline Q_{\alpha, \gamma, a}^{\tau}(u)=\int_{\R^2_+}\Big(|(\nabla-i\underline \Ab_{\alpha, \gamma, a})u|^2+V_{\underline{\mathbf{ B}}_{\alpha, \gamma, a},\tau}|u|^2\Big)\,dx_1dx_2.\end{equation}
The form domain is
\begin{equation*}\label{eq:domQ}
\mathcal{D}(\underline Q^{\tau}_{\alpha, \gamma, a})=\big\{u\in L^2(\R^2_+)~:~ (\nabla-i\underline\Ab_{\alpha, \gamma, a}) u \in L^2(\R^2_+),\ |x_1\underline b_2-x_2\underline b_1|u\in L^2(\R^2_+)\big\}.\end{equation*}
We denote by $\underline{\sigma}(\alpha, \gamma, a, \tau)$ the bottom of the spectrum of the operator $\mathcal L_{\underline\Ab_{\alpha,\gamma, a}}+V_{\underline{\mathbf{ B}}_{\alpha,\gamma,a},\tau}$. We have
\begin{equation}\label{eq:unders}
\underline\sigma(\alpha,\gamma,a,\tau)=\inf \mathrm {sp}(\mathcal L_{\underline\Ab_{\alpha, \gamma, a}}+V_{\underline{\mathbf{ B}}_{\alpha, \gamma, a},\tau})=\inf_{\substack {u\in  \mathcal{D}(\underline Q^{\tau}_{\alpha,\gamma, a})\\u\neq 0}}\frac {\underline Q^{\tau}_{\alpha,\gamma,a}(u)}{\|u\|_{L^2(\R^2_+)}^2}.
\end{equation}
By~\eqref{eq:F1}, we have 
\begin{equation}\label{eq:l3}
\lambda_{\alpha, \gamma,a} =\inf_\tau \underline\sigma(\alpha,\gamma,a,\tau).
\end{equation}
Hence, the study of $\lambda_{\alpha, \gamma,a}$ transforms to that of the associated band function $\tau\mapsto\underline\sigma(\alpha,\gamma,a,\tau)$. This study will be the subject of the next subsections.
\subsection{Case of a magnetic field parallel to the  $x_3-$axis}\label{sec:tangent}
We first treat the simple case when the magnetic field $\mathbf{B}_{ \alpha,\gamma,a} =(0,0,1)\mathsf s_{ \alpha,a}$ (i.e. when $\gamma=0$). In this case, the field is parallel to the  $x_3-$axis, thus $\underline{\mathbf{B}}_{\alpha, 0, a} = 0$. The operator $\mathcal L_{\underline\Ab_{\alpha,0,a}}+V_{\underline{\mathbf{ B}}_{\alpha,0,a},\tau}$ reduces to a simpler operator 
\[
	 \mathcal L_{\underline\Ab_{\alpha,0,a}}+V_{\underline{\mathbf{ B}}_{\alpha,0,a},\tau} = -(\nabla-i\underline \Ab_{\alpha,0,a})^2+\tau^2.
\]
For each $\tau\in\R$, the bottom of the spectrum of $\mathcal L_{\underline\Ab_{\alpha,0,a}}+V_{\underline{\mathbf{ B}}_{\alpha,0,a},\tau}$ equals
\[\underline\sigma(\alpha,0,a,\tau)=\mu(\alpha,a)+\tau^2,\]
where $\mu(\alpha,a)$ is the bottom of the spectrum of the operator $ \mathcal L_{\underline\Ab_{\alpha,0,a}}= -(\nabla + i \underline{\mathbf{A}}_{\alpha,0, a})^2$.
It immediately follows that 
\begin{equation}\label{eq: lambda mu}
\lambda_{\alpha, 0,a,\R^3_+}=\inf_\tau\underline\sigma(\alpha,0,a,\tau)=\underline\sigma(\alpha,0,a,0) =\mu(\alpha,a).\end{equation}
We present some properties of the operator $\mathcal L_{\underline\Ab_{\alpha,0, a}}$, that is of $\mathcal L_{\underline\Ab_{\alpha,0, a}}+V_{\underline{\mathbf{ B}}_{\alpha,0,a},0}$, obtained in~\cite[Section~3]{Assaad3}. 
 We denote by $\inf \mathrm{sp}_{ess}$ the infimum of the essential spectrum. From \cite[Theorem 3.1]{Assaad3}, we know that 
 \begin{equation}\label{eq: inf ess sp gamma0}
	\inf \mathrm{sp}_{ess}\mathcal L_{\underline\Ab_{\alpha,0, a}}=\inf \mathrm{sp}_{ess}(\mathcal L_{\underline\Ab_{\alpha,0,a}}+V_{\underline{\mathbf{ B}}_{\alpha,0,a},\tau=0})=|a|\Theta_0.
\end{equation}
As a consequence, if $\mu(\alpha, a) < |a|\Theta_0$  then $\mu(\alpha,a)$ is an eigenvalue of $\mathcal L_{\underline\Ab_{\alpha,0,a}}+V_{\underline{\mathbf{ B}}_{\alpha,0,a},0}$. The foregoing properties will be used in the proof of Theorem \ref{thm:main}, in the case $\gamma = 0$ (see Section \ref{sec:proof}).
%
\subsection{Case of a magnetic field non-parallel to  the $x_3-$axis}\label{sec: not tangent}
Now, we treat the case where  the magnetic field $\mathbf{B}_{ \alpha,\gamma,a}$ is not parallel to the $x_3-$ axis, that is the case when $\gamma\neq0$ ({see Figure \ref{fig:1}}). In this case, two auxiliary operators will be involved in the analysis. These operators are denoted by $H^{\mathrm{bnd}}_{\alpha,\gamma,a}[\tau]$ and $H^{\mathrm{stp}}_{\alpha, \gamma, a}[\tau]$ and are respectively defined on $\R^2_+$ and $\R^2$ with a constant (resp. piecewise constant) magnetic field. We refer to $H^{\mathrm{bnd}}_{\alpha,\gamma,a}[\tau]$ as the `boundary operator' since  it will be used in the proof of Proposition~\ref{prop:sp-lim} while studying the operator $\mathcal L_{\underline\Ab_{\alpha, \gamma, a}}+V_{\underline{\mathbf{ B}}_{\alpha,\gamma,a},\tau}$ near the boundary of $\R_+^2$ away from the discontinuity line $l_\alpha = P_\alpha \cap \R^2_+$ (see Figure~\ref{fig:1}). Similarly, we refer to $H^{\mathrm{stp}}_{\alpha,\gamma,a}[\tau]$ as the `step operator' since  it will be used in the study near the discontinuity line away from the boundary (see the proof of Proposition~\ref{prop:ess}). We introduce these operators in what follows.
\subsubsection{The boundary operator}
Let $\tau\in\R$. We define $H^{\mathrm{bnd}}_{\alpha,\gamma,a}[\tau]$ as the magnetic Neumann realization of the following self-adjoint operator on $\R^2_+$
\begin{equation}\label{eq:Hbnd}
H^{\mathrm {bnd}}_{\alpha, \gamma, a} [\tau] =-(\nabla-ia\Ab^{\mathrm{bnd}}_{\gamma})^2+[a(x_1\sin\alpha\sin\gamma-x_2\cos\alpha\sin\gamma)-\tau]^2,
\end{equation}
where $\Ab^{\mathrm{bnd}}_{\gamma}\in H^1_{\mathrm{loc}}(\R^2_+)$ is a magnetic potential with an associated constant magnetic field $\curl \Ab^{\mathrm{bnd}}_{\gamma}=\cos\gamma$.
This operator was studied in~\cite{popoff2013schrodinger} in the case $a=1$. Using translation, it was proven that the infimum of the spectrum of  $H^{\mathrm{bnd}}_{\alpha, \gamma, 1}[\tau]$ is independent of $\tau$. More precisely, in \cite[Lemma2.3]{popoff2013schrodinger} it is shown that  
\begin{equation}\label{eq:Hbnd1}
\inf\mathrm{sp}\big(H^{\mathrm{bnd}}_{\alpha,\gamma,1}[\tau]\big)=\zeta_{\nu_0},\qquad \forall \tau \in\R,
\end{equation}
where $\zeta_{\nu_0}$ is the value defined in~\eqref{eq:zeta} for $\nu_0=\arcsin(\sin\alpha\sin\gamma)$.
\begin{lem}[Bottom of the spectrum of the boundary operator]\label{lem: bott essHbnd} Let $a \in [-1,1)\setminus\{0\}$, $\alpha\in (0,\pi)$ and let $\gamma\in (0, \pi/2]$. Let $\tau\in \R$. It holds 
\begin{equation*}\label{eq:s1}
\inf\mathrm{sp} (H^{\mathrm{bnd}}_{\alpha, \gamma, a}[\tau])=|a|\zeta_{\nu_0}.
\end{equation*}
\end{lem}
\begin{proof} By a simple scaling argument, one can prove that 
\[\inf\mathrm{sp}( H^{\mathrm{bnd}}_{\alpha, \gamma, a}[\tau])=|a|\inf \mathrm{sp}(H^{\mathrm{bnd}}_{\alpha,\gamma,1}[\tau/a])\]. Combining this with \eqref{eq:Hbnd1} completes the proof.
\end{proof}
\vskip 3mm
\subsubsection{The step operator}
Let $\tau\in\R$. We define $H^{\mathrm{stp}}_{\alpha, \gamma, a}[\tau]$ as the following self-adjoint operator on $\R^2$
\begin{equation}\label{eq:Hedg}
H^{\mathrm{stp}}_{\alpha, \gamma, a}[\tau] =  -(\nabla-i\Ab^{\mathrm{stp}}_{\alpha, \gamma, a})^2+[(x_1\sin\alpha\sin\gamma-x_2\cos\alpha\sin\gamma)\mathsf s^{\mathrm{stp}}_{\alpha, a}-\tau]^2,
\end{equation}
where $\Ab^{\mathrm{stp}}_{\alpha, \gamma, a}\in H^1_{\mathrm{loc}}(\R^2)$ is such that $\curl \Ab^{\mathrm{stp}}_{\alpha, \gamma, a}=\mathsf s^{\mathrm{stp}}_{\alpha, a}\cos\gamma$, and $\mathsf s^{\mathrm{stp}}_{\alpha,a}$ is the following step function on $\R^2$
\[
	\mathsf s^{\mathrm{stp}}_{\alpha,a} := \mathbbm 1_{P_\alpha^+}+a\mathbbm 1_{P_\alpha^-},
\]
 with 
\begin{equation*}
	P_\alpha^+ := \{(x_1, x_2) \in \R^2 \, \vert \, x_1\sin\alpha-x_2\cos\alpha>0\},\end{equation*} \begin{equation*}P_\alpha^- := \{(x_1, x_2)\in \R^2 \, \vert \, x_1\sin\alpha-x_2\cos\alpha<0\}. \end{equation*}
On can see the magnetic field $\curl \Ab^{\mathrm{stp}}_{\alpha, \gamma, a}$ in $\R^2$ as the analogous of the magnetic field $\curl\underline \Ab_{ \alpha,\gamma,a}$ in $\R^2_+$, defined in~\eqref{eq:Ab}, with the sets $P_\alpha^+$ and 	$P_\alpha^-$ as the analogous of the sets $D^1_\alpha$ (in~\eqref{eq: def D1alpha R2}) and $D^2_\alpha$ (in~\eqref{eq: def D12alpha R2}) respectively.

The next lemma  determines the infimum of the spectrum of $H^{\mathrm{stp}}_{\alpha,\gamma, a}[\tau]$.
\begin{lem}[Bottom of the spectrum of the step operator]\label{lem: bott ess sp Hstp} Let $a \in [-1,1)\setminus\{0\}$, $\alpha\in (0,\pi)$ and let $\gamma\in (0, \pi/2]$. Let $\tau\in \R$. It holds 
\begin{equation*}
\inf\mathrm{sp}\big(H^{\mathrm{stp}}_{\alpha, \gamma, a}[\tau]\big)=\inf_{\xi\in \R} \big[\mu_a(\tau\sin\gamma+\xi\cos\gamma)+(\xi\sin\gamma-\tau\cos\gamma)^2\big],
\end{equation*}
where $\mu_a(\cdot)$ is the value defined in~\eqref{eq:mu_a}. 
\end{lem}
\begin{proof} For simplicity, we denote   $H^{\mathrm{stp}}_{\alpha, \gamma, a}[\tau]$, $\mathbf{A}^{\mathrm{stp}}_{\alpha, \gamma, a}$ and $\mathsf{s}^{\mathrm{stp}}_{\alpha, \gamma, a}$ by  $H^{\mathrm{stp}}$, $\mathbf{A}^{\mathrm{stp}}$ and $\mathsf{s}^{\mathrm{stp}}$ respectively.
 To estimate the bottom of the spectrum of $H^{\mathrm{stp}}$, we perform a rotation of angle $\alpha$ and  get that\footnote{We refer to~\cite[Sec.1]{popoff2013schrodinger} for rotation invariance principles.} the operator $H^{\mathrm{stp}}$ is unitarily equivalent to the following operator
\[\tilde H^{\mathrm{stp}}:=-(\nabla-i\tilde \Ab^{\mathrm{stp}})^2+(x_2\sin\gamma\, \tilde{\mathsf s}^{\mathrm{stp}}+\tau)^2\]
defined on $\R^2$, with $\curl\tilde\Ab^{\mathrm{stp}}=\tilde{\mathsf s}^{\mathrm{stp}}\cos\gamma$ and $\tilde{\mathsf s}^{\mathrm{stp}} :=\mathbbm 1_{\{x_2<0\}}+a\mathbbm 1_{\{x_2>0\}}$.
{Thus, we get}
\begin{equation}\label{eq: inf spec rot}
	{\inf \mathrm{sp}(H^{\mathrm{stp}}) = \inf \mathrm{sp}(\tilde H^{\mathrm{stp}})}.
\end{equation}
Performing a suitable change of gauge, we choose $\tilde\Ab^{\mathrm{stp}}=-(x_2\cos\gamma\,\tilde{\mathsf s}^{\mathrm{stp}} ,0)$. Then, we write the expression of $\tilde H^{\mathrm{stp}}$ explicitly as
\begin{equation*}\label{eq:rot}
\tilde{H}^{\mathrm{stp}}=-(\partial_{x_1}+i x_2\cos\gamma\, \tilde{\mathsf{s}}^{\mathrm{stp}})^2-\partial^2_{x_2} +(x_2\sin\gamma\,\tilde{\mathsf s}^{\mathrm{stp}} +\tau)^2.
\end{equation*}
By a Fourier transform in the $x_1$ variable, we get
\begin{equation}\label{eq: fiber op}
	\tilde{H}^{\mathrm{stp}}=\int^\oplus_{\xi\in\R}\Big(-\partial^2_{x_2}+(\xi+x_2\cos\gamma\,\tilde{\mathsf s}^{\mathrm{stp}} )^2+(x_2\sin\gamma\, \tilde{\mathsf s}^{\mathrm{stp}} +\tau)^2\Big)\,d\xi,
\end{equation}
where  $-\partial^2_{x_2}+(\xi+x_2\cos\gamma\, \tilde{\mathsf s}^{\mathrm{stp}} )^2+(x_2\sin\gamma\,\tilde{\mathsf s}^{\mathrm{stp}} +\tau)^2$ is a self-adjoint fiber operator on $\R$. Hence 
\[
	\inf\mathrm{sp}\big(\tilde H^{\mathrm{stp}}\big)=\inf_\xi\big[\inf\mathrm{sp}\big(-\partial^2_{x_2}+(\xi+x_2\cos\gamma\,\tilde{\mathsf s}^{\mathrm{stp}} )^2+(x_2\sin\gamma\, \tilde{\mathsf s}^{\mathrm{stp}} +\tau)^2\big)\big].
\]
We can now rewrite
\[(\xi+\mathsf s^{\mathrm{rot}}_{\mathrm{s}}x_2\cos\gamma )^2+(\mathsf s^{\mathrm{rot}}_{\mathrm{s}} x_2\sin\gamma+\tau)^2=(\mathsf s^{\mathrm{rot}}_{\mathrm{s}}x_2+\tau\sin\gamma+\xi\cos\gamma)^2+(\xi\sin\gamma-\tau\cos\gamma)^2.\]
Then using that $\tilde{\mathsf{s}}^{\mathrm{stp}} =  \mathbbm 1_{\{x_2<0\}}+a\mathbbm 1_{\{x_2>0\}}$, the fiber operator in \eqref{eq: fiber op} is unitary equivalent to the operator given by
\[\mathfrak h_a[\tau\sin\gamma+\xi\cos\gamma]+(\xi\sin\gamma-\tau\cos\gamma)^2,\]
where $\mathfrak h_a[\cdot]$ is the operator defined in~\eqref{eq:ha}. Thus,
\begin{equation}\label{eq: spect rot to h xi}
	\inf\mathrm{sp}\big(\tilde H^{\mathrm{stp}}\big)= \inf\mathrm{sp}\big(\mathfrak h_a[\tau\sin\gamma+\xi\cos\gamma]+(\xi\sin\gamma-\tau\cos\gamma)^2\big).
\end{equation}
 Moreover, we have  
\begin{multline}\label{eq: inf sp ha}
\inf\mathrm{sp}\big(\mathfrak h_a[\tau\sin\gamma+\xi\cos\gamma]+(\xi\sin\gamma-\tau\cos\gamma)^2\big)
\\
= \inf\mathrm{sp}\big(\mathfrak h_a[\tau\sin\gamma+\xi\cos\gamma]\big)+(\xi\sin\gamma-\tau\cos\gamma)^2
\\
=\mu_a(\tau\sin\gamma+\xi\cos\gamma)+(\xi\sin\gamma-\tau\cos\gamma)^2,
\end{multline}
where $\mu_a(\cdot)$ is the bottom of the spectrum of $ h_a[\tau\sin\gamma+\xi\cos\gamma]$ (see ~\eqref{eq:mu_a}). 
Gathering {\eqref{eq: inf spec rot}, \eqref{eq: spect rot to h xi} and \eqref{eq: inf sp ha}} completes the proof.
\end{proof}
\subsubsection{Bottom of the essential spectrum of the 2D reduced operator}
In this section, we determine the infimum of the essential spectrum of the $2$D operators $\mathcal L_{\underline\Ab_{\alpha,\gamma,a}}+V_{\underline{\mathbf{ B}}_{\alpha,\gamma,a},\tau}$ introduced in Section \ref{sec: red 2D op}. For each $a\in [-1,1)\setminus \{0\}$, $\gamma\in (0,\pi/2]$, $\alpha\in (0,\pi)$ and $\tau\in \R$, let
\begin{equation}\label{eq: sigma ess 2d}
	\underline\sigma_{ess}(\alpha,\gamma,a,\tau) := \inf \mathrm{sp}_{ess}(\mathcal L_{\underline\Ab_{\alpha,\gamma,a}}+V_{\underline{\mathbf{B}}_{\alpha,\gamma,a},\tau} ).
\end{equation}
Knowing this infimum will be useful in determining values of $(\alpha,\gamma,a,\tau)$ where the bottom of the spectrum $\underline\sigma(\alpha,\gamma,a,\tau)$ of these operators is an eigenvalue. This will be used in establishing Theorem~\ref{thm:main} later. The next proposition is the main result of this section.
\begin{proposition}[Characterization of $\underline\sigma_{ess}(\alpha,\gamma,a,\tau)$]\label{prop:ess}
	 Let $a\in[-1,1)\setminus\{0\}$,  $\alpha\in(0,\pi)$, $\gamma\in(0,\pi/2]$ and $\tau\in \R$. Let $\underline\sigma_{ess}(\alpha,\gamma,a,\tau)$ be as in \eqref{eq: sigma ess 2d}, we have
	\begin{equation*}\label{eq:ess}
	 \underline\sigma_{ess}(\alpha,\gamma,a,\tau)=\inf_{\xi\in \R}\big(\mu_a(\tau\sin\gamma+\xi\cos\gamma)+(\xi\sin\gamma-\tau\cos\gamma)^2\big),\end{equation*}
	where $\mu_a(\cdot)$ is  the value defined in~\eqref{eq:mu_a}.
	\end{proposition}
For the proof of Proposition \ref{prop:ess} we need the following lemma.
\begin{lemma}\label{lem:persson}
	Let $a\in [-1,1)\setminus \{0\}$, $\alpha\in (0,\pi)$, $\gamma\in (0,\pi/2]$ and let $\tau \in \R$. Let $\underline\sigma_{ess}(\alpha,\gamma,a,\tau)$ be as in \eqref{eq: sigma ess 2d}. It holds
\[ 
	\underline\sigma_{ess}(\alpha,\gamma,a,\tau)=\lim_{R\rightarrow+\infty}\Sigma(\mathcal L_{\underline\Ab_{\alpha,\gamma,a}}+V_{\underline{\mathbf{ B}}_{\alpha,\gamma,a},\tau},R),
\]
with
\[
	\Sigma(\mathcal L_{\underline\Ab_{\alpha,\gamma,a}}+V_{\underline{\mathbf{ B}}_{\alpha,\gamma,a},\tau},R) := \inf_{\substack{u\in  C_0^\infty(\overline {\R^2_+}\cap\complement \mathcal B_R)\\ u \neq 0}}\frac {\underline Q^{\tau}_{\alpha,\gamma, a}(u)}{\|u\|^2_{L^2(\R_+^2)}},
\]
	where	$\mathcal B_R$  is a ball of radius $R$ centered at the origin, $\complement \mathcal B_R$ is its complement in $\R^2$, and $\underline Q^{\tau}_{\alpha,\gamma,a}$ is the quadratic form defined in~\eqref{eq:Qa}.
	\end{lemma}
Lemma \ref{lem:persson} is a well-known Persson-type result, useful to characterize the bottom of essential spectra. We refer the reader to~\cite{persson1960bounds,popoff2013schrodinger,agmon2014lectures} for this type of results, and~\cite[Appendix~A]{Assaad3} for a detailed proof in similar situations.
Moreover, in the proof of Proposition \ref{prop:ess}, we shall see the importance of determining where  the electric potential $V_{\underline{\mathbf{ B}}_{\alpha,\gamma, a},\tau}$ attains its infimum and where it is big. To that end, we define the set 
\begin{equation}\label{eq:S}
\Upsilon_{\alpha, \gamma, a,\tau} = \left\{x\in\overline {\R^2_+}~:~V_{\underline{\mathbf{ B}}_{\alpha,\gamma,a},\tau}(x)=\inf_{y\in\R^2_+}V_{\underline{\mathbf{ B}}_{\alpha,\gamma,a},\tau}(y)\right\}.
\end{equation}
We note that $\Upsilon_{\alpha, \gamma, a,\tau}$  is not necessary $V_{\underline{\mathbf{ B}}_{\alpha, \gamma, a},\tau}^{-1}(\{0\})$; determining this set  depends on the values of $a\in[-1,1)\setminus\{0\}$ and $\tau\in\R$, as shown in what follows. We recall that 
$V_{\underline{\mathbf{ B}}_{\alpha,\gamma,a},\tau}$ is defined for $x=(x_1, x_2)\in \R^2_+$ as
	\begin{eqnarray*} V_{\underline{\mathbf{ B}}_{\alpha, \gamma, a},\tau}(x)&=&\big(x_1\underline b_2-x_2\underline b_1-\tau\big)^2,
		\\
		&=& \nonumber} 	{\begin{cases}  (x_1\sin\alpha\sin\gamma - x_2 \cos\alpha\sin\gamma - \tau)^2 &\mbox{for }\, \, (x_1, x_2)\in D^1_\alpha, \\ [a(x_1\sin\alpha\sin\gamma - x_2 \cos\alpha\sin\gamma )- \tau]^2 &\mbox{for }\,\, (x_1, x_2)\in D^2_\alpha,\end{cases}
	\end{eqnarray*}
where $D_\alpha^1$ and $D_\alpha^2$ are as in~\eqref{eq: def D1alpha R2} and~\eqref{eq: def D12alpha R2}.
We now define, for $x = (x_1, x_2)\in \R^2$,
\begin{align*}
V^{(1)}_{\underline{\mathbf{ B}}_{\alpha, \gamma},\tau}(x)&=(x_1\sin\alpha\sin\gamma-x_2\cos\alpha\sin\gamma-\tau)^2\\
V^{(2)}_{\underline{\mathbf{ B}}_{\alpha, \gamma, a},\tau}(x)&=[a(x_1\sin\alpha\sin\gamma-x_2\cos\alpha\sin\gamma)-\tau]^2
\end{align*}
and  the following subsets of $\R^2$
\begin{align}\label{eq:S12}
\Upsilon^{(1)}_{\alpha, \gamma,  \tau}&=\big(V^{(1)}_{\underline{\mathbf{ B}}_{\alpha, \gamma},\tau}\big)^{-1}(\{0\})=\left\{(x_1,x_2)\in\R^2~:~x_1\sin\alpha-x_2\cos\alpha=\frac \tau{\sin\gamma}\right\},\nonumber\\
\Upsilon^{(2)}_{\alpha, \gamma, a, \tau}& =\big(V^{(2)}_{\underline{\mathbf{ B}}_{\alpha, \gamma, a},\tau}\big)^{-1}(\{0\})=\left\{(x_1,x_2)\in\R^2~:~x_1\sin\alpha-x_2\cos\alpha=\frac \tau{a\sin\gamma}\right\}.
\end{align}
Note that $\Upsilon^{(1)}_{\alpha, \gamma, \tau}$ and $\Upsilon^{(2)}_{\alpha, \gamma, a, \tau}$ are two lines parallel to the discontinuity line $l_\alpha$ of equation $x_1\sin\alpha-x_2\cos\alpha=0$ . Moreover, for $x\in\R^2$
\begin{equation}\label{eq:V12}
V^{(1)}_{\underline{\mathbf{ B}}_{\alpha, \gamma},\tau}(x)=\sin^2\gamma\dist^2(x,\Upsilon^{(1)}_{\alpha, \gamma, \tau})\qquad V^{(2)}_{\underline{\mathbf{ B}}_{\alpha, \gamma, a},\tau}(x)=a^2\sin^2\gamma\dist^2(x,\Upsilon^{(2)}_{\alpha, \gamma, a,\tau}).
\end{equation}

We keep denoting by $\Upsilon^{(1)}_{\alpha, \gamma, \tau}$ (resp. $\Upsilon^{(2)}_{\alpha, \gamma, a,\tau}$) the intersection between $\overline{\R_+^2}$ and $\Upsilon^{(1)}_{\alpha, \gamma, \tau}$ (resp. $\Upsilon^{(2)}_{\alpha, \gamma, a,\tau}$). 
\begin{lem}[The set $\Upsilon_{\alpha, \gamma, a,\tau}$]\label{lem:upsilon}
	Let $a\in [-1, 1)\setminus \{0\}$, $\alpha\in (0,\pi)$, and $\gamma\in (0,\pi/2]$. Let $\Upsilon_{\alpha, \gamma, a,\tau}\subset\R^2$ be the set defined in \eqref{eq:S}. It holds 
\[
\Upsilon_{\alpha, \gamma, a,\tau} = \begin{cases}
	l_\alpha &\mbox{if}\,\,\, a\in [-1,0), \tau < 0 \vspace{0.1cm}
	\\
	\Upsilon^{(1)}_{\alpha, \gamma, \tau}\cup \Upsilon^{(2)}_{\alpha, \gamma, a,\tau} &\mbox{if}\,\,\, a\in [-1,0), \tau \geq 0 \vspace{0.1cm}
	\\
	\Upsilon^{(1)}_{\alpha, \gamma, \tau} &\mbox{if}\,\,\, a\in (0,1), \tau \geq 0 \vspace{0.1cm}
	\\
	\Upsilon^{(2)}_{\alpha, \gamma, a,\tau} &\mbox{if}\,\,\, a\in (0,1), \tau < 0.
\end{cases}
\]
\end{lem}
Indeed (see Figures~\eqref{fig:2} and~\eqref{fig:3}),
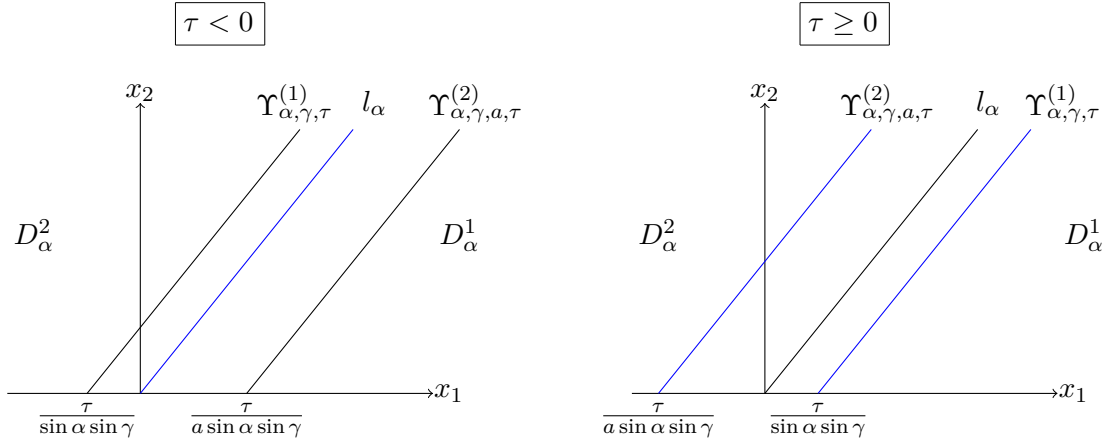
\begin{figure}
	\centering
	\begin{tikzpicture}[scale= 0.7]
				\draw[->] (0,0) to (0,5.5);
				\draw[->] (-2.5,0) to (5.5,0);
				\draw[blue] (0,0) to (4,5);
				\node at (4.4,5.5) {$l_\alpha$};
				\draw (-1,0) to (3,5);
				\node at (2.9, 5.5) {$\Upsilon^{(1)}_{\alpha, \gamma, \tau}$};
				\node at (-1, -0.5) {$\frac{\tau}{\sin\alpha\sin\gamma}$};
				\node at (6.3, 5.5) {$\Upsilon^{(2)}_{\alpha, \gamma, a,\tau}$};
				\draw (2,0) to (6,5);
				\node at (2,-0.5) {$\frac{\tau}{a\sin\alpha\sin\gamma}$};
				\node at (-2,3) {{$D^2_\alpha$}};
				\node at (6,3) {{$D^1_\alpha$}};
				\node at (0,5.7) {{$x_2$}};
				\node at (5.8, 0) {{$x_1$}};
				\node at (1.5,7) {$\boxed{\textcolor{white}{\sum}\hspace{-0.4cm}\tau <0\hspace{0.15cm}}$};
			\end{tikzpicture}
			\hspace{0.5cm}
		\begin{tikzpicture}[scale= 0.7]
			\draw[->] (0,0) to (0,5.5);
			\draw[->] (-2.5,0) to (5.5,0);
			\draw (0,0) to (4,5);
			\node at (4.2,5.5) {$l_\alpha$};
			\draw[blue] (1,0) to (5,5);
			\node at (5.6, 5.5) {$\Upsilon^{(1)}_{\alpha, \gamma, \tau}$};
			\node at (1, -0.5) {$\frac{\tau}{\sin\alpha\sin\gamma}$};
			\node at (2.3, 5.5) {$\Upsilon^{(2)}_{\alpha, \gamma, a,\tau}$};
			\draw[blue] (-2,0) to (2,5);
			\node at (-2,-0.5) {$\frac{\tau}{a\sin\alpha\sin\gamma}$};
			\node at (-2,3) {{$D^2_\alpha$}};
			\node at (6,3) {{$D^1_\alpha$}};
			\node at (0,5.7) {{$x_2$}};
			\node at (5.8, 0) {{$x_1$}};
			\node at (1.5,7) {$\boxed{\textcolor{white}{\sum}\hspace{-0.4cm}\tau \geq 0\hspace{0.15cm}}$};
		\end{tikzpicture}		
	\caption{ For $\alpha\in(0,\pi)$, $
		\gamma\in(0,\pi/2]$ and $a\in[-1,0)$, the set $\Upsilon_{\alpha, \gamma, a,\tau}$ is drawn in blue. For $\tau\geq0$ (at right), $\Upsilon_{\alpha, \gamma,a,\tau}=\Upsilon^{(1)}_{\alpha, \gamma, \tau}\cup\Upsilon^{(2)}_{\alpha, \gamma, a,\tau}$. For $\tau<0$ (at left), $\Upsilon_{\alpha, \gamma, a,\tau}=l_\alpha$.}
	\label{fig:2}
\end{figure}
\begin{figure}
	\centering\begin{tikzpicture}[scale= 0.7]
					\draw[->] (0,0) to (0,5.5);
					\draw[->] (-4.5,0) to (4.8,0);
					\draw (0,0) to (4,5);
					\node at (4.5,5.5) {$l_\alpha$};
					\draw[blue] (-3,0) to (1,5);
					\node at (1.2, 5.5) {$\Upsilon^{(2)}_{\alpha, \gamma, a,\tau}$};
					\node at (-3.2, -0.5) {$\frac{\tau}{a\sin\alpha\sin\gamma,\tau}$};
					\node at (3.1, 5.5) {$\Upsilon^{(1)}_{\alpha, \gamma, \tau}$};
					\draw (-1,0) to (3,5);
					\node at (-0.8,-0.5) {$\frac{\tau}{\sin\alpha\sin\gamma}$};
					\node at (-2,3) {{$D^2_\alpha$}};
					\node at (4.5,3) {{$D^1_\alpha$}};
					\node at (0,5.7) {{$x_2$}};
					\node at (5.1, 0) {{$x_1$}};
					\node at (0.5,7) {$\boxed{\textcolor{white}{\sum}\hspace{-0.4cm}\tau <0\hspace{0.15cm}}$};
				\end{tikzpicture}
				\hspace{0.3cm}
			\begin{tikzpicture}[scale= 0.7]
				\draw[->] (0,0) to (0,5.5);
				\draw[->] (-1,0) to (7.3,0);
				\draw  (0,0) to (4,5);
				\node at (4,5.5) {$l_\alpha$};
				\draw [blue](1,0) to (5,5);
				\node at (5.3, 5.5) {$\Upsilon^{(1)}_{\alpha, \gamma, \tau}$};
				\node at (0.8, -0.5) {$\frac{\tau}{\sin\alpha\sin\gamma}$};
				\node at (7.3, 5.5) {$\Upsilon^{(2)}_{\alpha, \gamma, a,\tau}$};
				\draw (3,0) to (7,5);
				\node at (3.2,-0.5) {$\frac{\tau}{a\sin\alpha\sin\gamma}$};
				\node at (-0.5,3) {{$D^2_\alpha$}};
				\node at (7,3) {{$D^1_\alpha$}};
				\node at (0,5.7) {{$x_2$}};
				\node at (7.6, 0) {{$x_1$}};
				\node at (3.7,7) {$\boxed{\textcolor{white}{\sum}\hspace{-0.4cm}\tau \geq 0\hspace{0.15cm}}$};
			\end{tikzpicture}		
	\caption{ For $\alpha\in(0,\pi)$, $
		\gamma\in(0,\pi/2]$ and $a\in(0,1)$, the set $\Upsilon_{\alpha, \gamma, a,\tau}$ is drawn in blue. For $\tau\geq0$ (at right), $\Upsilon_{\alpha, \gamma, a,\tau}=\Upsilon^{(1)}_{\alpha, \gamma, \tau}$. For $\tau<0$ (at left), $\Upsilon_{\alpha, \gamma, a,\tau}=\Upsilon^{(2)}_{\alpha, \gamma, a,\tau}$.}
	\label{fig:3}
\end{figure}
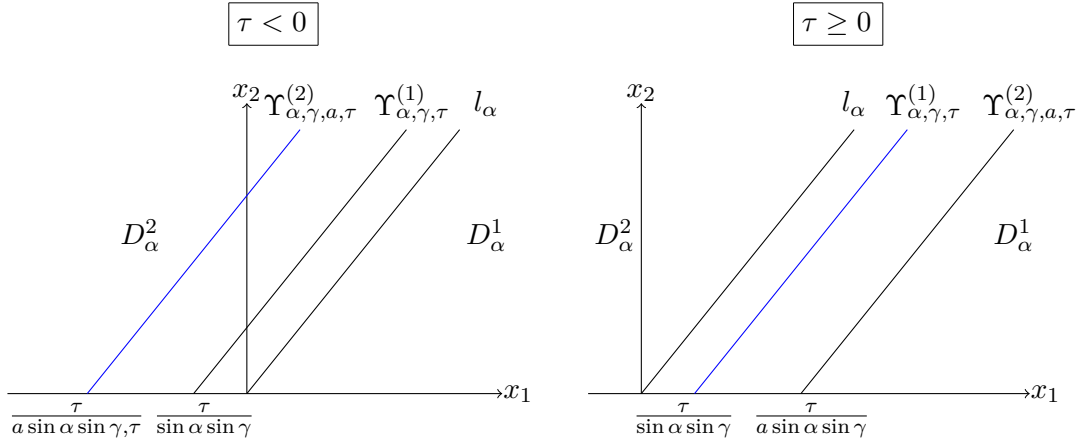
\begin{itemize}
	\item \emph{Case $a\in[-1,0)$ and $\tau < 0$.}
	One observes that $V_{\underline{\mathbf{ B}}_{{\alpha, \gamma, a}},\tau}^{-1}(\{0\})=\emptyset$. In this case, 
	\begin{equation}\label{eq:S16}
		\Upsilon_{\alpha, \gamma, a,\tau}=l_\alpha\quad \mathrm{and}\quad \inf V_{\underline{\mathbf{ B}}_{\alpha, \gamma, a},\tau}=\tau^2.
	\end{equation}
	\item \emph {Case $a\in[-1,0)$ and $\tau\geq0$}.
	Here, 
	\begin{equation*}\label{eq:S13}
		\Upsilon_{\alpha, \gamma, a,\tau}=V_{\underline{\mathbf{ B}}_{\alpha, \gamma, a},\tau}^{-1}(\{0\})=\Upsilon^{(1)}_{\alpha, \gamma, \tau}\cup\Upsilon^{(2)}_{\alpha, \gamma, a,\tau}.
	\end{equation*}
	Note that this set is $l_\alpha$ for $\tau=0$.
	\item \emph{Case $a\in(0,1)$ and $\tau < 0$}.
	In this case,  
	\begin{equation*}\label{eq:S15}
		\Upsilon_{\alpha, \gamma, a,\tau}=V_{\underline{\mathbf{ B}}_{\alpha, \gamma, a},\tau}^{-1}(\{0\})=\Upsilon^{(2)}_{\alpha, \gamma, a,\tau}.
	\end{equation*}
	\item \emph{Case $a\in(0,1)$ and $\tau \geq 0$.}
	We have
	\begin{equation*}\label{eq:S14}
		\Upsilon_{\alpha, \gamma, a,\tau}=V_{\underline{\mathbf{ B}}_{\alpha, \gamma, a},\tau}^{-1}(\{0\})=\Upsilon^{(1)}_{\alpha, \gamma, \tau}.
	\end{equation*}
	Again, this set is  $l_\alpha$ for $\tau=0$.
\end{itemize}
Now, we prove Proposition \ref{prop:ess}. 
\begin{proof}[Proof of Proposition \ref{prop:ess}]
	The idea of the proof is similar to that in~\cite[Proposition~3.2]{popoff2013schrodinger} (see also~\cite[Lemma~3.7]{Assaad3}). However, one has to take into consideration the particular properties of the electric potential discussed above, which are induced by the discontinuity of our magnetic field.
	
	The `step operator' $H^{\mathrm{stp}}_{\alpha,\gamma,a}[\tau]$ (in~\eqref{eq:Hedg}) and the quadratic form $\underline{Q}^\tau_{\alpha,\gamma,a}$ (in~\eqref{eq:Qa}) are used frequently in the proof. Throughout the proof, we simplify their notation and denote them respectively by  $H^{\mathrm{stp}}$ and $\underline{Q}$. 
	
In light of Lemma~\ref{lem:persson}, it suffices to prove that
	\begin{equation}\label{eq:ess1}
	\lim_{R\rightarrow+\infty}\Sigma(\mathcal L_{\underline\Ab_{\alpha,\gamma,a}}+V_{\underline{\mathbf{ B}}_{\alpha,\gamma,a},\tau},R)=\inf_{\xi\in \R}\big(\mu_a(\tau\sin\gamma+\xi\cos\gamma)+(\xi\sin\gamma-\tau\cos\gamma)^2\big).
	\end{equation}
	We establish separately an upper bound and a lower bound for the limit above.

	 \paragraph{\underline{\emph {Upper bound}}} Let $\epsilon>0$ and $R>0$. Considering the operator $H^{\mathrm{stp}}$, the min-max principle ensures the existence of a  normalized function $u_\epsilon \in C_0^\infty(\R^2)\setminus\{0\}$ such that 
	\begin{eqnarray}\label{eq:up1}
	 \langle H^{\mathrm{stp}}u_\epsilon,u_\epsilon\rangle &<& \inf\mathrm{sp}(H^{\mathrm{stp}})+\epsilon
	 \\
	 &=&\inf_{\xi\in \R}\big(\mu_a(\tau\sin\gamma+\xi\cos\gamma)+(\xi\sin\gamma-\tau\cos\gamma)^2\big)+\epsilon,\nonumber
	 \end{eqnarray}
	 where the last equality follows from Lemma \ref{lem: bott ess sp Hstp}.
	 Let the function  $u_{\epsilon,\mathbf{r}}$ be the translation of $u_\epsilon$ by a vector $\mathbf{r}$, i.e., $u_{\epsilon, \mathbf{r}}(x) = u_\epsilon (x-\mathbf{r})$ for $x\in\R^2$, where  $\mathbf{r}$ is an upward direction vector of the discontinuity line $l_\alpha:x_1\sin\alpha-x_2\cos\alpha=0$. We have $u_{\epsilon,\mathbf{r}}\in  C_0^\infty(\R^2){\setminus\{0\}}$. Moreover, 
	 there exists $r_0>0$ such that the function $u_{\epsilon,\mathbf{r}}$ is supported in $\R_+^2\cap\complement \mathcal B_R$, for $|\mathbf{r}|>r_0$. Using that $H^{\mathrm{stp}}$ is invariant by translation in the $l_\alpha$ direction (see~\eqref{eq:Hedg}), we get 
	 \begin{equation}\label{eq:up2}
\langle H^{\mathrm{stp}}u_{\epsilon,\mathbf{r}},u_{\epsilon,\mathbf{r}}\rangle=\langle H^{\mathrm{stp}}u_\epsilon,u_\epsilon\rangle.
	 \end{equation}
Combining \eqref{eq:up1} and \eqref{eq:up2} gives
 \[
 	\langle H^{\mathrm{stp}}u_{\epsilon,\mathbf{r}},u_{\epsilon,\mathbf{r}}\rangle<\inf_{\xi\in \R}\big(\mu_a(\tau\sin\gamma+\xi\cos\gamma)+(\xi\sin\gamma-\tau\cos\gamma)^2\big)+\epsilon.
 \]
 Now, using the support properties of $u_{\epsilon, \mathbf{r}}$, a direct calculation shows that $\langle H^{\mathrm{stp}}(u_{\epsilon,\mathbf{r}}),u_{\epsilon,\mathbf{r}}\rangle = \underline Q(u_{\epsilon,\mathbf{r}})$. Then,
 \[
  \underline Q(u_{\epsilon,\mathbf{r}})<\inf_{\xi\in \R}\big(\mu_a(\tau\sin\gamma+\xi\cos\gamma)+(\xi\sin\gamma-\tau\cos\gamma)^2\big)+\epsilon.
 \]
	Having $u_{\epsilon,\mathbf{r}}$ a non-zero normalized function in $C_0^\infty(\overline\R_+^2\cap\complement \mathcal B_R)$, we have
	\begin{eqnarray*}
		\Sigma(\mathcal L_{\underline\Ab_{\alpha,\gamma,a}}+V_{\underline{\mathbf{ B}}_{\alpha,\gamma,a},\tau},R)&=& {\inf_{\substack{u\in  C_0^\infty(\overline {\R^2_+}\cap\complement \mathcal B_R)\\ u \neq 0}}\frac {\underline Q(u)}{\|u\|^2_{L^2(\R_+^2)}}}
		\\
		&<&\inf_{\xi\in \R}\big(\mu_a(\tau\sin\gamma+\xi\cos\gamma)+(\xi\sin\gamma-\tau\cos\gamma)^2\big)+\epsilon.\nonumber
	\end{eqnarray*}
	Taking first $\epsilon\rightarrow 0$ and then $R\rightarrow+\infty$, we get the upper bound in~\eqref{eq:ess1}.

	 \paragraph{\underline{\emph {Lower bound}}.} 
	Let $(\rho,\theta)$ be the polar coordinates in $\R^2$. We consider a partition of unity $(\chi_j^{\mathrm{pol}})_{j=1,2,3}\subset C^\infty(\overline\R_+\times[0,\pi])$ such that: for $j\in\{1,2,3\}$, $0\leq \chi_j^{\mathrm{pol}}\leq1$  and $\forall (\rho,\theta) \in \R_+\times(0,\pi)$, {$\chi_j^{\mathrm{pol}}(\rho,\theta)=\chi_j^{\mathrm{pol}}(1,\theta)$ and} 
		\begin{eqnarray*}\label{eq:chi}
		  \chi^{\mathrm{pol}}_1 (\rho,\theta)&=&1\ \mathrm{for}\ \theta\in\left(0,\frac 18\alpha\right],
		  \\
		  \chi^{\mathrm{pol}}_2 (\rho,\theta)&=&1\ \mathrm{for}\ \theta\in\left[\frac 14 \alpha,\frac 14 \alpha+\frac {3\pi} 4\right],\\ 
		  \chi^{\mathrm{pol}}_3(\rho,\theta) &=&1\ \mathrm{for}\ \theta\in\left[\frac 18 \alpha+\frac {7\pi} 8,\pi\right).	
			\end{eqnarray*}
		{Moreover, $\sum_{j=1}^3 |\chi^{\mathrm{pol}}_j|^2=1$ and  $\sum_{j=1}^3 |(\chi^{\mathrm{pol}}_j)'|^2\leq C$, where $C$ is a constant dependent on $\alpha$ but independent of $a$.} 
	Let $(\chi_j)_{j=1,\cdots,3}$ be the associated functions in Cartesian coordinates
	\[\chi_j(x_1,x_2)=\chi_j^{\mathrm{pol}}(\rho,\theta),\qquad (x_1,x_2)\in\R^2_+.\]
	For $R>0$ and  $u\in  C_0^\infty(\overline {\R^2_+}\cap\complement \mathcal B_R)$, we use the IMS formula to write (see~\cite[Theorem~3.2]{cycon2009schrodinger})
	\begin{equation}\label{eq:ims}
	\underline Q(u) =\sum_{j=1}^3\underline Q(\chi_ju)-
	\sum_{j=1}^3\|u |\nabla \chi_j| \|^2_{L^2(\R^2_+)}.
	\end{equation}
	We start by bounding the error term $\sum_{j=1}^3\|u |\nabla \chi_j| \|^2_{L^2(\R^2_+)}$. For $x=(x_1,x_2)\in \R_+^2$, we have
	\[|\nabla_x\chi_j(x_1,x_2)|^2=|\partial_r\chi_j^{\mathrm{pol}}(\rho,\theta)|^2+\frac 1{r^2}|\partial_\theta\chi_j^{\mathrm{pol}}(\rho,\theta)|^2=\frac 1{r^2}|\partial_\theta\chi_j^{\mathrm{pol}}(\rho,\theta)|^2,\]
	where the last equality follows from the fact that $\chi_j^{\mathrm{pol}}$ is constant in the radial coordinate.
	Thus, using $\sum_{j=1}^3 |(\chi^{\mathrm{pol}}_j)'|^2\leq C$ and  that $u$ is supported outside $\mathcal{B}_R$, we get
	\begin{equation*}\label{eq:er1}
	\sum_{j=1}^3\|u |\nabla \chi_j| \|^2_{L^2(\R^2_+)}\leq\frac C{R^2}\|u\|^2_{L^2(\R^2_+)}. 
   \end{equation*}
   Next, we consider the main term, $\sum_{j=1}^3\underline Q(\chi_ju)$, in~\eqref{eq:ims}. We start by bounding $\underline{Q}(\chi_2u)$. Extending $\chi_2 u$ by zero over $\R^2$, we get that $\chi_2 u$ is in the domain of the operator $H^{\mathrm{stp}}$. By Lemma \ref{lem: bott ess sp Hstp}, we notice that
   \begin{equation}\label{eq:ess2}
   \underline Q(\chi_2u)=\langle H^{\mathrm{stp}}(\chi_2u),\chi_2u\rangle\geq \inf_{\xi \in \R}\big(\mu_a(\tau\sin\gamma+\xi\cos\gamma)+(\xi\sin\gamma-\tau\cos\gamma)^2\big)\|\chi_2u\|^2.
   \end{equation}
   Now, we bound $\underline Q(\chi_ju)$ for $j=1,3$. Here, we recall the sets $\Upsilon^{(1)}_{\alpha, \gamma, \tau}$ and $\Upsilon^{(2)}_{\alpha, \gamma, a,\tau}$ defined in~\eqref{eq:S12}. We choose a large $R_0>0$ and assume w.l.o.g that $\alpha\in(0,\pi/2)$, then  an elementary computation yields for $R>R_0$ (see Figure \ref{fig:4}):
   \[\dist\left(\supp \chi_1 u,\Upsilon^{(1)}_{\alpha, \gamma, \tau}\right)\geq\frac 1{|\sin\gamma|}\left|R\sin\left(\frac{3\alpha}{4}\right)\sin\gamma-\tau\right|\]
   and \[ \dist\left(\supp \chi_3 u,\Upsilon^{(2)}_{\alpha, \gamma, a,\tau}\right)\geq\frac 1{|a\sin\gamma|}|aR\sin\alpha\sin\gamma+\tau|.\]
  \begin{figure}
	\centering
			\begin{tikzpicture}[scale= 0.9]
				\draw[->] (-6.3,0) to (6.3,0);
				\node at (6.8, 0) {$x_1$};
				\node at (0, 5.7) {$x_2$};
				\draw[fill=gray, opacity=0.6] (6.3,0) -- (6.3,2.5) -- (0,0);
				\draw[fill=gray, opacity=0.6] (-6.3,0) -- (-6.3,1.2) -- (0,0);
				\draw[fill=white]([shift=(30:1cm)]3.15,-0.5) arc (0:180:4cm);
				\draw(0,0) to (4,5);
				\draw[->] (0,0) to (0,5.5);
				\node at (4.4,5.5) {$l_\alpha$};
			\draw (1,0) to (5,5);
			\node at (5.3, 5.4) {$\Upsilon^{(1)}_{\alpha,\gamma,\tau}$};
			\node at (2.3, 5.4) {$\Upsilon^{(2)}_{\alpha,\gamma,a,\tau}$};
			\draw (-4,-5/2) to (2,5);
			\node at (-4,0) {$\bullet$};
			\draw (0,0) to (-6.3,1.2);
			\node at (-6.5, 1.7) {$\theta  = \frac{1}{8}\alpha + \frac{7}{8}\pi$};
			\draw (0,0) to (6.3,2.5);
			\node at (6.5, 2.7) {$\theta  = \frac{\alpha}{4}$};
			\node at (5.3, 1) {$\mathrm{supp}\chi_1 u$};
			\node at (-5.3, 0.5) {$\mathrm{supp}\chi_3 u$};
			\node at (3.75, 1.47) {$\bullet$};
			\draw[dotted] (3.7, 1.5) to (2.83, 2.34);
			\node at (3, 1.7) {$d_1$};
			\draw[dotted] (-4,0) to (-2.75, -1);
			\node at (-3, -0.4) {$d_2$};
			\draw[white] (-6.3,0) -- (-6.3, 1.2);
			\draw[white] (6.3,0) -- (6.3,2.5);
			\end{tikzpicture}
	\caption{{$\mathrm{dist}(\mathrm{supp}\chi_1 u,\Upsilon^{(1)}_{\alpha,\gamma,\tau})\geq d_1$ and $\mathrm{dist}(\mathrm{supp}\chi_3 u,\Upsilon^{(2)}_{\alpha,\gamma,a,\tau})\geq d_2$. }}
	\label{fig:4}
\end{figure}
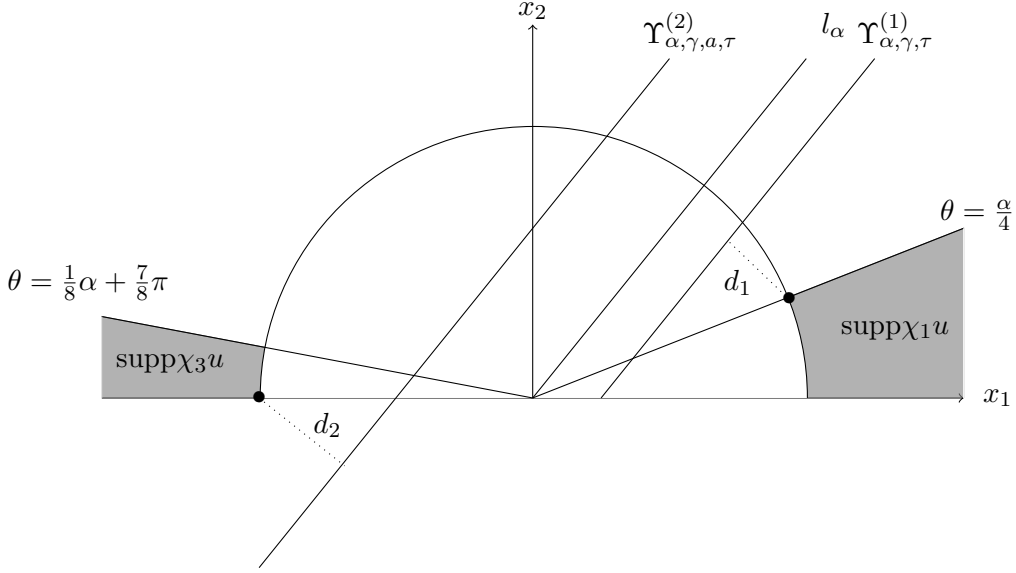
Hence, using the support properties of $\chi_1$, the definition of $V_{\underline{\mathbf{B}}_{\alpha,\gamma,a},\tau}$ in \eqref{eq:Vb}, and~\eqref{eq:V12} we get  for all $x\in\supp \chi_1 u$ 
 \[V_{\underline{\mathbf{B}}_{\alpha,\gamma,a},\tau}(x)=V^{(1)}_{\underline{\mathbf{ B}}_{\alpha, \gamma},\tau}(x)=\sin^2\gamma\dist^2(x,\Upsilon^{(1)}_{\alpha, \gamma, \tau})\geq \left|R\sin\left(\frac{3\alpha}{4}\right)\sin\gamma-\tau\right|^2.
 \]
  Similarly, using~\eqref{eq:V12}, we have for all $x\in\supp \chi_3 u$
  \[V_{\underline{\mathbf{B}}_{\alpha,\gamma,a},\tau}(x)= V^{(2)}_{\underline{\mathbf{ B}}_{\alpha, \gamma, a},\tau}(x)=a^2\sin^2\gamma\dist^2(x,\Upsilon^{(2)}_{\alpha, \gamma, a,\tau} \geq |aR\sin\alpha\sin\gamma+\tau|^2.\]
Thus, we can write
  \begin{equation}\label{eq:ess3}
  \underline Q(\chi_1u)\geq  \left|R\sin\left(\frac{3\alpha}{4}\right)\sin\gamma-\tau\right|^2 \|\chi_1u\|^2 \qquad  \underline Q(\chi_3u)\geq |aR\sin\alpha\sin\gamma+\tau|^2\|\chi_3u\|^2.
  \end{equation}
  Consequently for all $R>R_0$,~\eqref{eq:ims},~\eqref{eq:ess2} and~\eqref{eq:ess3} imply 
 \[\Sigma(\mathcal L_{\underline\Ab_{\alpha,\gamma,a}}+V_{\underline{\mathbf{ B}}_{\alpha,\gamma,a},\tau},R)\geq\inf_{\xi\in \R}\big(\mu_a(\tau\sin\gamma+\xi\cos\gamma)+(\xi\sin\gamma-\tau\cos\gamma)^2\big)-\frac C{R^2}.\]
 Taking the limit $R\rightarrow +\infty$, we establish the lower bound in~\eqref{eq:ess1}.
\end{proof}
Now, we state an immediate consequence of Proposition~\ref{prop:ess}.
\begin{corollary}\label{cor:comp}
 For $a\in[-1,1)\setminus\{0\}$,  $\alpha\in(0,\pi)$, and $\gamma\in(0,\pi/2]$. Let $\underline\sigma_{ess}(\alpha,\gamma,a,\tau)$ be as in \eqref{eq: sigma ess 2d}, we have
\[
 \inf_{\tau\in \R}\underline\sigma_{ess}(\alpha,\gamma,a,\tau)\geq \beta_a,\]
 where  $\beta_a$ is the value defined in~\eqref{eq:beta}.
\end{corollary}
\begin{proof}
By the definition of $\beta_a$, we have
\begin{multline*}\underline\sigma_{ess}(\alpha,\gamma,a,\tau)=\inf_\xi\big(\mu_a(\tau\sin\gamma+\xi\cos\gamma)+(\xi\sin\gamma-\tau\cos\gamma)^2\big)\\
 \geq \inf_\xi\mu_a(\tau\sin\gamma+\xi\cos\gamma)+\inf_\xi(\xi\sin\gamma-\tau\cos\gamma)^2\geq\beta_a.\end{multline*}
\end{proof}
%
\subsubsection{Bottom of the spectrum of the 2D reduced operator at infinity}\label{sec:bound}
Now, we consider the bottom of the spectrum, $\underline\sigma(\alpha,\gamma,a,\tau)$,  of the operator $\mathcal L_{\underline\Ab_{\alpha,\gamma,a}}+V_{\underline{\mathbf{ B}}_{\alpha,\gamma,a},\tau}$  as a function of $\tau$. In Proposition~\ref{prop:sp-lim} below, we study the behavior of $\underline\sigma(\alpha,\gamma,a,\tau)$ as $|\tau|$ goes to infinity. We will use this proposition to provide a condition on  $(\alpha,\gamma,a)$ such that $\inf_\tau\underline\sigma(\alpha,\gamma,a,\tau)$--that is $\lambda_{\alpha, \gamma,a}$ (see~\eqref{eq:l3})--is attained by some $\tau\in\R$. This, together with the upper bound of the essential spectrum in Corollary~\ref{cor:comp}, will be used to get the result in  Theorem~\ref{thm:main}, when the strict inequality in~\eqref{eq:lb} is satisfied. 

\begin{proposition}\label{prop:sp-lim} 
Let $\alpha\in(0,\pi)$ and $\gamma\in (0,\pi/2]$. For $a\in[-1,0)$, we have
\[\lim_{\tau\rightarrow-\infty} \underline\sigma(\alpha,\gamma,a,\tau)=+\infty,\quad \lim_{\tau\rightarrow+\infty} \underline\sigma(\alpha,\gamma,a,\tau)=|a|\zeta_{\nu_0}.\]
For $a\in(0,1)$, we have
\[\lim_{\tau\rightarrow-\infty} \underline\sigma(\alpha,\gamma,a,\tau)=a\zeta_{\nu_0},\quad \lim_{\tau\rightarrow+\infty} \underline\sigma(\alpha,\gamma,a,\tau)=\zeta_{\nu_0}.\]
Here, $\zeta_{\nu_0}$ is defined in~\eqref{eq:zeta} for $\nu_0=\arcsin(\sin\alpha\sin\gamma)$.
\end{proposition}
\begin{proof} In this proof, we simplify the notation and write  $H^{\mathrm{bnd}}$ for  the `boundary operator' $H^{\mathrm{bnd}}_{\alpha,\gamma,a}[\tau]$ in \eqref{eq:Hbnd1} and $\underline{Q}$ for the quadratic form $\underline{Q}^\tau_{\alpha,\gamma,a}$ in \eqref{eq:Qa} associated to the operator $\mathcal L_{\underline\Ab_{\alpha,\gamma,a}}+V_{\underline{\mathbf{B}}_{\alpha,\gamma,a},\tau}$.
\\
\noindent{\underline{\emph{Case $a\in[-1,0)$}}.} Establishing the limit when $\tau\rightarrow -\infty$ is straightforward. Indeed, considering the electric potential in~\eqref{eq:V0}, by~\eqref{eq:S16} we have $\inf  V_{\underline{B}_{\alpha,\gamma,a}, \tau}=\tau^2$, for any $\tau <0$. Then, $\lim_{\tau\rightarrow-\infty}\inf  V_{\underline{B}_{\alpha,\gamma,a}, \tau}=+\infty$. Hence,
\[
	\lim_{\tau\rightarrow-\infty}\underline\sigma(\alpha,\gamma,a,\tau)=+\infty.
\]
Now, we treat the case $\tau \rightarrow +\infty$. Here,  the forgoing operator $H^{\mathrm{bnd}}$ in~\eqref{eq:Hbnd} will be involved. By the min-max principle and Lemma~\ref{lem: bott essHbnd}, for any $\epsilon >0$, there exists a normalized function $u_\epsilon\in C_0^\infty(\overline {\R^2_+}){\setminus\{0\}}$ such that
\begin{equation}\label{eq:s2}
\langle H^{\mathrm{bnd}} u_\epsilon,u_\epsilon\rangle<|a|\zeta_{\nu_0}+\epsilon.
\end{equation}
We define the function $u_{\epsilon,\tau}$ as follows
\[u_{\epsilon,\tau}(x) = u_\epsilon\left(x_1-\frac{\tau}{a\sin\alpha\sin\gamma},x_2\right),\quad\mathrm{for}\ x=(x_1,x_2)\in\R^2_+.\]
For a sufficiently large $\tau$, we have $\supp u_{\epsilon,\tau}\in D^2_\alpha$, where $D^2_\alpha$ is the set in~\eqref{eq: def D12alpha R2}.  Performing a suitable change of gauge, in which we associate the function $\tilde u_{\epsilon,\tau}$ to the function $u_{\epsilon,\tau}$, we get 
\begin{eqnarray*}\label{eq:s3}
\underline Q(\tilde u_{\epsilon,\tau})&=&\langle(\mathcal L_{\underline\Ab_{\alpha,\gamma,a}}+V_{\underline{\mathbf{B}}_{\alpha,\gamma,a},\tau})\tilde u_{\epsilon,\tau},\tilde u_{\epsilon,\tau}\rangle
\\
&=&\langle H^{\mathrm{bnd}}  u_{\epsilon,\tau}, u_{\epsilon,\tau}\rangle
\\
&=&\langle H^{\mathrm{bnd}}   u_{\epsilon}, u_{\epsilon}\rangle\\
&<&|a|\zeta_{\nu_0}+\epsilon,\nonumber
\end{eqnarray*}
where in the last inequality we used \eqref{eq:s2}.  Taking $\tau$ to $+\infty$, we get
\begin{equation*}\label{eq:s4}
\lim\sup_{\tau\rightarrow+\infty} \underline\sigma(\alpha,\gamma,a,\tau)\leq |a|\zeta_{\nu_0}.
\end{equation*}
Next, we establish the lower bound for $\lim_{\tau\rightarrow+\infty} \underline\sigma(\alpha,\gamma,a,\tau)$. We consider a partition of unity $(\tilde\chi_j)_{j\in\{1,2,3\}}\subset C^\infty(\R)$ satisfying
\begin{equation*}\supp \tilde\chi_1\subset \left(\frac {1}{4\sin\gamma},+\infty\right),\ \supp \tilde\chi_2\subset \left(\frac {1}{2a\sin\gamma},\frac {1}{2\sin\gamma}\right),\ \supp \tilde\chi_3\subset \left(-\infty,\frac {1}{4a\sin\gamma}\right)  
\end{equation*}
\begin{equation*}
\sum_j|\tilde\chi_j|^2=1,\qquad \sum_j|\tilde\chi_j'|^{2}\leq C,
\end{equation*}
for a certain $C>0$ independent of $\tau$. Let $(\chi_j)_{j\in\{1,2,3\}}\subset C^\infty(\R^2)$ be the partition of unity of $\R^2$ induced from $(\tilde\chi_j)_{j\in\{1,2,3\}}$ as follows
	\[\chi_j(x_1,x_2)=\tilde \chi_j\Big(\frac{x_1\sin\alpha-x_2\cos\alpha}{\tau}\Big).\]
	Consequently, we have for $j\in\{1,2,3\}$
	\begin{equation*}
\supp \chi_j\subset  R_j,\quad \sum_j|\chi_j|^2=1,\quad \mathrm{and}\ \sum_j|\nabla\chi_j|^{2}\leq \frac {C}{\tau^2},
	\end{equation*}
	where 
	\begin{align*}
	R_1&:=\left\{(x_1,x_2)\in \R^2~:~x_1\sin\alpha-x_2\cos\alpha>\frac{\tau}{4\sin\gamma}\right\}\\
		R_2&:=\left\{(x_1,x_2)\in \R^2~:~\frac{\tau}{2a\sin\gamma}<x_1\sin\alpha-x_2\cos\alpha<\frac{\tau}{2\sin\gamma}\right\}\\
			R_3&:=\left\{(x_1,x_2)\in \R^2~:~x_1\sin\alpha-x_2\cos\alpha<\frac{\tau}{4a\sin\gamma}\right\}.
	\end{align*}
	Thus, for any  $u\in\mathcal{D}(\underline Q)$ (see \eqref{eq:domQ}), the IMS formula gives
	\begin{equation}\label{eq:ims2}
	\underline{Q}(u) =\sum_{j=1}^3\underline Q(\chi_j u)-\sum_{j=1}^3\|u |\nabla \chi_j| \|^2_{L^2(\R^2_+)}
	\geq \sum_{j=1}^3\underline Q(\chi_ju)-
	\frac {C}{\tau^2}.
	\end{equation}
We perform a suitable change of gauge and use~\eqref{eq:Hbnd1}, together with the support properties of $\chi_1u$, to get 
	\begin{eqnarray}\label{eq:s5}
	\underline Q(\chi_1u) &=& \int_{\mathbb{R}_+^2}\left( |(\nabla - i \mathbf{\underline{A}}_{\alpha,\gamma,a})(\chi_1u)|^2 + V_{\underline{B}_{\alpha,\gamma,a}, \tau}|\chi_1 u|^2\right)\,dx_1dx_2
	\\
	&\geq&\zeta_{\nu_0}  \|\chi_1 u\|^2\nonumber.
	\end{eqnarray}
	Similarly, considering the support of $\chi_3u$, doing a change of gauge and using Lemma \ref{lem: bott essHbnd}, we find
	\begin{equation}\label{eq:s6}
	\underline Q(\chi_3u)\geq |a|\zeta_{\nu_0}\|\chi_3u\|^2.
	\end{equation}
	Finally, considering  the support of $\chi_2u$, a simple computation using the definition of the electric potential in~\eqref{eq:V0} and Lemma~\ref{lem:upsilon} (see also~Figure~\ref{fig:2}) gives
	\begin{equation}\label{eq:Vb}
 V_{\underline{B}_{\alpha,\gamma,a}, \tau}\geq\frac{\tau^2}{4},\quad \mathrm{for}\ x\in\supp\chi_2u.
	\end{equation} 
	Hence, there exists $\tau_0>0$ and $M>|a|\zeta_{\nu_0}$ such that for $\tau>\tau_0$
	\begin{equation}\label{eq:s7}
	\underline Q(\chi_2u)\geq M\|\chi_2u\|^2.
	\end{equation}
	Implementing ~\eqref{eq:s5}, \eqref{eq:s6} and \eqref{eq:s7} in~\eqref{eq:ims2}, we get for $a\in[-1,0)$
	\[\lim\inf_{\tau\rightarrow+\infty}\underline\sigma(\alpha,\gamma,a,\tau)\geq|a|\zeta_{\nu_0}.\]
	\\
	\noindent{\emph{Case $a\in(0,1)$.}} Adopting a similar approach as above, using Lemma~\ref{lem:upsilon} for positive values of $a$, one can establish the results of the proposition in this case. We omit further computation details.
\end{proof}
\section{Proof of Theorem \ref{thm:main} and Proposition \ref{prop:exn}}\label{sec:proof}
	\begin{proof}[Proof of Theorem~\ref{thm:main}]
		The proof in the case $\gamma=0$, is a direct consequence of the results in Section~\ref{sec:tangent}. Indeed, from \eqref{eq: lambda mu} and \eqref{eq: inf ess sp gamma0} it follows that 
		\[
			\lambda_{\alpha, 0, a} \leq |a| \Theta_0.
		\] 
		Now, from \eqref{eq: prop zeta}, we know that $\zeta_0 = \Theta_0$ and having $\beta_a \geq |a| \Theta_0$ (see Section \ref{sec:eff2}), we get that 
		\begin{equation*}
\lambda_{\alpha, 0, a} \leq \min (\beta_a, |a| \Theta_0 ) = \min (\beta_a, |a|\zeta_0).
		\end{equation*}
		Moreover, it follows from Section \ref{sec:tangent} that if $\lambda_{\alpha, 0, a} < \min (\beta_a, |a|\zeta_0 )$, then $\lambda_{\alpha, 0, a}$ is an eigenvalue of the operator $\mathcal L_{\underline\Ab_{\alpha,\gamma,a}}+V_{\underline{\mathbf{ B}}_{\alpha,\gamma,a},\tau^*}$, with the particular choice $\tau^* = 0$. 
		 
	Next, we treat the case $\gamma\neq 0$. We first establish the upper bound of $\lambda_{\alpha, \gamma,a}$ in~\eqref{eq:lb1}. The result is a consequence of Proposition~\ref{prop:ess} and Proposition~\ref{prop:sp-lim},  as it is shown below. We have (see \eqref{eq:l3}) 
		\begin{equation}\label{eq: lambda}
			\lambda_{\alpha, \gamma,a} = \inf_\tau\underline\sigma(\alpha,\gamma,a,\tau),
		\end{equation} 
		where $\underline{\sigma}(\alpha,\gamma,a,\tau)$ is as in \eqref{eq:unders}. We consider the following two cases.\\
\underline{\emph{Case $a\in[-1,0)$}}. From Proposition~\ref{prop:ess}, we have 
		\begin{equation*}
			{\underline{\sigma}_{ess}(\alpha,\gamma,a,\tau) = \inf_{\xi\in \R} (\mu_a(\tau\sin\gamma + \xi\cos\gamma) + (\xi\sin\gamma - \tau\cos\gamma)^2),}	
		\end{equation*}
		where $\mu_a(\cdot)$ is introduced in \eqref{eq:mu_a}. Let $\xi_a$ be the unique  minimum of $\mu_a(\cdot)$ (see Section \ref{sec:eff2}). For $\widetilde\tau :=\xi_a\sin\gamma$, one can see that $\underline\sigma_{ess}(\alpha,\gamma,a,\widetilde\tau)$ is attained by $\xi=\xi_a\cos\gamma$ and satisfies
		\[\underline\sigma_{ess}(\alpha,\gamma,a,\widetilde\tau)=\mu_a(\xi_a)=\beta_a.\] 
		This implies that
		\begin{equation}\label{eq: fin 1}
			\underline{\sigma}(\alpha, \gamma, a, \widetilde\tau) \leq \beta_a.
		\end{equation}
		 Moreover, by Proposition \ref{prop:sp-lim}, we have 
		\begin{equation}\label{eq: fin 2}
			\underline{\sigma}(\alpha, \gamma, a,\widetilde\tau) \leq |a| \zeta_{\nu_0}.
		\end{equation}
		Combining \eqref{eq: lambda}--\eqref{eq: fin 2} yields \eqref{eq:lb1}. \\
		\underline{\emph{Case $a\in(0,1)$}}. By Proposition~\ref{prop:sp-lim}, we have
		\begin{equation*}
			\underline{\sigma}(\alpha,\gamma,a, \tau) \leq a\zeta_0. 
		\end{equation*}
		Moreover, $\beta_a=a$ for $a\in (0,1)$ (see Section 2.1),  and $\zeta_{\nu_0}\leq1$ (see Section \ref{sec:nu+}). This yields
		\[
			\lambda_{\alpha, \gamma,a}\leq  a\zeta_{\nu_0} = \min (\beta_a, a\zeta_{\nu_0}).
		\]	

		Now, we consider the case when the strict inequality in~\eqref{eq:lb} is satisfied. From Proposition~\ref{prop:sp-lim}, we have
		\[\inf_\tau\underline\sigma(\alpha,\gamma,a,\tau)=\lambda_{\alpha, \gamma,a}<|a|\zeta_{\nu_0}=\min\left(\lim_{\tau\rightarrow-\infty} \underline\sigma(\alpha,\gamma,a,\tau),\lim_{\tau\rightarrow+\infty} \underline\sigma(\alpha,\gamma,a,\tau)\right).\]
		 Hence, $\inf_\tau\underline\sigma(\alpha,\gamma,a,\tau)$ is attained by some $\tau_*\in\R$. Moreover, by Corollary~\ref{cor:comp} we know that 
		 \[
		 	\lambda_{\alpha, \gamma,a}=\underline\sigma(\alpha,\gamma,a,\tau_*)<\beta_a\leq \underline\sigma_{ess}(\alpha,\gamma,a,\tau_*).
		 \]
		 We then deduce that $\lambda_{\alpha, \gamma,a}$ is an eigenvalue of $\mathcal L_{\underline\Ab_{\alpha,\gamma,a}}+V_{\underline{\mathbf{ B}}_{\alpha,\gamma,a},\tau_*}$.
	\end{proof}
\begin{proof}[Proof of Proposition~\ref{prop:exn}] The proof is inspired by the construction done in~\cite[Proof of Proposition~3.9]{Assaad3} while studying 2D smooth domains under  discontinuous magnetic fields, and by~\cite[Proof of Theorem~1.1]{exner2018bound} while studying 2D corner domains under constant magnetic fields. 
	
	We fix  $a\in[-1,1)\setminus\{0\}$, $\alpha\in(0,\pi)$, and  $\gamma\in[0,\pi/2]$. 
\emph{Let $\tau=0$}. We define the function $\varphi_{\alpha,\gamma,a}\in H^1_{\mathrm{loc}}(\R_+^2)$ by
\begin{equation*}
\varphi_{\alpha,\gamma,a}(x_1,x_2)=
\left\{
\begin{array}{ll}
\Big(\frac 12x_1x_2+\frac {a-1}{2}  x_2^2\cot \alpha \Big)\cos\gamma& \mathrm{if}\ (x_1,x_2)\in D_\alpha^1,\\
\frac a2  x_1x_2\cos\gamma& \mathrm{if}\ (x_1,x_2)\in D_\alpha^2.
\end{array}
\right.
\end{equation*}
This function satisfies $\underline \Ab_{\alpha,\gamma,a}=\breve\Ab_{\alpha,\gamma,a}+\nabla\varphi_{\alpha,\gamma,a}$, where $\underline \Ab_{\alpha,\gamma,a}$ is the potential in~\eqref{eq:Abar}, and $\breve\Ab_{\alpha,\gamma,a}=1/2(-x_2,x_1) \underline{\mathsf s}_{\alpha,a}\cos\gamma$, for $\underline{\mathsf s}_{\alpha,a}=\mathbbm 1_{D_\alpha^1}+a\mathbbm 1_{D_\alpha^2}$  being the step function in~\eqref{eq:step} (see~\cite[Lemma 1.1]{leinfelder1983gauge} for the existence of such gauge functions in more general situations). We define the quadratic form $\breve Q_{\alpha,\gamma,a}$ as follows
\begin{equation*}
\breve Q_{\alpha,\gamma,a}(v)=\int_{\R^2_+}\Big(\big|(\nabla-i\breve \Ab_{\alpha,\gamma,a})v\big|^2+V_{\underline{\mathbf{ B}}_{\alpha,\gamma,a},0}|v|^2\Big)\,dx_1\,dx_2
\end{equation*}
in the domain
\[\mathcal{D}(\breve Q_{\alpha,\gamma,a})=\left\{v \in L^2(\R_+^2)~:~(\nabla-i\breve \Ab_{\alpha,\gamma,a})v \in L^2(\R_+^2), |x_1\sin\alpha-x_2\cos\alpha|v\in L^2(\R_+^2) \right\},\]
where $V_{\underline{\mathbf{ B}}_{\alpha,\gamma,a},0}=\underline{\mathsf s}_{\alpha,a}^2(x_1\sin\gamma\sin\alpha-x_2\sin\gamma\cos\alpha)^2$ is the electric potential defined in~\eqref{eq:V0} for $\tau=0$.
We explicitly express $\breve Q_{\alpha,\gamma,a}(v)$ by
\begin{multline*}\int_{\R^2_+}\Big(\big|(\partial_{x_1}+\frac 12i\underline{\mathsf s}_{\alpha,a} x_2\cos\gamma)v\big|^2+\big|(\partial_{x_2}-\frac 12i\underline{\mathsf s}_{\alpha,a} x_1\cos\gamma)v\big|^2\\+\underline{\mathsf s}_{\alpha,a}^2\sin^2\gamma(x_1\sin\alpha- x_2\cos\alpha)^2|v|^2\Big)\,dx_1dx_2.\end{multline*}
For any $v\in\mathcal{D}(\breve Q_{\alpha,\gamma,a})$, we have 
\begin{equation*}\label{eq:q0}
\breve Q_{\alpha,\gamma,a}(v)=\underline Q^{\tau=0}_{\alpha,\gamma,a}(e^{i\varphi_{\alpha,\gamma,a}}v),
\end{equation*}
where $\underline Q^{\tau}_{\alpha,\gamma,a}$ is the quadratic form  in~\eqref{eq:Qa}.
In the rest of the proof, we write $\breve Q$ for $\breve Q_{\alpha,\gamma,a}$ and $\underline{\mathsf{s}}$ for $\underline{\mathsf{s}}_{\alpha,a}$.
We now express $\breve Q$ in the  polar coordinates $(\rho,\theta)\in (0,+\infty)\times(0,\pi)=:\breve D_{\mathrm{pol}}$ as follows
\[\breve Q_{\mathrm{pol}}(v)=\int_0^\pi\int_0^{+\infty}\Big(|\partial_{\rho}v|^2+\frac 1{\rho^2}\big|(\partial_{\theta}-i\underline{\mathsf s}_{\mathrm{pol}}\frac {\rho^2}2\cos\gamma)v\big|^2+\underline{\mathsf s}_{\mathrm{pol}}^2\rho^2\sin^2\gamma\sin^2(\alpha-\theta)|v|^2\Big)\,\rho\, d\rho\, d\theta,\]
where $\underline{\mathsf s}_{\mathrm{pol}}(\rho,\theta)=\underline{\mathsf s}( x_1, x_2)$ and
	\begin{multline*}\mathcal{D}(\breve Q_{\mathrm{pol}})=\Big\{v \in L^2_\rho(\breve D_{\mathrm{pol}})~:~\partial_\rho v \in L^2_\rho(\breve D_{\mathrm{pol}}),\,\\ \frac 1 \rho\Big(\partial_{\theta}-i\underline{\mathsf s}_{\mathrm{pol}}\frac {\rho^2}2\cos\gamma\Big)v \in L^2_\rho(\breve D_{\mathrm{pol}}),\, \rho v \in L^2_\rho(\breve D_{\mathrm{pol}})\Big\}.
	\end{multline*}
	For any $D\subset\R^2$, we denote by $L^2_\rho(D)$ the weighted space of weight $\rho$.  Consider further the quadratic form $\tilde{Q}_{\mathrm{pol}}$, defined on $\tilde{D}_{\mathrm{pol}}:=(0,+\infty)\times(-\pi+\alpha,\alpha)$ by
	\begin{equation*}
	\tilde{Q}_{\mathrm{pol}}(u)=\int_{-\pi+\alpha}^\alpha\int_0^{+\infty}\Big(|\partial_\rho u|^2+\frac 1 {\rho^2}\Big|\Big(\partial_{\theta}+i\tilde{\mathsf s}_{\mathrm{pol}}\frac {\rho^2}2\cos\gamma \Big)u\Big|^2+\tilde{\mathsf s}^2_{\mathrm{pol}}\rho^2\sin^2\gamma\sin^2\theta|u|^2\Big)\rho\,d\rho\,d\theta,
	\end{equation*}
	where
\begin{multline*}\mathcal{D}(\tilde{Q}_{\mathrm{pol}})=\Big\{u \in L^2_\rho(\tilde{D}_{\mathrm{pol}})~:~\partial_\rho u \in L^2_\rho(\tilde{D}_{\mathrm{pol}}),\\ \frac 1 \rho\Big(\partial_{\theta}+i\tilde{\mathsf s}_{\mathrm{pol}}\frac {\rho^2}2\cos\gamma\Big)u \in L^2_\rho(\tilde{D}_{\mathrm{pol}}),\, \rho u \in L^2_\rho(\tilde D_{\mathrm{pol}})  \Big\},
\end{multline*}
	and \begin{equation*}
	\tilde{\mathsf s}_{\mathrm{pol}}(\rho,\theta)=
	\left\{
	\begin{array}{ll}
	a& \mathrm{if}\ (\rho,\theta)\in (0,+\infty)\times(-\pi+\alpha,0),\\
	1& \mathrm{if}\ (\rho,\theta)\in (0,+\infty)\times(0,\alpha).
	\end{array}
	\right.
	\end{equation*} 
	For any $u\in \mathcal{D}(\tilde{Q}_{\mathrm{pol}})$, we have $\tilde{Q}_{\mathrm{pol}}(u)=\breve{Q}_{\mathrm{pol}}(v)$, where $v(\rho,\theta)=u(\rho,-\theta+\alpha)$.
	
	In light of the computation above and from Theorem \ref{thm:main},  a sufficient condition for $\inf_{\tau} \underline\sigma(\alpha,\gamma,a,\tau)$ to be attained by some $\tau_*\in\R$ and to be an eigenvalue of the operator $\mathcal L_{\underline\Ab_{\alpha,\gamma,a}}+V_{\underline{\mathbf{B}}_{\alpha,\gamma,a},\tau_*}$ is to find a trial function $u_0\in \mathcal{D}(\tilde{Q}_{\mathrm{pol}})$ satisfying
	\begin{equation}\label{eq:exn1}
	\tilde{Q}_{\mathrm{pol}}(u_0)<\Lambda\|u_0\|^2_{L^2_\rho(\R^2_+)},
	\end{equation}
	where $\Lambda=\Lambda[\alpha,\gamma,a]$ is the minimum between $\beta_a$ and $|a|\zeta_{\nu_0}$.
	Towards this, we consider the function 
	\[u_0(\rho,\theta)=e^{-\omega \frac {\rho^2}2}e^{-i\rho g(\theta)},\]
	where  $g\colon(-\pi+\alpha,\alpha) \rightarrow \R$ is a piecewise-differentiable function and  $\omega>0$. In what follows, we will suitably choose $g$ and $\omega$. We define the functional $\mathcal J$ on $\dom \tilde{Q}_{\mathrm{pol}}$ by
	\[u\mapsto \mathcal J[u]=\tilde{Q}_{\mathrm{pol}}(u)-\Lambda\|u\|^2_{L^2_\rho(\tilde{D}_{\mathrm{pol}})}.\]
	The condition in~\eqref{eq:exn1}  is now equivalent to 
	\begin{equation}\label{eq:exn2}
	\mathcal J[u_0]<0.
	\end{equation} 
	We compute $\mathcal J[u_0]$ and get
	\begin{eqnarray*}
	\mathcal J[u_0] &=& \int_0^{+\infty} \rho e^{-\omega \rho^2}\,d\rho \int_{-\pi+\alpha}^{\alpha} \Big(g^2(\theta)+g'^{\,2}(\theta)-\Lambda\Big)\,d\theta
	\nonumber\\
	&&- \int_0^{+\infty} \rho^2 e^{-\omega \rho^2}\,d\rho \int_{-\pi+\alpha}^{\alpha} \tilde{\mathsf s}_{\mathrm{pol}} g'(\theta)\cos\gamma\,d\theta\nonumber\\
	&&+\int_0^{+\infty} \rho^3 e^{-\omega \rho^2}\,d\rho \int_{-\pi+\alpha}^\alpha \Big( \omega^2+\tilde{\mathsf s}^2_{\mathrm{pol}}\sin^2\gamma\sin^2\theta+\frac 14 \tilde{\mathsf s}^2_{\mathrm{pol}}\cos^2\gamma\Big)\,d\theta.\label{eq:exn3}
	\end{eqnarray*}
We use the following properties of  $\mathcal E_n=\int_0^{+\infty} \rho^n e^{-\omega \rho^2}\,d\rho$, for $n\geq0$: $\mathcal E_1=1/(2\omega)$, $\mathcal E_2=\sqrt \pi/(4 \omega^{3/2})$, and $\mathcal E_3=1/(2\omega^2)$ (see~\cite[Equations 3.461]{gradshteyn2015table}). Hence,~\eqref{eq:exn3} becomes 
	\begin{multline}
	\mathcal J[u_0]=\frac 1{2\omega} \int_{-\pi+\alpha}^{\alpha} \Big(g^2(\theta)+g'^{\,2}(\theta)-\Lambda\Big)\,d\theta- \frac{\sqrt \pi}{4 \omega^{3/2}}\int_{-\pi+\alpha}^{\alpha} \tilde{\mathsf s}_{\mathrm{pol}}g'(\theta)\cos\gamma \,d\theta\\
	+\frac 1{2\omega^2} \int_{-\pi+\alpha}^\alpha \Big( \omega^2+\tilde{\mathsf s}^2_{\mathrm{pol}}\sin^2\gamma\sin^2\theta+\frac 14 \tilde{\mathsf s}^2_{\mathrm{pol}}\cos^2\gamma\Big)\,d\theta.\label{eq:exn4}
	\end{multline}
	Now, we choose 
	\begin{equation*}
	g(\theta)=
	\left\{
	\begin{array}{ll}
	c_1 e^\theta+c_2 e^{-\theta}& \mathrm{if}\ -\pi+\alpha<\theta\leq 0,\\
	c_3 e^\theta+c_4 e^{-\theta}& \mathrm{if}\ 0<\theta<\alpha,
	\end{array}
	\right.
	\end{equation*} 
	where $c_i$, $i=1,\cdots,4$,  are real coefficients satisfying the  condition $c_1+c_2=c_3+c_4$ which makes the function $g$ continuous on $(-\pi+\alpha,\alpha)$. Implementing this choice in~\eqref{eq:exn4} yields
	\begin{multline*}
	\mathcal J[u_0]=\frac{(2-e^{-2\alpha}-e^{-2\pi+2\alpha})}{2\omega}c_1^2+\frac{(-e^{-2\alpha}+e^{2\pi-2\alpha})}{2\omega}c_2^2
	+\frac{(-e^{-2\alpha}+e^{2\alpha})}{2\omega}c_3^2+\\
	\frac{(1-e^{-2\alpha})}{\omega}c_1c_2+\frac{(-1+e^{-2\alpha})}{\omega}c_1c_3+\frac{(-1+e^{-2\alpha})}{\omega}c_2c_3+\\
	\frac {(1-a-e^{-\alpha}+ae^{-\pi+\alpha})\sqrt{\pi}\cos\gamma}{4\omega^{\frac 32}}c_1
	+\frac{(1-a-e^{-\alpha}+ae^{\pi-\alpha})\sqrt{\pi}\cos\gamma}{4\omega^{\frac 32}}c_2+\\
	\frac{(e^{-\alpha}-e^{\alpha})\sqrt{\pi}\cos\gamma}{4\omega^{\frac 32}}c_3+\frac{4\pi\omega^2-4\pi\omega \Lambda+ (a^2(\pi-\alpha)+\alpha)\cos^2\gamma}{8\omega^2}\\
	+\frac {2\big(a^2(\pi-\alpha)+\alpha+(a^2-1)\cos\alpha\sin\alpha\big)\sin^2\gamma}{8\omega^2}
	\end{multline*}
	Notice that $\mathcal J[u_0]$ is quadratic  in $c_1$, $c_2$ and $c_3$. Minimizing $\mathcal J[u_0]$ with respect to these coefficients gives a unique solution $(c_1,c_2,c_3)$, which is
	\begin{align*}
	c_1=& \frac{e^{\pi-2\alpha}\big((-1+a)e^{\pi}+(-1+a)e^{\pi+2\alpha}+2e^\alpha(-a+e^\pi)\big)\sqrt{\pi}\cos\gamma\big(-1+\coth\pi\big)}{16\sqrt{\omega}} \\
	c_2=&\frac{\big(-1+a+(-1+a)e^{2\alpha}-2(-1+ae^\pi) e^\alpha\big)\sqrt{\pi}\cos\gamma\big(-1+\coth\pi\big)}{16\sqrt{\omega}}\\
	c_3=&\frac{e^{-\alpha}\big(-a+e^\pi+(-1+a)\cosh(\pi-\alpha)\big)\sqrt{\pi}\cos\gamma\csch\pi}{8\sqrt{\omega}}.
	\end{align*}
	We compute $\mathcal J[u_0]$ corresponding  to the coefficients above, and get $\mathcal J[u_0]=P[\alpha,\gamma,a](x)$ with $x=\frac1\omega>0$, where $P[\alpha,\gamma,a]
	+$ is as in \eqref{eq: P}. This, together with the condition in~\eqref{eq:exn2}, complete the proof.
\end{proof}
We consider  cases of $(\alpha,\gamma,a,\tau)$ where the infimum of the spectrum of $\mathcal L_{\underline\Ab_{\alpha,\gamma,a}}+V_{\underline B_{\alpha,\gamma,a},\tau}$ is an eigenvalue below  the essential spectrum.
The following theorem reveals an exponential decay of the corresponding eigenfunction, for large values of $|x|$. This is a standard Agmon-estimate result. For the proof details,  we refer the reader to similar results in~\cite[Theorem~9.1]{bonnaillie2003analyse} and~\cite{bonnaillie2012discrete}.
\begin{theorem}\label{thm:v0-decay}
	 Let $a\in[-1,1)\setminus\{0\}$,  $\alpha\in(0,\pi)$, $\gamma\in(0,\pi/2]$ and $\tau\in \R$. Consider the case where $\underline\sigma(\alpha,\gamma,a,\tau)<\underline\sigma_{ess}(\alpha,\gamma,a,\tau)$. Let $v_{\alpha,\gamma,a,\tau}$ be the  normalized eigenfunction corresponding to $\underline\sigma(\alpha,\gamma,a,\tau)$. For all $\eta\in { \big(0,\sqrt{\underline\sigma_{ess}(\alpha,\gamma,a,\tau)-\underline\sigma(\alpha,\gamma,a,\tau)}\big)}$, there exists a constant $C$ which depends on $\alpha$ and $\eta$ such that
	\[\underline Q_{\alpha,\gamma,a}^\tau(e^{\eta\phi}v_{\alpha,\gamma,a,\tau})\leq C,\]
	where  $\phi(x)=|x|$, for $x\in\R^2_+$.
\end{theorem} 

\section{Proof of Theorem \ref{thm: localization}}\label{sec: proof localization}
We consider the open and bounded set $\Omega\subset\mathbb{R}^3$ defined in the settings of Section \ref{sec: application}. Let $a\in [-1,1) \setminus \{0\}$ and $\mathfrak b>0$, we recall the piecewise-constant magnetic field $\mathbf B$ in \eqref{eq: magn field applic}
\begin{equation}\label{eq:B1}
	\mathbf{B}(x) = \mathbf{s}(x)(0,0,1), \qquad \mathbf{s} = \mathbbm{1}_{\Omega_1} + a \mathbbm{1}_{\Omega_2},
\end{equation}
and the linear operator $\mathcal{P}_{\mathfrak{b}, \mathbf{F}}$ introduced in \eqref{eq: def semicl op}. 
Recall also the discontinuity surface $S$,  at which the strength $|\mathbf{B}|$ exhibits a  discontinuity jump, and  the discontinuity edge $\Gamma$, which is the boundary of $S$. 

We denote by $Q_{\mathfrak{b}, \mathbf{F}}$ the quadratic form associated to $\mathcal{P}_{\mathfrak{b}, \mathbf{F}}$, defined by
\begin{equation}\label{eq: dom Q b}
	Q_{\mathfrak{b}, \mathbf{F}}(u) = \int_\Omega \big|(\nabla - i \mathfrak{b}\mathbf{F})u\big|^2\,dx, \qquad \mathcal{D}(Q_{\mathfrak{b}, \mathbf{F}}) = H^1(\Omega).
\end{equation}
The bottom of the spectrum $\lambda(\mathfrak{b})$ is  equal to 
\begin{equation}\label{eq:minmax}
	\lambda(\mathfrak{b}) = \inf_{\substack {u\in \mathcal{D} (Q_{\mathfrak{b}, \mathbf{F}}) \\ u\neq 0}}\frac{Q_{\mathfrak{b}, \mathbf{F}}(u)}{\|u\|_{L^2(\Omega)}^2}.
\end{equation}
We consider large values of $\mathfrak b$. The main goal of this section is to prove Theorem \ref{thm: localization}, that is to establish the localization of the eigenfunction corresponding to the eigenvalue $\lambda(\mathfrak{b}) $ near the set $D$ of points of the discontinuity edge $\Gamma$, given by 
\begin{equation}
	D= \big\{ \overline{x}\in \Gamma\, \vert\, \lambda_{\alpha_{\overline{x}}, \gamma_{\overline{x}}, a} < |a| \Theta_0\big\}.
\end{equation}
But first, as seen below, the discussion leading to the proof of Theorem \ref{thm: localization} establishes as a by-product the following  rough asymptotics of $\lambda(\mathfrak{b})$, as $\mathfrak b\rightarrow +\infty$.
\begin{theorem}[Asymptotics for $\lambda(\mathfrak{b})$]\label{thm: asympt semicl} Under Assumption \ref{asu}, it holds
	\[
	\lambda(\mathfrak{b}) = \mathfrak{b}\inf_{x\in{\Gamma}}  \lambda_{\alpha_x,\gamma_x,a}+ o(\mathfrak{b}), \qquad \mbox{as}\ \mathfrak b\rightarrow+\infty.
	\]
	\end{theorem}
The proof of the theorem above is split in two parts, in Proposition \ref{pro: lower bound bottom} we prove the lower bound and in Proposition \ref{pro: upper bound lambda b} we establish the corresponding upper bound.  Proposition \ref{pro: lower bound bottom} is a particular result induced from Proposition~\ref{pro: lower bound quadratic form}. The latter proposition is essential in establishing the Agmon estimates in Theorem~\ref{thm: localization}.

\subsection{Change of variables}
We will localise the study of the energy in different regions of $\overline\Omega$, which we classify into four categories: regions away from the discontinuity surface $S$ and the boundary $\partial\Omega$, regions meeting $S$ and away from $\partial\Omega$, regions meeting $\partial\Omega$ and away from $S$ and its boundary $\Gamma$, and regions meeting $\Gamma$. A rigorous definition is given later (see Section~\ref{sec:low}). In each of these regions, we will use suitable local coordinates. When working away from the discontinuity surface $S$ and its boundary $\Gamma$, the situation is well-known and already analysed in previous papers (see e.g. \cite[Chapter 9]{fournais2010spectral}). We focus then on new situations when the foregoing regions meet $\overline S$. Below, we  describe the appropriate local coordinates to use in two cases: the first one is when the regions meet $\Gamma$ and the second one is when these regions meet  $S$ away from $\Gamma$.

\subsubsection{Boundary coordinates near the discontinuity  edge $\Gamma$}\label{sec: coordinates discontinuity line}
In this section, we will define a local change of coordinates near the discontinuity edge $\Gamma$.

 Let $x_0\in\Gamma$. After performing a translation, we may assume that the Cartesian coordinates of the point $x_0$ are all $0$ ($x_0 = 0$). In what follows, we work near the point $x_0$. 

As seen below, our problem will have $\mathcal L_{\alpha_0,\gamma_0,a}$ as a leading operator, where  $\mathcal L_{\alpha_0,\gamma_0,a}$ is defined in~\eqref{eq:La+}, for $\alpha_0$ being the angle between the tangent plane of $\partial\Omega$  and the discontinuity surface $S$ at $x_0$, and $\gamma_0$ being the angle between the magnetic field $\mathbf{B}$ and the discontinuity edge $\Gamma$ at this point. To show this link with the leading operator, we define a coordinates-transformation $\Phi=\Phi_{x_0}:(x_1,x_2,x_3)\mapsto (y_1,y_2,y_3)$, in a neighborhood $\mathcal N_{x_0}$ of $x_0$, s.t. $\Phi(x_0)=(0,0,0)$ and there exists a neighborhood $\mathcal U_0$ of $(0,0,0)$ where
\begin{equation}\label{c1}
	\Phi(\mathcal N_{x_0}\cap \Gamma)=\mathcal U_0\cap(y_3\mbox{-axis})\,
	\end{equation}
\begin{equation}\label{c2}
	\Phi(\mathcal N_{x_0}\cap S)=\mathcal U_0\cap P_{\alpha_0}\,
\end{equation}
\begin{equation}\label{c3}
	\Phi(\mathcal N_{x_0}\cap (\partial\Om_1\setminus\Gamma))=\mathcal U_0\cap\{(y_1,0,y_3):y_1>0\}\,
\end{equation}
\begin{equation}\label{c4}
	\Phi(\mathcal N_{x_0}\cap (\partial\Om_2\setminus\Gamma))=\mathcal U_0\cap\{(y_1,0,y_3):y_1<0\}\,
\end{equation}
\begin{equation}\label{c5}
	\Phi(\mathcal N_{x_0}\cap \Om_1)=\mathcal U_0\cap\mathcal D_{\alpha_0}^1\ \mbox{and}\ 	\Phi(\mathcal N_{x_0}\cap \Om_2)=\mathcal U_0\cap\mathcal D_{\alpha_0}^2.
\end{equation}
Here, $\mathcal D_{\alpha_0}^1$ and $\mathcal D_{\alpha_0}^2$ are the sets in~\eqref{eq: def D1alpha} and~\eqref{eq: def D12alpha}, and $P_{\alpha_0}=\mathcal D_{\alpha_0}^1\cap\mathcal D_{\alpha_0}^2$ is a semi plane making an angle of $\alpha_0$ with $(y_1y_3)$ (see Figure~\ref{fig:coordinates} for illustration):
\begin{equation*}
	P_{\alpha_0}:\begin{cases}
	y_2&=y_1\tan\alpha_0,\ y_1>0,\ \mbox{if}\ \alpha_0\in(0,\frac {\pi}{2})\\
	y_2&=y_1\tan\alpha_0,\ y_1<0,\ \mbox{if}\ \alpha_0\in(\frac {\pi}{2},\pi)\\
y_2&=0,\ \mbox{if}\ \alpha_0=\frac\pi{2}.
	\end{cases}
\end{equation*}

To that end, we use the same 'magnetic normal coordinates'  transformation in~\cite{PopRay}, which we denote by $\Phi$ in our paper. In~\cite{PopRay}, $\Phi$ is the composition of two local diffeomorphisms $\Phi_1$ and $\Phi_2$ respectively defined near $x_0$ and $\Phi_1(x_0)$ (see the precise definitions below in this section). 

Roughly speaking, $\Phi_1:(x_1,x_2,x_3)\mapsto(r,t,s)$ is a standard tubular coordinates-transformation that straightens  $\Gamma$, and sends the boundary $\partial\Om$ and the surface $S$ (near $x_0$) respectively to surfaces $\breve{(\partial\Om)}$ and $\breve S$ which make  at the point $\Phi_1(x)$ the same (opening) angle $\alpha_x$ made between $\partial \Om$  and $S$ at the point $x$, for $x\in\Gamma$. However, these surfaces are not necessarily  planar surfaces (see~Figure~\ref{fig:coordinates}). 

To straighten $\breve{(\partial\Om)}$ and $\breve S$, and to transform the variable opening angle $\alpha_x$ to the forgoing constant angle $\alpha_0$, we perform a second transformation, $\Phi_2$, near $\Phi_1(x_0)=(0,0,0)$. In other words, the local diffeomorphism $\Phi_2:(r,t,s)\mapsto (y_1, y_2, y_3)$ is defined such that $\breve{(\partial\Om)}$ is sent to a patch of the $(y_1y_3)$-plane  and $\breve S$ to a patch of the aforementioned semi plane $P_{\alpha_0}$ (again see~Figure~\ref{fig:coordinates}).

In what follows, we will make more precise the rough discussion above. We will borrow from~\cite{PopRay} the following terminology: we refer to $\Phi_1$ (resp.~$\Phi_2$) as the first (resp.~second) normalization transformation.

\paragraph{\emph{The first normalization}} We first consider the tubular coordinates-transformation $\Phi$ defined in a neighborhood of $(0,0,0)$ by (see~\cite{PopRay}, also~e.g.~\cite[Section 3]{Pan1} or~\cite{FKP})
\[\Phi^{-1}_1(r,t,s)={rV(s)+}\xi(s)+tn(s),\]
where $s\mapsto\xi(s)$  is a parametrization by arc length of the edge $\Gamma$, $n(s)$ is the inward unit normal vector at the point $\xi(s)$,  and $V(s)$ is the unit vector normal to $\Gamma$ at $\xi(s)$ in the tangent plane of $\partial \Om$, pointing toward $\Om_1$. The orientation of the parametrization of $\Gamma$ is fixed such that $\det(V(s),n(s),\xi'(s))>0$.

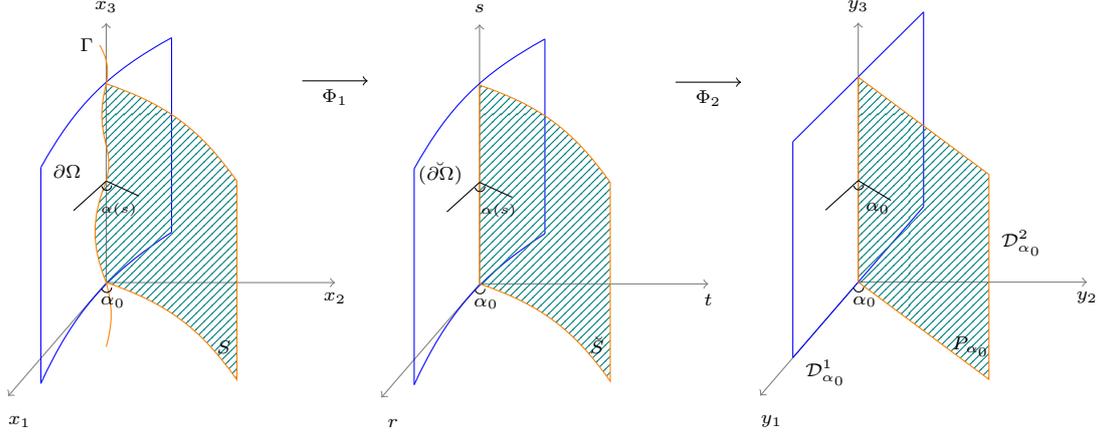
\begin{figure}
	\centering\begin{tikzpicture}[scale= 0.43]
		\draw[gray, ->] (0,0.78) to (0,8.78);
		\draw[gray, ->] (0,0.78) to (7,0.78);
		\draw[gray, ->] (0,0.78) to (-3,-2.72);
		\draw[blue] (-2, -2.33) to[bend left = 16] (2,2.33);
		\draw[blue] (-2,4.33) to[bend left = 16] (2, 8.33);
		\node at (-0.6,8.1) {\tiny{$\Gamma$}};
		\node at (-1.2, 4.2) {\tiny{$\partial\Omega$}};
		\draw[blue] (-2,-2.33) to (-2, 4.33);
		\draw[blue] (2,8.33) to (2,2.33);
		\draw[orange, pattern=north east lines,  pattern color=teal] (0,0.78) to[bend left=22] (0,3.91)to[bend right=14] (0,5)to[bend left=14] (0,6.91)to [bend left = 18] (4,3.91) to (4,-2.22)  to[bend right = 18](0,0.78);
		\draw[orange] (0,6.91) to [bend right = 20] (-0.2, 8.1);
		\draw[orange] (0,0.78) to [bend left = 15] (0, -1.2);
		\draw (0,3.91) to (-1,3);
		\draw (0,3.91) to (1,3.45);
		\draw[->] (6,7) to (8,7);
		\draw ([shift=(30:1cm)]-1,0.1) arc (190:361:0.15cm);
		\draw ([shift=(30:1cm)]-1,3.25) arc (190:361:0.15cm);
		\node at (0.2, 0.18) {\tiny{$\alpha_0$}};
		\node at (0.4, 3.05) {\tiny{${}_{\alpha(s)}$}};
		\node at (0,9.28) {\tiny{$x_3$}};
		\node at (7, 0.3) {\tiny{$x_2$}};
		\node at (-2.65, -3.5) {\tiny{$x_1$}};
		\node at (7, 6.5) {\tiny{$\Phi_1$}};
		\node at (3.6, -1.2) {\tiny{$S$}};
	\end{tikzpicture}
	\begin{tikzpicture}[scale= 0.43]
		\draw[gray,->] (0,0.78) to (0,8.78);
		\draw[gray,->] (0,0.78) to (7,0.78);
		\draw[gray,->] (0,0.78) to (-3,-2.72);
		\draw[blue] (-2, -2.33) to[bend left = 16] (2,2.33);
		\draw[blue] (-2,4.33) to[bend left = 16] (2, 8.33);
		\node at (-1.2, 4.2) {\tiny{$(\breve{\partial\Omega})$}};
		\draw[blue] (-2,-2.33) to (-2, 4.33);
		\draw[blue] (2,8.33) to (2,2.33);
		\draw[orange, pattern=north east lines,  pattern color=teal] (0,0.78) to (0,6.91) to [bend left = 18] (4,3.91) to (4,-2.22) to [bend right = 18] (0,0.78);
		\draw (0,3.91) to (-1,3);
		\draw (0,3.91) to (1,3.45);
		\draw[->] (6,7) to (8,7);
		\draw([shift=(30:1cm)]-1,0.1) arc (190:361:0.15cm);
		\draw ([shift=(30:1cm)]-1,3.25) arc (190:361:0.15cm);
		\node at (0.2, 0.18) {\tiny{$\alpha_0$}};
		\node at (0.6, 3.05) {\tiny{${}_{\alpha(s)}$}};
		\node at (0,9.28) {\tiny{$s$}};
		\node at (7, 0.3) {\tiny{$t$}};
		\node at (-2.65, -3.5) {\tiny{$r$}};
		\node at (7, 6.5) {\tiny{$\Phi_2$}};
		\node at (3.6, -1.1) {\tiny{$\breve{S}$}};
	\end{tikzpicture}
	\begin{tikzpicture}[scale= 0.43]
		\draw[gray,->] (0,0.78) to (0,8.78);
		\draw[gray,->] (0,0.78) to (7,0.78);
		\draw[gray,->] (0,0.78) to (-3,-2.72);
		\draw[blue] (-2, -1.55) to (2, 3.11);
		\draw[blue] (-2,5.11) to (2, 9.11);
		\draw[blue] (-2,-1.55) to (-2, 5.11);
		\draw[blue] (2,9.11) to (2,3.11);
		\draw[orange, pattern=north east lines,  pattern color=teal] (0,0.78) to (0,7.1) to (4,4.1) to (4,-2.22) to (0,0.78);
		\draw (0,3.91) to (-1,3);
		\draw (0,3.91) to (1,3.3);
		\draw([shift=(30:1cm)]-1,0.1) arc (190:361:0.15cm);
		\draw ([shift=(30:1cm)]-1,3.25) arc (190:361:0.15cm);
		\node at (0.2, 0.18) {\tiny{$\alpha_0$}};
		\node at (0.6, 3.05) {\tiny{$\alpha_0$}};
		\node at (0,9.28) {\tiny{$y_3$}};
		\node at (7, 0.3) {\tiny{$y_2$}};
		\node at (-1, -2) {\tiny{$\mathcal{D}^1_{\alpha_0}$}};
		\node at (5, 2) {\tiny{$\mathcal{D}^2_{\alpha_0}$}};
		\node at (-2.65, -3.5) {\tiny{$y_1$}};
		\node at (3.45, -1.2) {\tiny{$P_{\alpha_{0}}$}};
	\end{tikzpicture}
	\caption{An illustration of the coordinates-transformation $\Phi = \Phi_2\circ \Phi_1$ as a composition of the local diffeomorphisms $\Phi_1$ and $\Phi_2$, as defined in Section \ref{sec: coordinates discontinuity line}. The shaded regions respectively represent (from the left to the right) the surfaces $S$, $\breve{S}$ and $P_{\alpha_0}$ near $(0,0,0)$.}
	\label{fig:coordinates}
\end{figure}

Let now $g_0$ be the Riemann metric on $\mathbb{R}^3$. Under $\Phi_1$, the matrix $G_1$ of the metric $g_0$ satisfies 
\begin{equation*}\label{eq: approx matrix1}
	G_1^{-1} = Id + {  \mathcal{O}(|\mathbf{r}|)}.
\end{equation*}
The Jacobian $J_{{\Phi_{1}}}$ satisfies 
\begin{equation*}\label{eq: approx determinant1}
|J_{\Phi_{1}}|=	1 + \mathcal{O}(|\mathbf{r}|).
\end{equation*}
Note that $\Phi_1$  transforms the discontinuity surface $S$ near $x_0$ to a surface-neighborhood $\breve S$ of $(0,0,0)$  making an angle $\alpha(s):=\alpha_{\Phi_1^{-1}(0,0,s)}$ with the boundary $(\breve{\partial\Om})$ at the point $(0,0,s)$ (see Figure~\ref{fig:coordinates}). Clearly, $s\mapsto \alpha(s)$ is a smooth function.
\paragraph{\emph{The second normalization}} We now introduce the second change of coordinates, $\Phi_2$, in a neighborhood of $(0,0,0) = \Phi_1(x_0)$. First we underline that we can take the neighborhood of $(0,0,0)$ such that $(\breve{\partial\Omega})\cup\breve{S}$ is included in 
\[
	\{(r,t,s) \, \vert\, \phi_s(r)=t \,\,\,\mbox{or}\,\,\, r= \psi_s(t)\},
\]
where $\phi_s$ and $\psi_s$ are smooth functions depending smoothly on the parameter $s$ and satisfying 
\[
\phi_s(0)=0, \quad \phi'_s(0)=0,\quad	\psi_s(0) = 0, \quad \psi^\prime_s(0) = \mathrm{cot}(\alpha(s)).
\]
Following~\cite[Section 2.2.1]{PopRay}, we first introduce the change of variables 
\[
	(u, v) = C_s(r,t),
\]
where $C_s$ is a local diffeomorphism near $(0,0)$ defined by 
\begin{equation*}
	\begin{cases} u = r - \psi_s(t) + \cot\alpha(s)(t-\phi_s(r)), \\
	 v=t-\phi_s(r). \end{cases}
\end{equation*}
We then define a local diffeomorphism near $(0,0,0)$  by 
\[
	\breve{\Phi}(r,t,s) = (u,v,s) := (C_s(r,t),s),
\]
The matrix $\breve{G}$ associated to $\breve{\Phi}$ satisfies (see~\cite{PopRay})
\[
	\breve{G}^{-1} = Id + {  \mathcal{O}}(|\mathbf u|),
\]
where $|\mathbf{u}|$ denotes the norm of $(u,v,s)$.
The jacobian associated to $\breve{\Phi}$ satisfies
\[
	|J_{\breve{\Phi}}| = 1 + \mathcal{O}(|\mathbf{u}|).
\]
  Consequently, {  in $(\breve{\partial\Omega})\cup\breve{S}$} there exists a neighborhood of $(0,0,0)$  which is sent by $\breve{\Phi}$ to the following region
\begin{equation}\label{eq: region after breve phi}
	\{(u,v,s)\,\vert\, v= 0 \,\,\,\mbox{or}\,\,\, v= \tan\alpha(s) u\}.
\end{equation}
Finally, we want to replace the variable angle $\alpha(s)$  in~\eqref{eq: region after breve phi} with the constant opening angle $\alpha_0=\alpha(0)$.  To do that, we first perform a rotation $R_{-\alpha(s)/2}$ of angle $-\alpha(s)/2$ to get 
\[
	(\breve{u}, \breve{v}) = R_{-\alpha(s)/2} (u, v).
\] 
Let $\tau(s) = \tan (\alpha(s)/2)$ (notice that $\tau(0) = \tan(\alpha_0/2)$. We do the following rescaling 
\[
	{ \tilde{u} }= \breve{u}, \quad { \tilde{v}} = \tau(s)^{-1}\tau(0)\breve{v}, \quad s=s.
\]
We then perform an inverse rotation $R_{\alpha(s)/2}$ and define
\[
	(y_1,y_2):= R_{\alpha(s)/2}({ \tilde{u}, \tilde{v}}), \quad y_3=s.
\]
Now, we introduce the diffeomorphism $\Phi_2$ near $(0,0,0)$ by setting 
\[
	\Phi_2(r,t,s) = (y_1,y_2,y_3).
\]
By $\Phi_2$, one can map $(\breve{\partial\Omega})\cup\breve{S}$ near $(0,0,0)$ into a subset of
\[
	\{(y_1,y_2,y_3)\, \vert\, y_2 = 0\,\,\,\mbox{or}\,\,\,y_2 = y_1\tan\alpha_0 \}.
\]
\paragraph{\emph{The composition of the two normalization}} We define $\Phi$ in a neighborhood of $\mathcal N_{x_0}$ as the composition of $\Phi_2$ and $\Phi_1$ 
\begin{equation}\label{eq:Phi}
	\Phi=\Phi_2\circ\Phi_1:(x_1,x_2,x_3)\mapsto y=(y_1,y_2,y_3).\end{equation}
 One can see now that the properties of $\Phi$ in~\eqref{c1}--~\eqref{c5} hold true for a suitable $\mathcal N_{x_0}$ (see Figure \ref{fig:coordinates}). Moreover,
let $G$ be the matrix of the metric $g_0$ corresponding to $\Phi$, and $G^{-1}$ be its inverse. By the above discussion, we have
\begin{equation}\label{eq: approx matrix}
 G^{-1}= (g^{\ell m})=Id +{  \mathcal{O}}(|y|).
\end{equation}
The Jacobian $J_{{\Phi}}$ satisfies
\begin{equation}\label{eq: approx determinant}
	|J_{{\Phi}}|=	1 + \mathcal{O}(|y|).
\end{equation}
\paragraph{\emph{The quadratic form in the boundary coordinates}} We consider the diffeomorphism $\Phi$ defined above near $x_0$.
Under $\Phi$, the Lebesgue measure $dx$ transforms into $dx = { |J_{{\Phi}}|}dy$.

We denote any vector potential  $\mathbf{F}\in H^{1}(\Omega, \mathbb{R}^3)$  satisfying $\mathrm{curl}\mathbf{F} = \mathbf{B}$ (as in~\eqref{eq:B1}) by $\widetilde{\mathbf{F}}=(\widetilde{F}_1, \widetilde{F}_2, \widetilde{F}_3)$ in the new coordinates near $x_0$. We have 
\[
F_1 dx_1 + F_2dx_2 + F_3dx_3 = \widetilde{F}_1 dy_1 + \widetilde{F}_2 dy_2 + \widetilde{F}_3 dy_3.
\]
Denoting  the magnetic field $\mathbf{B}$ by $\widetilde{\mathbf{B}} $ in the local coordinates, we have (see~\cite[Section 5]{helffer2004magnetic})
\begin{equation}\label{eq: curl tilde F}
\mathrm{curl}\widetilde{\mathbf{F}} =|J_{\Phi}|\,\widetilde{\mathbf{B}}.
\end{equation}
In addition, for any $u\in H^1(\Om)$ supported in 	$\mathcal N_{x_0}$, the quadratic form $Q_{\mathfrak{b}, \mathbf{F}}(u)$ becomes
\begin{eqnarray}\label{eq: diff quad form}
	Q_{\mathfrak{b}, \mathbf{F}}(u) &=&\int_{\mathcal N_{x_0}\cap\Om} |(\nabla - i \mathfrak{b}\mathbf{F})u|^2 dx
	\\
	&=& \int_{\widetilde{\mathcal N}_{x_0}\cap\R_+^3} \, |J_{\Phi}| \,\left[\sum_{1\leq \ell,m \leq 3}{g^{\ell m}} (\partial_{y_\ell} - i\mathfrak{b}\widetilde{\mathbf{F}}_\ell)\widetilde{u}\overline{(\partial_{y_m} - i \mathfrak{b} \widetilde{\mathbf{F}}_m)\widetilde{u}}\right]\, dy\nonumber
\end{eqnarray}
where $\widetilde{\mathcal N}_{x_0}=\Phi(\widetilde{\mathcal N}_{x_0})$ and $\tilde{u}(\cdot) = u( \Phi^{-1}(\cdot))$.

 Finally, the following lemma presents a gauge transformation of the potential $\widetilde{\mathbf{F}}$ that will be useful in comparing the operator $(\nabla -i\mathfrak{b} {\mathbf{F}})^2$ (in~\eqref{eq: def semicl op}) with the leading operator $\mathcal{L}_{\alpha,\gamma,a}$ (in~\eqref{eq:La+}) near the discontinuity edge $\Gamma$. 
\begin{lemma}\label{lem: change gauge D4} Let $a\in [-1,1)\setminus \{0\}$ and $B(0,\ell)$ be a ball of radius $\ell$ such that $B(0,\ell)\cap\mathbb{R}^3_+\subset  \widetilde{\mathcal N}_{x_0}$. Consider the vector potential $\mathbf{F} \in H^1(\Omega)$ such that $\curl\mathbf{F} =  \mathbf{B}$,  where $\mathbf{B}$ is as in~\eqref{eq:B1}. There exists $\omega\in H^{2}(B(0, {\ell}) \cap \mathbb{R}^3_+)$ such that the vector potential $\widetilde{\mathbf{F}}$ can be written as
	\begin{equation*}
	\widetilde{\mathbf{F}}(y)= \mathbf{A}_{\alpha_{0},\gamma_{0},a} (y)+\nabla \omega(y) + {  \mathcal{O}(|y|^2)},\qquad \forall y\in B(0, \ell),
	\end{equation*}
	where $\mathbf{A}_{\alpha_{0},\gamma_{0},a}$ is the vector potential introduced in \eqref{eq:Ag}. Here $\alpha_0$ is the angle between the discontinuity surface $S$ and the tangent plane to $\partial\Omega$ at $x_0$, and $\gamma_0$ is the angle between the magnetic field $\mathbf{B}$ and the discontinuity edge $\Gamma$ at $x_0$.
\end{lemma}
\begin{proof} We consider the magnetic field $\mathbf{B}_{\alpha_{0},\gamma_{0},a}=\mathrm{curl}\mathbf{A}_{\alpha_0,\gamma_0,a}$ introduced in \eqref{eq:Bs}. Let now $\widetilde{\mathbf{B}}$ (resp. $\widetilde{\mathbf{F}}$) be the magnetic field $\mathbf{B}$ (resp. the vector potential $\mathbf{F}$) expressed in the local coordinates.
	
	 Notice that the magnetic field $\mathbf B$ is constant in each of $\Om_1$ and $\Om_2$. Indeed, it is equal to $\mathbf b$ in $\Om_1$ and $a\mathbf b$ in $\Om_2$, where $\mathbf b$ is a unit magnetic field making an angle $\alpha_0$ with the boundary and $\gamma_0$ with $\Gamma$ at $x_0$. Also, notice that $\Om_1$ is sent to $\mathcal D_{\alpha_0}^1$ and $\Om_2$ to $\mathcal D_{\alpha_0}^2$ near $x_0$ (see~\eqref{c5}). Consequently, using the properties of the local coordinates transformation and a Taylor expansion of the unit field yields (see~e.g.~\cite [Section~5]{lu2000surface} or~\cite{helffer2004magnetic})
	\begin{equation}\label{eq: Taylor B}
		\widetilde{\mathbf{B}} = \mathbf{B}_{\alpha_{0},\gamma_{0},a} +  {  \mathcal{O}}(|y|).
	\end{equation}
	 In addition, by~\eqref{eq: curl tilde F} and~\eqref{eq: approx determinant}, we have
	\begin{equation}\label{eq: curl tilde F1}
		\mathrm{curl}\widetilde{\mathbf{F}} = \widetilde{\mathbf{B}}\Big(1 + {  \mathcal{O}}(|y|)\Big).
	\end{equation}  
Then, \eqref{eq: Taylor B} and ~\eqref{eq: curl tilde F1} give
\begin{equation}\label{eq: curl of the difference}
	\mathrm{curl}\widetilde{\mathbf{F}} =\mathbf{B}_{\alpha_{0},\gamma_{0},a} + {  \mathcal{O}}(|y|).
\end{equation}
Using the inverse curl formula, we define a vector potential $\mathbf{A}_\mathrm{diff}$ such that 
\[
\mathrm{curl}\mathbf{A}_{\mathrm{diff}} = \mathrm{curl}(\widetilde{\mathbf{F}} - \mathbf{A}_{\alpha_0, \gamma_0,a}) = \mathrm{curl}\widetilde{\mathbf{F}} - \mathbf{B}_{\alpha_{0},\gamma_{0},a}.
\]
We take
\begin{equation}
	\mathbf{A}_{\mathrm{diff}}(y) = \int_0^1 {  \big(\mathrm{curl}\widetilde{\mathbf{F}}(\lambda y) - \mathbf{B}_{\alpha_{0},\gamma_{0},a}(\lambda y)\big)} \wedge (\lambda y) \, d\lambda.
\end{equation}
It follows from \eqref{eq: curl of the difference} that $\mathbf{A}_{\mathrm{diff}}(y) = {  \mathcal{O}}(|y|^2)$, for each $y\in B(0, \ell)\cap\R_+^3$. Moreover, since the two vector potentials $\widetilde{\mathbf{F}} - \mathbf{A}_{\alpha_0, \gamma_0, a}$ and $\mathbf{A}_{\mathrm{diff}}$ are generating the same magnetic field in $B(0, \ell)\cap\R_+^3$, there exists a function $\omega \in H^2(B(0, \ell) \cap \mathbb{R}^3_+)$ such that 
\begin{equation}
	\widetilde{\mathbf{F}} - \mathbf{A}_{\alpha_0, \gamma_0, a} = \mathbf{A}_{\mathrm{diff}} + \nabla \omega.
\end{equation}
Using that $\mathbf{A}_{\mathrm{diff}}(y) = { \mathcal{O}}(|y|^2)$ completes the proof.
\end{proof}	
\subsubsection{Local coordinates near the discontinuity surface $S$ (away from $\Gamma$)}\label{sec: local coordinates S}
Let $S$ be the discontinuity surface introduced earlier. Under our hypothesis, this is a smooth simply connected set with boundary $\Gamma$.  We will consider standard tubular coordinates  near $S$ away from the discontinuity edge $\Gamma$. 

Let $x_0\in S$ (away from $\Gamma$), we consider a small neighborhood $V_{x_0}$ of $x_0$ in $S$, in which a local coordinate system is well defined. We denote such coordinates by $(r,s)$ . Let $\phi$ be the local diffeomorphism corresponding to these coordinates
 \[ \phi: V_{x_0} \rightarrow U\subset\mathbb{R}^2,\ \mbox{such that}\ \phi(x) = (r,s).\]
 Denoting by $\mathbf{n}$ the unit normal to $S$ at the point $\phi^{-1}{ (r,s)}$, we  define the coordinate transformation $\Phi$ in a small neighborhood of $x_0$ such that
\begin{equation}
	(x_1, x_2, x_3) = \Phi^{-1}(r,s,t) = \phi^{-1}(r,s) + t\mathbf{n},
\end{equation}
where $t= \dist(x, S)$ if $x\in\Omega_1$ and $t= -\dist(x, S)$ if $x\in\Omega_2$. 

We denote by $y=(y_1,y_2,y_3)  =  (r,t,s)$. The coordinates $(y_1,y_2,y_3)$  have similar properties to those of the tubular coordinates defined in the previous section (see also e.g.~\cite{fournais2010spectral}).   Similar estimates to those in \eqref{eq: approx matrix}--\eqref{eq: diff quad form} are valid.

\subsection{Lower bound}\label{sec:low}

\begin{pro}[Lower bound of local energies]\label{pro: lower bound quadratic form} { Under Assumption~\ref{asu}}, there exist $C>0$ and $\mathfrak{b}_0 >0 $ such that for $\mathfrak{b}\geq \mathfrak{b}_0$, for all $R_0 >1$ and $u\in { \mathcal{D}(Q_{\mathfrak{b}, \mathbf{F}})}$, it holds
\begin{equation}\label{eq: est quadratic form lb}
	Q_{\mathfrak{b}, \mathbf{F}}(u) \geq \int_\Omega (U_{\mathfrak{b}}(x) - CR_0^{-2}\mathfrak{b}) |u(x)|^2\,dx,
\end{equation}
where
\begin{equation}
	U_{\mathfrak{b}}(x) = \begin{cases}   
		|a| \mathfrak{b} &\mbox{if}\,\,\,  \mathrm{dist} (x, \partial\Omega \cup S) {\geq} R_0\mathfrak{b}^{-\frac{1}{2}},  \\
		|a|\Theta_0 \mathfrak{b} - CR_0^4 \mathfrak{b}^{\frac{1}{2}} &\mbox{if}\,\,\,  \mathrm{dist}(x,\partial\Omega){ <} R_0\mathfrak{b}^{-\frac{1}{2}} {\leq} \mathrm{dist}(x, S), \\
		\beta_a \mathfrak{b} - CR_0^4 \mathfrak{b}^{\frac{1}{2}}&\mbox{if}\,\,\, \mathrm{dist}(x,S){ <} R_0\mathfrak{b}^{-\frac{1}{2}} { \leq } \mathrm{dist}(x, { \partial\Omega}),  \\
		\lambda_{\alpha_j,\gamma_j,a} \mathfrak{b} - CR_0^4 \mathfrak{b}^{\frac{1}{2}}&\mbox{if}\,\,\,  \mathrm{dist}(x, x_j){ <} R_0\mathfrak{b}^{-\frac{1}{2}},\,\,\, x_j\in\Gamma.
\end{cases}
\end{equation}
where $\Theta_0$, $\beta_a$, and $\lambda_{\alpha_j,\gamma_j,a}$ are respectively the values in ~\eqref{eq:teta0}, \eqref{eq:beta} and \eqref{eq: def lambda}, with $\alpha_j$ being the angle between the tangent plane of $\partial\Omega$  and the discontinuity surface $S$ at $x_j$ (taken towards $\Omega_1$) and $\gamma_j$ being the angle between the magnetic field $\mathbf{B}$ and the discontinuity edge $\Gamma$ at $x_j$.
\end{pro}
\begin{proof} We will work locally in $\Omega$, proceeding similarly as in \cite[Theorem 9.1.1]{fournais2010spectral}. Let $0 < \rho < 1$, we can define a partition of unity, $(\chi_{j})_{j\in \mathcal{I}}$ of $\R^3$, and assume that it satisfies the following. 
	
	For  $(\chi_j)$ with a support intersecting $\overline\Om$, we suppose that they are supported in  balls of centers $x_j\in\overline\Om$ and radius $R_0\mathfrak{b}^{-\rho}$, and
\begin{equation}\label{eq: prop partition unity}
		\chi_j \in C^\infty_0 (\mathbb{R}^3; \mathbb{R}), \qquad 
		\sum_j \chi_j^2 = 1, \qquad \sum_j |\nabla \chi_j|^2 \leq CR_0^{-2}\mathfrak{b}^{2\rho}.
	\end{equation}
	Moreover, we assume that each ball is either entirely contained in $\Omega$ or it is centered at  the boundary $\partial\Omega$. In the first case,   we further assume that the ball is either disjoint with $S$  or centered at $S$. In the second case, we assume that the ball is either disjoint with $\overline S=S\cup\Gamma$ or centered at $\Gamma$. We then define the following sets of centers $x_j$, which we assume to constitute a partition of the set of the entire centers: 
\begin{enumerate}
		\item $D_{1} := \big\{x_j\in \Omega \vert\, \supp(\chi_j) \cap (\partial\Omega \cup S) = \emptyset\big\}$, 
		\item $D_{2} := \big\{x_j\in\partial\Omega\,\vert\, \supp(\chi_j)\cap (S\cup \Gamma) = \emptyset\big\}$,
		\item $D_{3}:= \big\{x_j \in S\, \vert\,\supp(\chi_j)\cap \partial\Omega = \emptyset\big\}$,
		\item $D_{4} := \big\{x_j \in \Gamma\big\}$.
	\end{enumerate}
We also denote by $\mathcal{I}_\ell$, for $\ell=1,2,3,4$, the indices subsets of $\mathcal{I}$ such that $\mathcal{I}_\ell := \{j\in \mathcal{I}:x_j\in D_\ell\}$.
We now use the IMS formula to write
\begin{equation}\label{eq: IMS}
	Q_{\mathfrak{b}, \mathbf{F}}(u)=\sum_{\ell=1}^4\sum_{j\in\mathcal{I}_\ell} Q_{\mathfrak{b}, \mathbf{F}}(\chi_ju) - \sum_{j}\| |\nabla\chi_j| u\|_{L^2(\Omega)}^2 .
\end{equation}
First,  using the properties in \eqref{eq: prop partition unity}, we can estimate the (error) term in the RHS of the equality above as follows
\begin{equation}\label{eq: est IMS gradient}
	\sum_j\| |\nabla\chi_j| u \|_{L^2(\Omega)}^2 \leq CR_0^{-2}\mathfrak{b}^{2\rho} \|u\|_{L^2(\Omega)}^2.
\end{equation} 
Next, we estimate the (main) term in the foregoing RHS, i.e. $\sum_{\ell=1}^4\sum_{j\in\mathcal{I}_\ell} Q_{\mathfrak{b}, \mathbf{F}}(\chi_ju) $.\\

\noindent\textit{Estimates in  $D_1$.}  For $j\in\mathcal{I}_1$, we work as in Theorem \cite[Theorem 9.1.1]{fournais2010spectral}, and get
\begin{equation}
	 Q_{\mathfrak{b}, \mathbf{F}}(\chi_j u) \geq \min(1,|a|) \mathfrak{b}\sum_{i\in\mathcal{I}_1} \int_\Omega dx\, |\chi_j u|^2 \geq |a| \mathfrak{b}\sum_{i\in\mathcal{I}_1} \int_\Omega dx\, |\chi_j u|^2, 
\end{equation}
having $|a|\leq 1$.\\

\noindent\textit{Estimates in  $D_2$.} Let $j\in\mathcal{I}_2$. Mainly (\cite[Theorem 9.1.1]{fournais2010spectral}), local tubular coordinates can be defined near $\partial\Om$ away from the discontinuity zone $\overline S$, which permits to   compare the quadratic form with the one associated to the model operator introduced in Section \ref{sec:nu+}. Using the monotonicity properties of the bottom of the spectrum of the aforementioned model operator (see  Section \ref{sec:nu+}), one can write 
\begin{eqnarray}
	 Q_{\mathfrak{b}, \mathbf{F}}(\chi_j u) &\geq & |a| \Theta_0 {  \mathfrak b}\int_\Omega dx\,  |\chi_j u|^2 - Cc_1(R_0,\mathfrak{b})\|\chi_j u\|^2,
\end{eqnarray}
where 
\begin{equation}\label{eq: def c1}
	c_1(R_0, \mathfrak{b}) = R_0\mathfrak{b}^{1-\rho} + R_0^4\mathfrak{b}^{2-4\rho} + R_0^2\mathfrak{b}^{\frac{3}{2} -2\rho}.
\end{equation}
\noindent\textit{Estimates in  $D_3$.} Now, we consider the balls centered at the discontinuity surface $S$ and do not meet the boundary $\partial\Om$. In this situation, the local coordinates in Section \ref{sec: local coordinates S} are involved. For any $j\in \mathcal I_3$, we take $\mathfrak{b}$ sufficiently large so that the ball $B(x_j, R_0\mathfrak{b}^{-\rho})$ is included in the domain of the local diffeomorphism $\Phi=\Phi_{x_j}$, introduced in Section~\ref{sec: local coordinates S} and where $\Phi^{-1}(x_j) = (0,0,0)$.  Using the properties of $\Phi$, one gets that $\Phi(B(x_j, R_0\mathfrak{b}^{-\rho}))\subset B(0, cR_0\mathfrak{b}^{-\rho})$ for some positive constant $c>0$. 
Moreover,  for any $ v\in \mathcal D(Q_{\mathfrak{b}, \mathbf{F}})$ s.t. $ \supp v\subset B(x_j, R_0\mathfrak{b}^{-\rho})$, we have
\begin{multline}\label{eq3}
	(1- CR_0\mathfrak{b}^{-\rho})\int_{ \Phi\big(B(x_j, R_0\mathfrak{b}^{-\rho})\big)} dy\,   |(\nabla -i\mathfrak{b} \widetilde{\mathbf{F}}){ \tilde v}|^2 \leq Q_{\mathfrak{b}, \mathbf{F}}({ v}) 
	\\
	\leq (1+ CR_0\mathfrak{b}^{-\rho})\int_{ \Phi\big(B(x_j, R_0\mathfrak{b}^{-\rho})\big)} dy\,  |(\nabla -i\mathfrak{b} \widetilde{\mathbf{F}})\tilde v|^2,
\end{multline}
where $\widetilde{\mathbf{F}}$ and $ \tilde v$ are respectively the transforms of the vector potential $\mathbf{F}$ and the function $v$ by $\Phi$.

Similarly as in \cite{FKP}, we can perform a change of gauge that replaces the vector potential $\widetilde{\mathbf{F}}$ by a new one  $\hat{\mathbf{F}}=(\hat{F}_1,\hat{F}_2,\hat{F}_3)$ ( $\widetilde{\mathbf{F}} = \hat{\mathbf{F}} + \nabla \varphi$ for a certain $H^2$-function $\varphi$) such that
	\begin{equation*}\label{eq: gauge D3}
	\hat{F}_1 = \begin{cases}{ -} y_2   + \mathcal{O}(|y|^2), &\mbox{if}\,\,\,y_2 >0 ,\\ 
		{ -}a y_2+ \mathcal{O}(|y|^2), &\mbox{if}\,\,\,y_2 <0. \end{cases}, \quad \hat{F}_2 = \mathcal{O}(|y|^2),\ \hat{F}_3 = \mathcal{O}(|y|^2).
\end{equation*}
Above, we used that the magnetic field is  tangent to the surface $S$.  
We consider the function  $\hat{v} = e^{i\mathfrak{b}\varphi}\widetilde{v}$. By gauge invariance properties of the energy, we have 
\begin{equation}\label{eq2}
	\int_{\Phi{\big( B(x_j, R_0\mathfrak{b}^{-\rho})\big)}} dy\,   |(\nabla -i\mathfrak{b} \widetilde{\mathbf{F}})\tilde v|^2 =\int_{\Phi{\big( B(x_j, R_0\mathfrak{b}^{-\rho})\big)}} dy\,   |(\nabla -i\mathfrak{b} \hat{\mathbf{F}})\hat{v}|^2 .\end{equation}
Now we use Cauchy's inequality in the RHS of the equation above to further replace the vector potential $\hat{\mathbf{F}}$ by  the vector potential $ \mathsf{A}_a$ defined (in Section~\ref{sec:eff2}) by $\mathsf{A}_a(y)= ( -\delta_a(y)y_2,0,0)$ with $\delta_a = \mathbbm{1}_{y_2 >0} + a \mathbbm{1}_{y_2 <0}$. We then get 
\begin{multline}
\int_{\Phi{\big( B(x_j, R_0\mathfrak{b}^{-\rho})\big)}} dy\,  |(\nabla -i\mathfrak{b} \hat{\mathbf{F}})\hat{v}|^2 
\\
\geq  (1- \mathfrak{b}^{-\delta})\int_{\mathbb{R}^3}dy\,|(\nabla - i \mathfrak{b}\mathsf{A}_a)\hat{v}|^2 - C\mathfrak{b}^2(R_0^2 \mathfrak{b}^{-2\rho})^2 \mathfrak{b}^{\delta}\int_{\mathbb{R}^3}dy\, |\hat{v}|^2,
\end{multline}
for $\delta\in(0,1)$.
Using Remark \ref{rem: beta a R3}, we have
\begin{equation}\label{eq1}
	\int_{\mathbb{R}^3}dy\,|(\nabla - i \mathsf{A}_a)\hat{v}|^2 \geq \beta_a \int_{\mathbb{R}^3} dy\,|\hat{v}|^2, 
\end{equation}
where $\beta_a$ is the bottom of the spectrum of the operator $(\nabla - i \mathsf{A}_a)^2$ in Section~\ref{sec:eff2}. 
Then, using~\eqref{eq3},~\eqref{eq2}, and a scaling argument in~\eqref{eq1}, we get for $v=\chi_ju$
\begin{equation}\label{eq:4}
	Q_{\mathfrak{b}, \mathbf{F}}(\chi_j u) \geq \big(\beta_a \mathfrak{b} - Cc_2(R_0, \mathfrak{b})\big)\int_{\Omega}dx\,  |\chi_j u|^2,
\end{equation}
where $c_2(R_0, \mathfrak{b})$ is a constant depending on $R_0$ and $\mathfrak{b}$ which explicitly reads as 
\begin{equation}\label{eq: def c2}
	c_2(R_0, \mathfrak{b}) = R_0\mathfrak{b}^{1-\rho} + \mathfrak{b}^{1-\delta} + R_0^4\mathfrak{b}^{2-4\rho+\delta},
\end{equation}
and where, recalling  the Jacobian, $J_{\Phi^{-1}}$, of the change of coordinates function $\Phi^{-1}$,  we used that 
\[
	\int_{\mathbb{R}^3}dy\, |\hat{v}|^2 = \int_{\mathbb{R}^3}dy\, |\tilde v|^2 = \int_{\Om} dx\, |\chi_j u|^2 (J_{\Phi^{-1}}),
\]
together with the estimate \eqref{eq: approx determinant}.\\

\noindent \textit{Estimates in  $D_4$.}  Finally, we consider  the balls centered at the discontinuity edge $\Gamma$. For $j\in\mathcal I_4$, we have $\supp \chi_j \subset B(x_j, R_0\mathfrak{b}^{-\rho})$ where $x_j\in\Gamma$.  Here, the proof outline is quite similar to the one above (for the  balls centered in $D_3$), but with using  the local coordinates and the notation introduced in Section~\ref{sec: coordinates discontinuity line}. Thus, we will omit some computation details. Let $\Phi=\Phi_{x_j}$ be the diffeomorphism introduced in the foregoing section with $\Phi(x_j) = (0,0,0)$.  We suppose that there exists a positive constant $c>0$ such that ${\Phi}\big({B(x_j, R_0\mathfrak{b}^{-\rho})}{\cap\Om}\big)\subset B(0, cR_0\mathfrak{b}^{-\rho})$. We consider $ v\in \mathcal D(Q_{\mathfrak{b}, \mathbf{F}})$ s.t. $ \supp v\subset B(x_j, R_0\mathfrak{b}^{-\rho})$. As above, using the estimates~\eqref{eq: approx matrix} and~\eqref{eq: approx determinant}, we can write 
	\begin{multline}
	(1- CR_0\mathfrak{b}^{-\rho})\int_{ \Phi\big( B(x_j, R_0\mathfrak{b}^{-\rho})\cap\Om\big)}  dy\, |(\nabla -i\mathfrak{b} \widetilde{\mathbf{F}}){ \tilde{v}}|^2 \leq Q_{\mathfrak{b}, \mathbf{F}}(v)\\ \leq (1+ CR_0\mathfrak{b}^{-\rho})\int_{\Phi\big( B(x_j, R_0\mathfrak{b}^{-\rho})\cap\Om\big)} dy\, |(\nabla -i\mathfrak{b} \widetilde{\mathbf{F}})\tilde{v}|^2.
\end{multline}
{ Using now the gauge transform result in Lemma \ref{lem: change gauge D4} with $B(0,\ell)=B(0, cR_0\mathfrak{b}^{-\rho})$ and $x_0=x_j$ in this lemma, then applying Cauchy's inequality, we get for  $\delta \in(0,1)$}
\begin{multline}\label{eq: Cauchy 2 D4}
	 \int_{\Phi\big( B(x_j, R_0\mathfrak{b}^{-\rho})\cap\Om\big)} dy\, |(\nabla - i\mathfrak{b} \widetilde{\mathbf{F}})\tilde{v}|^2 
	 \\
	 \geq  (1- \mathfrak{b}^{-\delta})\int_{\Phi\big( B(x_j, R_0\mathfrak{b}^{-\rho})\cap\Om\big)} dy\, |(\nabla - i\mathfrak{b} \mathbf{A}_{\alpha_j,\gamma_j,a})\hat{v}|^2 - CR_0^4\mathfrak{b}^{2-4\rho + \delta}\int_{\mathbb{R}^3}dy\, |\hat{v}|^2,
\end{multline}
where  $\hat{v} = e^{i\mathfrak b\omega}\tilde{v}$ and $\omega$ is the gauge function in Lemma \ref{lem: change gauge D4}.   Using~\eqref{eq:Lalfa} and~\eqref{eq: def lambda} and proceeding similarly as in establishing~\eqref{eq:4} above, we get
\begin{equation}
	Q_{\mathfrak{b}, \mathbf{F}}(\chi_j u) \geq \big(\mathfrak{b}\lambda_{\alpha_j,\gamma_j, a} - C c_3(R_0, \mathfrak{b}))\int_{\mathbb{R}^3} dx\, |\chi_j u|^2,
\end{equation}
where 
\begin{equation}\label{eq: def c3}
	c_3(R_0, \mathfrak{b}) = R_0\mathfrak{b}^{1-\rho} + \mathfrak{b}^{1-{\delta}} + R_0^4\mathfrak{b}^{2-4\rho + {\delta}}.
\end{equation}

Finally, we choose $\rho = \delta = 1/2$ and $R_0 > 1$ large. This yields
\begin{equation}
	|c_j(R_0, \mathfrak{b})|\leq CR_0^4\mathfrak{b}^{\frac{1}{2}}, \quad\mbox{for}\  j=1,2,3.
\end{equation}
Gathering the results in the proof above establishes the proposition statement.
\end{proof}
 Actually, the proof of Proposition~\ref{pro: lower bound quadratic form} yields the following particular result on the lower bound of the ground state energy:
\begin{pro}[Lower bound of $\lambda(\mathfrak{b})$]\label{pro: lower bound bottom} Under Assumption \ref{asu}, there exist $C>0$ and $\mathfrak{b}_0 >0$ such that for all $\mathfrak{b}\geq \mathfrak{b}_0$, it holds
\begin{equation}
	\lambda(\mathfrak{b}) \geq \mathfrak{b} \inf_{x\in \Gamma} \lambda_{\alpha_x,\gamma_x,a} - C\mathfrak{b}^{\frac{3}{4}}.
\end{equation}
\end{pro}
\begin{proof}
The result is a consequence of the min-max principle in~\eqref{eq:minmax}, the proof in Proposition \ref{pro: lower bound quadratic form} (taking $\rho= 3/8$, $\delta = 1/4$, $R_0 = 1$ in the proof of this proposition), and the fact that  
\begin{equation}\label{eq: lb lambda b}
	\inf_{x\in\Gamma} \lambda_{\alpha_x,\gamma_x,a}<\min (|a|, \Theta_0 |a|, \beta_a) =\Theta_0|a|
\end{equation}
under Assumption~\ref{asu}.
\end{proof}
\subsection{Upper bound}
In this section, we will establish  an upper bound of $\lambda(\mathfrak{b})$.  Note that this upper bound is not optimal. However it will be sufficient  to get the right order in  the localization estimates in Theorem \ref{thm: localization}.
\begin{pro}[Upper bound of $\lambda(\mathfrak{b})$]\label{pro: upper bound lambda b} Let $\overline{x}\in D$. Under Assumption \ref{asu}, there exists $C_{\overline{x}}>0$ depending on $\overline{x}$ and $\mathfrak{b}_0 >0$ such that for all $\mathfrak{b}\geq \mathfrak{b}_0$, it holds
	\begin{equation}\label{eq: ub lambda b x bar}
	\lambda(\mathfrak{b}) \leq \mathfrak{b} \lambda_{\alpha_{\overline{x}},\gamma_{\overline{x}},a} + C_{\overline{x}}\mathfrak{b}^{\frac{5}{6}}.
	\end{equation}
\end{pro}
\begin{proof}
	Recall the set $D$ defined in~\eqref{eq:D} by
	\begin{equation*}
		D= \big\{ \overline{x}\in \Gamma\,\,\, \vert\,\,\, \lambda_{\alpha_{\overline{x}}, \gamma_{\overline{x}}, a} < |a| \Theta_0\big\}.
	\end{equation*}
Let $\overline{x}\in D$. Below, we define a suitable trial state  supported near $\overline x$. In a convenient neighborhood $\overline{\mathcal{N}}$ of $\overline{x}$, we can consider the local coordinates $(y_1, y_2,y_3)$ introduced in Section \ref{sec: coordinates discontinuity line}, with the related diffeomorphism  $\Phi=\Phi_{\overline x}$ which satisfies $\Phi(\overline x)=(0,0,0)$.  For $\mathfrak{b}>0$ sufficiently large, we can assume that 
	\begin{equation*}
		\mathcal N_{\rho,\eta}:= (-\mathfrak{b}^{-\rho},\mathfrak{b}^{-\rho}) \times  (0,\mathfrak{b}^{-\rho}) \times (-\mathfrak{b}^{-\eta},\mathfrak{b}^{-\eta})\subset \Phi(\overline{\mathcal{N}}),
	\end{equation*} 
	for some $\rho\in (0,1/2)$ and $\eta > 0$ to be fixed later. We also set $\mathcal{N}_{\rho} := (-\mathfrak{b}^{-\rho}, \mathfrak{b}^{-\rho})\times (0,\mathfrak{b}^{-\rho})$ for a later use. Let $\chi\in C^\infty(\mathbb{R})$ such that $\|\chi\|_2=1$ and 
	\[
	0\leq \chi\leq 1, \qquad \chi = 1 \,\,\mbox{in}\,\, ({-1/2},1/2), \quad\mbox{and}\ \supp\chi\subset ({-1},1).
	\] 
	We define the functions $\chi_\rho(\cdot) = \chi(\mathfrak{b}^\rho\cdot)$ and $\chi_\eta(\cdot) = \chi(\mathfrak{b}^\eta\cdot)$.
	
	Under Assumption~\ref{asu}, we can consider a normalized eigenfunction, $\overline{v}$, of the operator $\underline{\mathcal{L}}_{\alpha_{\overline{x}},{\gamma}_{\overline{x}}, a} + V_{\underline{\mathbf{B}}_{{\alpha}_{\overline{x}},{\gamma}_{\overline{x}}, a}, \tau_\ast}$ mentioned in Theorem~\ref{thm:main},  for some $\tau^\ast\in \mathbb{R}$, and satisfying $\underline\sigma(\alpha_{\overline{x}},\gamma_{\overline{x}},a,\tau_\ast)=\lambda_{\alpha_{\overline{x}}, \gamma_{\overline{x}},a}$. The existence of $\overline{v}$ is ensured by Theorem \ref{thm:main} (see Remark~\ref{rem:asu}).
	
	We define the function  $\overline{u}$ by\footnote{This rescaling is useful to get the leading order in the upper bound of  $\lambda(\mathfrak{b})$ as $\mathcal O(\mathfrak b)$.} $\overline{u}(y)= \sqrt{\mathfrak{b}}\overline{v}(\sqrt{\mathfrak{b}} y)$. 
	We then introduce our trial function $u_{\mathrm{trial}}$ such that
	\begin{equation*}
		u_{\mathrm{trial}}(x) = \begin{cases}  (\tilde{u}\circ \Phi)(x) &\mbox{if} \,\,\, x\in \Phi^{-1}({ \mathcal N_{\rho,\eta}})\cap \Omega,\\ 0 &\mbox{otherwise.}  \end{cases}
	\end{equation*}
	with
	\begin{equation*}
		\tilde{u}(y)=\tilde{u}(y_1,y_2,y_3) := \begin{cases}\mathfrak{b}^{\eta/2}\chi_\eta(y_3)\chi_\rho(y_1)\chi_\rho(y_2) \overline{u}(y_1,y_2) e^{i \mathfrak{b}\omega(y)} e^{i\sqrt{\mathfrak{b}}\tau_\ast y_3} &\mbox{if}\,\,\, y\in {\mathcal N_{\rho,\eta}}\cap\mathbb{R}^3_+ ,\\ 0 &\mbox{otherwise.} \end{cases}
	\end{equation*}
	where $\omega$ is the gauge function in Lemma \ref{lem: change gauge D4}  and the normalization factor $\mathfrak{b}^{\eta/2}$ ensures that $\|\mathfrak{b}^{\eta/2}\chi_\eta\|_2 = 1$.

	Next, we evaluate $Q_{\mathfrak{b}, \mathbf{F}}(u_{\mathrm{trial}})$ and $\|u_{\mathrm{trial}}\|_2$. In what follows, the constants in the estimates will depend on the point $\overline{x}$. But for simplicity, we omit this dependence and write $C$ instead of $C_{\overline{x}}$ to denote these constants. From the definition of $u_{\mathrm{trial}}$ and the property of the diffeomorpism $\Phi$ (see \eqref{eq: approx determinant}), we can write, doing a change of variables, 
	\begin{eqnarray*}
		\int_{\Omega}dx\, |u_{\mathrm{trial}}|^2 &\geq& (1-C(\mathfrak{b}^{-\rho} + \mathfrak{b}^{-\eta}))\int_{{\mathcal N_{\rho,\eta}}\cap \mathbb{R}^3_+} dy\, |\tilde{u}|^2 
		\\
		&=&  (1-C(\mathfrak{b}^{-\rho} + \mathfrak{b}^{-\eta}))\int_{{\mathcal N_{\rho}}\cap \mathbb{R}^2_+} dy_1dy_2\, |\chi_\rho(y_1)\chi_\rho(y_2)\overline{u}(y_1,y_2)|^2 
		\\
		&=& (1-C(\mathfrak{b}^{-\rho}+\mathfrak{b}^{-\eta})\int_{\mathcal N_{\rho-1/2}\,\cap \mathbb{R}^2_+} dy_1dy_2\,|\chi_\rho(\mathfrak{b}^{-\frac{1}{2}}y_1)\chi_\rho(\mathfrak{b}^{-\frac{1}{2}}y_2)\overline{v}(y_1,y_2)|^2 
		\\
		&\geq& (1-C(\mathfrak{b}^{-\rho} +\mathfrak{b}^{-\eta}))\left(1-\int_{\mathcal N^c_{\rho-1/2}\,\cap \mathbb{R}^2_+} dy_1dy_2\,|\overline{v}(y_1,y_2)|^2 \right),
	\end{eqnarray*}
	where $\mathcal N^c_{\rho - 1/2}$ denotes the complement of $\mathcal N_{\rho -1/2}$. For $\rho \in (0,1/2)$, using the exponential decay of $\overline{v}$ stated in Theorem \ref{thm:v0-decay}, one gets that
	\begin{equation*}
		\int_{\Omega}dx\, |u_{\mathrm{trial}}|^2 \geq (1-C\mathfrak{b}^{-\rho} - C\mathfrak{b}^{-\eta}).
	\end{equation*}
	We now estimate $Q_{\mathfrak{b}, \mathbf{F}}(u_{\mathrm{trial}})$.
	Similarly as in Proposition \ref{pro: lower bound quadratic form}, using \eqref{eq: approx matrix} and \eqref{eq: approx determinant}, we can write
	\begin{equation}\label{eq: quadr form ub}
		Q_{\mathfrak{b}, \mathbf{F}}(u_{\mathrm{trial}}) \leq (1+ C\mathfrak{b}^{-\rho} +  C\mathfrak{b}^{-\eta})\int_{ \mathbb{R}^3_+} dy\, |(\nabla - i \mathfrak{b}\widetilde{\mathbf{F}})\tilde{u}|^2,
	\end{equation}
	where  $\widetilde{\mathbf{F}}$ is the transform of the vector potential $\mathbf{F}$ by $\Phi$.  Using the change of gauge in Lemma \ref{lem: change gauge D4}, we have
	\begin{multline}\label{eq: change gauge ub}
		\int_{\mathbb{R}^3_+}dy\, |(\nabla - i \mathfrak{b}\widetilde{\mathbf{F}})\tilde{u}|^2 
		\\
		 = \mathfrak{b}^{\eta}\int_{\mathbb{R}^3_+}dy\,  |(\nabla - i \mathfrak{b}\big(\mathbf{A}_{\overline{\alpha},\overline{\gamma},a}(y_1,y_2) + { {\mathcal{O}}}(|y|^2)\big)(\chi_\eta(y_3)\chi_\rho(y_1)\chi_\rho(y_2) \overline{u}(y_1,y_2))e^{i\sqrt{\mathfrak{b}}\tau_\ast y_3}|^2,
	\end{multline}
	where $\mathbf{A}_{\overline{\alpha},\overline{\gamma},a}$ is the vector potential in this lemma (introduced in \eqref{eq:Ag}), for $\overline\alpha=\alpha_{\overline x}$ and $\overline\gamma=\gamma_{\overline x}$. $\mathbf{A}_{\overline{\alpha},\overline{\gamma},a}$  was chosen independent of $y_3$. We now use Cauchy's inequality to write for some $0< \epsilon < 1$, 
\begin{multline}\label{eq: cauchy upper bound}
\mathfrak{b}^{\eta}\int_{\mathbb{R}^3_+}dy\,  |(\nabla - i \mathfrak{b}\big(\mathbf{A}_{\overline{\alpha},\overline{\gamma},a}(y_1,y_2) + { {\mathcal{O}}}(|y|^2)\big)(\chi_\eta(y_3)\chi_\rho(y_1)\chi_\rho(y_2) \overline{u}(y_1,y_2)e^{i\sqrt{\mathfrak{b}}\tau_\ast y_3})|^2
\\
\leq \mathfrak{b}^{\eta}(1+\mathfrak{b}^{-\epsilon})\int_{\mathbb{R}^3_+}dy\,  |(\nabla - i \mathfrak{b}\mathbf{A}_{\overline{\alpha},\overline{\gamma},a}(y_1,y_2)) \overline{u}(y_1,y_2)e^{i\sqrt{\mathfrak{b}}\tau_\ast y_3}|^2|\chi_\eta(y_3)|^2
\\
+C\mathfrak{b}^{\epsilon + \eta}\|\chi_\eta^\prime\|_\infty^2 \int_{\frac{1}{2}\mathfrak{b}^\eta \leq |y_3| \leq \mathfrak{b}^{\eta}}dy_3 \int_{\mathbb{R}^2_+}\,dy_1dy_2\,   |\overline{u}(y_1,y_2)|^2 
\\
+  C\mathfrak{b}^{\epsilon + \eta} \int_{\mathbb{R}^3_+}dy\,  |\chi_\eta(y_3)|^2\bigg(|\chi_\rho^\prime(y_1)|^2|\chi_\rho(y_2)|^2 + |\chi_\rho(y_1)|^2|\chi^\prime_\rho(y_2)|^2\bigg)|\overline{u}(y_1,y_2)|^2
\\
+C \mathfrak{b}^{ 2 +\epsilon +\eta} \int_{\mathbb{R}^3_+}dy\, |y|^4 |\chi_\eta(y_3)|^2| \overline{u}(y_1,y_2)|^2.
\end{multline}
We now take into account separately the integrals in the RHS of the inequality above. Via a change of coordinates, 
	\begin{multline*}
	\mathrm{I} := \mathfrak{b}^{\eta}(1+\mathfrak{b}^{-\epsilon})\int_{\mathbb{R}^3_+}dy\,  |(\nabla - i \mathfrak{b}\mathbf{A}_{\overline{\alpha},\overline{\gamma},a}(y_1,y_2)) \overline{u}(y_1,y_2)e^{i\sqrt{\mathfrak{b}}\tau_\ast y_3}|^2|\chi_\eta(y_3)|^2
		\\
		= \mathfrak{b}^{1+\eta}(1+\mathfrak{b}^{-\epsilon})\int_{\mathbb{R}^3_+}dy\,  |(\nabla - i \mathfrak{b}\mathbf{A}_{\overline{\alpha},\overline{\gamma},a}(y_1,y_2)) (\overline{v}(\sqrt{b}(y_1,y_2))e^{i\sqrt{\mathfrak{b}}\tau_\ast y_3})|^2|\chi_\eta(y_3)|^2
		\\
		= \mathfrak{b}^{1+\eta}(1 + \mathfrak{b}^{-\epsilon})\int_{\mathbb{R}}dy_3 |\chi_\eta(y_3)|^2  \underline{Q}^{\tau_\ast}_{\overline{\alpha},\overline{\gamma},a}(\overline{v})
		\\
		= \mathfrak{b}(1 + \mathfrak{b}^{-\epsilon}) \underline{Q}^{\tau_\ast}_{\overline{\alpha},\overline{\gamma},a}(\overline{v}), 
	\end{multline*}
	where $\underline{Q}^{\tau_\ast}_{\overline{\alpha},\overline{\gamma},a}$ is the quadratic form defined in \eqref{eq:Qa} and where we used that $\mathfrak{b}^{-\eta} = \|\chi_\eta\|_2^2$. Using that 
	\[
		 \underline{Q}^{\tau_\ast}_{\overline{\alpha}, \overline{\gamma},a}(\overline{v}) = \underline{\sigma}(\overline{\alpha}, \overline{\gamma}, a, \tau_\ast) = \lambda_{\overline{\alpha}, \overline{\gamma},a}, 
	\]
	we conclude that 
	\begin{equation}\label{eq: est A}
		\mathrm{I} = \mathfrak{b}(1 + \mathfrak{b}^{-\epsilon})\lambda_{\overline{\alpha}, \overline{\gamma},a} .
	\end{equation}
Next, let
\[
	 \mathrm{II}:= C\mathfrak{b}^{\epsilon + \eta} \|\chi_\eta^\prime\|_\infty^2 \int_{\frac{1}{2}\mathfrak{b}^\eta \leq |y_3| \leq \mathfrak{b}^{\eta}}dy_3 \int_{\mathbb{R}^2_+}\,dy_1dy_2\,   |\overline{u}(y_1,y_2)|^2.
\]
 We bound $\|\chi_\eta^\prime\|_{\infty}\leq C\mathfrak{b}^{\eta}$, and  consequently get
\begin{equation}\label{eq: est B}
	|\mathrm{II}| \leq C\mathfrak{b}^{\epsilon + 2\eta}\int_{\mathbb{R}^2_+}dy_1 dy_2\, |\overline{v}(y_1,y_2)|^2 \leq C\mathfrak{b}^{\epsilon + 2\eta}.
\end{equation}
For the term 
\[
	\mathrm{III}: = C\mathfrak{b}^{\epsilon + \eta} \int_{\mathbb{R}^3_+}dy\,  |\chi_\eta(y_3)|^2\bigg(|\chi_\rho^\prime(y_1)|^2|\chi_\rho(y_2)|^2 + |\chi_\rho(y_1)|^2|\chi^\prime_\rho(y_2)|^2\bigg)|\overline{u}(y_1,y_2)|^2,
\]
we use the support of $\chi_\rho^\prime$ and the exponential decay of $\overline{u}(y_1,y_2) = \sqrt{\mathfrak{b}} \overline{v}(\sqrt{\mathfrak{b}}(y_1,y_2))$ in Theorem \ref{thm:v0-decay}, and get 
\begin{equation}\label{eq: est C}
	|\mathrm{III}| \leq C\mathfrak{b}^{\epsilon}.
\end{equation}
To bound the term, 
\[
	\mathrm{IV} := C \mathfrak{b}^{ 2 +\epsilon +\eta} \int_{\mathbb{R}^3_+}dy\, |y|^4 |\chi_\eta(y_3)|^2| \overline{u}(y_1,y_2)|^2,
\]
we use again the decay of $\overline{v}$ from Theorem \ref{thm:v0-decay}, which implies that $\int |y|^4 |\overline{v}|^2 \leq C$ . We write 
\begin{equation}\label{eq: est D}
	|\mathrm{IV}| \leq C\mathfrak{b}^{2+\epsilon + \eta}\int_{\mathbb{R}^3_+} dy\, (|y_1|^4 + |y_2|^4 + |y_3|^4) |\chi_\eta(y_3)|^2 |\overline{u}(y_1,y_2)|^2  \leq
	C\mathfrak{b^{2+\epsilon - 4\eta}} + C\mathfrak{b}^{\epsilon}.
\end{equation}
Inserting~\eqref{eq: est A}--~\eqref{eq: est D} in \eqref{eq: change gauge ub}, we get
	\begin{equation}\label{eq: final ub 1}
		\int_{\mathcal{N}_{\rho,\eta}}dy |(\nabla - i\mathfrak{b}\widetilde{\mathbf{F}})\tilde{u}|^2 \leq \mathfrak{b}(1 + \mathfrak{b}^{-\epsilon})\lambda_{\alpha_j,\gamma_j, a} + C\mathfrak{b}^{\epsilon + 2\eta} + C\mathfrak{b}^{2+\epsilon - 4\eta}.
	\end{equation}
	Then, using the bound \eqref{eq: final ub 1} in \eqref{eq: quadr form ub}, we have  for $\rho\in (0,1/2)$, $\epsilon \in (0,1)$, and $\eta> 0$,
	\begin{multline*}
	Q_{\mathfrak{b}, \mathbf{F}}(u_{\mathrm{trial}}) \leq (1 + C\mathfrak{b}^{-\rho} + C\mathfrak{b}^{-\eta})\left[ \mathfrak{b}(1 + \mathfrak{b}^{-\epsilon})\lambda_{\alpha_j,\gamma_j, a} + C\mathfrak{b}^{\epsilon + 2\eta} + C\mathfrak{b}^{2+\epsilon - 4\eta} \right].
	\end{multline*}
	Taking $\epsilon = \rho = 1/6$ and $\eta = 1/3$ gives
	\begin{equation}\label{eq: lambda b x bar}
		\frac{Q_{\mathfrak{b}, \mathbf{F}}(u_{\mathrm{trial}})}{\|u_{\mathrm{trial}}\|_2^2} \leq \mathfrak{b}\lambda_{\alpha_{\overline{x}},\gamma_{\overline{x}}, a}+ C\mathfrak{b}^{\frac{5}{6}}.
	\end{equation}
\end{proof}
\subsection{Proof of Theorem~\ref{thm: asympt semicl}} The lower bound of $\lambda(\mathfrak b)$  is established in Proposition~\ref{pro: lower bound bottom}.  For the upper bound, introducing the $\limsup$ (as $\mathfrak b\rightarrow+\infty$) in Equation~\eqref{eq: ub lambda b x bar} of Proposition~\ref{pro: upper bound lambda b} gives 
	\[
	\limsup_{\mathfrak{b}\rightarrow +\infty}\frac{\lambda(\mathfrak{b})}{\mathfrak{b}} \leq \lambda_{\alpha_{\overline{x}}, \gamma_{\overline{x}}, a}.
	\]
	Taking then the infimum over $\overline{x}\in D$ and noticing that $\inf_{x\in \Gamma} \lambda_{\alpha_x,\gamma_x,a}=\inf_{x\in D} \lambda_{\alpha_x,\gamma_x,a}$ yield the desired upper bound.

\subsection{Proof of Theorem \ref{thm: localization}} The proof of Theorem \ref{thm: localization} follows now by standard arguments. 
\begin{proof}
	We use the common proof for Agmon-type estimates (see e.g. \cite[Theorem 9.4.1]{fournais2010spectral}). Let $R_0 > 1$, $\eta>0$ to be chosen later and let
	\begin{equation}
		g(x) := \eta \max(\mathrm{dist}(x, D), R_0 \mathfrak{b}^{-\frac{1}{2}}), \qquad x\in\Omega.
	\end{equation}	
	Notice that 
\begin{equation}\label{eq:g}
	\supp(\nabla g)\subset \{\mathrm{dist}(x,D) \geq R_0\mathfrak{b}^{-\frac{1}{2}}\}\ \mbox{and} \ |\nabla g| \leq C\eta. 
	\end{equation}
	By an integration by parts, the ground state  $\psi$  of the operator $\mathcal{P}_{\mathfrak{b}, \mathbf{F}}$ satisfies
	\begin{equation}\label{eq: agmon int parts}
		\lambda(\mathfrak{b}) \| e^{\sqrt{b} g}\psi\|^2 =  \mathrm{Re} \langle \mathcal{P}_{\mathfrak{b}, \mathbf{F}}\psi, e^{2\sqrt{\mathfrak{b}}g}\psi\rangle = Q_{\mathfrak{b}, \mathbf{F}}(e^{\sqrt{\mathfrak{b}} g}\psi) - \mathfrak{b}\| |\nabla g | e^{\sqrt{\mathfrak{b}} g}\psi\|^2.
	\end{equation} 
	Moreover, under Assumption~\ref{asu} and from Proposition \ref{pro: lower bound quadratic form} we get 
	\begin{eqnarray}\label{eq: lw quadratic form}
		Q_{\mathfrak{b}, \mathbf{F}}(e^{\sqrt{b}g}\psi) &\geq& \int_{\Omega}dx\, \big(U_{\mathfrak{b}}(x)  - CR_0^{-2}\mathfrak{b}\big)| e^{\sqrt{b}g}\psi|^2
		\\
		&\geq& \int_{\mathrm{dist}(x, D) \geq R_0 \mathfrak{b}^{-1/2}}dx\,  \big (|a|\Theta_0\mathfrak{b} - CR_0^4 \mathfrak{b}^{\frac{1}{2}} - CR_0^{-2}\mathfrak{b} \big) |e^{\sqrt{b} g}\psi|^2\nonumber \\
		&& + \int_{\mathrm{dist}(x, D) < R_0 \mathfrak{b}^{-1/2}}dx\,  \big (\lambda_D\mathfrak{b} - CR_0^4 \mathfrak{b}^{\frac{1}{2}} - CR_0^{-2}\mathfrak{b} \big) |e^{\sqrt{b} g}\psi|^2,\nonumber 
	\end{eqnarray}
	where $\lambda_{D} := \inf_{x\in D}\lambda_{\alpha_x, \gamma_x, a}=\inf_{x\in \Gamma}\lambda_{\alpha_x, \gamma_x, a}$. By  Assumption~\ref{asu}, we have  that $ \lambda_D<|a|\Theta_0 $. 
	
 Now, by Theorem~\ref{thm: asympt semicl} (more precisely see Proposition~\ref{pro: upper bound lambda b}),  we have
	\begin{equation}\label{eq: ub agmon}
		\lambda(\mathfrak{b}) \leq \mathfrak b\lambda_D+ o(\mathfrak b).
	\end{equation}
	Inserting \eqref{eq: lw quadratic form} and \eqref{eq: ub agmon} in \eqref{eq: agmon int parts} implies
	\begin{multline}\label{eq: agmon 2}
		\int_{\mathrm{dist}(x, D) \geq R_0 \mathfrak{b}^{-1/2}}dx\,  \big (|a|\Theta_0 - CR_0^4 \mathfrak{b}^{-\frac{1}{2}} - CR_0^{-2}\big) |e^{\sqrt{b} g}\psi|^2 
		\\
		+ \int_{\mathrm{dist}(x, D) < R_0 \mathfrak{b}^{-1/2}}dx\,  \big (\lambda_D - CR_0^4 \mathfrak{b}^{-\frac{1}{2}} - CR_0^{-2}\big) |e^{\sqrt{b} g}\psi|^2
		\\
		\leq (\lambda_{D} + o(1)) \|e^{\sqrt{b}g}\psi\|^2 + \| |\nabla g| e^{\sqrt{b}g}\psi\|^2.
	\end{multline}
	Finally, using the properties of $g$ in~\eqref{eq:g}, we write
	\[
	\| |\nabla g| e^{\sqrt{b} g}\psi \|^2 \leq \eta^2 \|e^{\sqrt{\mathfrak{b}}g}\psi\|^2,
	\]
	and insert this equation in \eqref{eq: agmon 2}. This yields, for $0<\eta < \sqrt{|a|\Theta_0 - \lambda_{D}}$, the existence of $R_0>1$ and  (a sufficiently large) $\mathfrak{b}_0$ such that 
	\begin{equation}
		\int_{\Omega}dx\, e^{2{  \eta}\sqrt{\mathfrak{b}} \mathrm{dist}(x, D)}|\psi|^2 \leq C(R_0,\eta) \|\psi\|^2. 
	\end{equation} 
	The estimate for the gradient term easily follows from the inequality above, using that $\psi$ is a ground state of $\mathcal{P}_{\mathfrak{b}, \mathbf{F}}$.
\end{proof}

\noindent\textbf{Acknowledgments.} The authors are thankful to  A. Kachmar for his valuable feedback on Section 5.

\appendix

\end{document}